\documentclass[a4paper,11pt]{article}

\usepackage[margin=3cm,footskip=1cm]{geometry}  
\usepackage[T1]{fontenc}   
\usepackage{xr}
    
\usepackage{amsmath,amssymb,amsthm,bm,mathtools,xspace,booktabs,url}
\usepackage{graphicx,etoolbox}
\usepackage[shortlabels]{enumitem}
\usepackage{mathrsfs}
\usepackage{diagbox}
\usepackage{makecell}
\usepackage{appendix}

\usepackage{lmodern} 
\usepackage{dsfont}
\usepackage{microtype}
\usepackage{syntonly}
\usepackage{float}
\usepackage{xcolor}
\usepackage{amsfonts}
\usepackage{color}
\newcommand{\can}[1]{{\color{magenta} #1 }}
\usepackage{etoc}
\usepackage{multirow} 
\usepackage[citecolor=blue,colorlinks=true,linkcolor=blue]{hyperref}

\usepackage[above,below,section]{placeins} 

\setlist{
listparindent=\parindent,
parsep=0pt,
}
      
\usepackage[comma,round]{natbib}
    
\AtBeginEnvironment{thebibliography}{%
\small
\setlength\itemsep{0.75em plus 0.5em minus 0.75em}%
}
    
\usepackage{caption}
\usepackage{authblk}

\makeatletter
\newcommand\CorrespondingAuthor[1]{%
\def\@makefnmark{}%
\footnotetext{Corresponding author: #1}%
}
\makeatother

\makeatletter
\renewenvironment{abstract}{%
\small%
\providecommand\keywords{%
\par\medskip\noindent\textit{Keywords:}\xspace}%
\begin{center}%
\bfseries \abstractname\vspace{-.5em}\vspace{\z@}%
\end{center}%
\quote%
}{\endquote}
\makeatother

\newcommand{\de}{\mathrm{d}}
\newcommand{\M}{{\mathcal M}}

\newcommand{\x}{{\mathbf x}}
\newcommand{\y}{{\mathbf y}}
\newcommand{\z}{{\mathbf z}}
\newcommand{\N}{\mathbb N}
\newcommand{\R}{\mathbb R}
\newcommand{\B}{{\mathcal B}}
\newcommand{\X}{{\mathcal X}}
\newcommand{\1}{{\mathbf 1}}

\def\P{\mathbb P}
\def\E{\mathbb E}
\DeclareMathOperator{\e}{e}
\DeclareMathOperator{\Var}{Var}
\DeclareMathOperator{\Cov}{Cov}
\DeclareMathOperator*{\argmin}{arg \, min} 

\DeclareMathOperator{\supp}{supp}
\DeclareMathOperator{\med}{med}

\newcommand{\beann}{\begin{eqnarray*}}
\newcommand{\eeann}{\end{eqnarray*}}

\newtheorem{definition}{Definition}[section]
\newtheorem{lemma}{Lemma}[section]
\newtheorem{thm}{Theorem}
\newtheorem{remark}{Remark}[section]
\newtheorem{cor}{Corollary}

\graphicspath{{./Figures/}}

\numberwithin{equation}{section}

\begin{document}
    
\title{Statistical learning and cross-validation 
for point processes}

\author[1]{
Ottmar Cronie\footnote{Corresponding author}
}
\affil[1]{
Biostatistics, School of Public Health and Community Medicine,
University of Gothenburg;
}
\affil[1]{
Dept.~of Mathematics and Mathematical Statistics, 
Umeå University, Umeå, Sweden. 

ottmar.cronie@gu.se
}
\author[2]{
Mehdi Moradi
} 
\affil[2]{
Dept.~of Statistics, Computer Science, and Mathematics, 
Public University of Navarre, Pamplona;
}
\affil[2]{
Inst.~of Advanced Materials and Mathematics, 
Public University of Navarre, Pamplona, Spain.  

mehdi.moradi@unavarra.es
}
\author[3]{
Christophe A.N. Biscio
}
\affil[3]{
Dept.~of Mathematical Sciences, Aalborg University,
Aalborg,
Denmark.

christophe@math.aau.dk
}

\date{}

\maketitle


\begin{abstract}
This paper presents the first general (supervised) statistical learning framework for point processes in general spaces. Our approach is based on the combination of two new concepts, which we define in the paper: i) {\em bivariate innovations}, which are measures of discrepancy/prediction-accuracy between two point processes, and ii) {\em point process cross-validation} (CV), which we here define through point process thinning. 
The general idea is to carry out the fitting by predicting CV-generated validation sets using the corresponding training sets; the prediction error, which we minimise, is measured by means of bivariate innovations. 
Having established various theoretical properties of our bivariate innovations, we study in detail the case where the CV procedure is obtained through independent thinning and we apply our statistical learning methodology to three typical spatial statistical settings, namely parametric intensity estimation, non-parametric intensity estimation and Papangelou conditional intensity fitting. Aside from deriving theoretical properties related to these cases, in each of them we numerically show that our statistical learning approach outperforms the state of the art in terms of mean (integrated) squared error.



\keywords 
Bivariate innovation,  Cross-Validation, 
Generalised random samples, 
Kernel intensity estimation,
Loss function, 
Monte-Carlo cross-validation,
Multinomial $k$-fold cross-validation, 
Papangelou conditional intensity function, 
Prediction, 
Subsampling,
Test function, 
Thinning
\end{abstract}

\section{Introduction}

As emphasised by e.g.~\citet{breiman2001statistical} and carefully outlined by e.g.~\citet{vapnik2013nature}, in contrast to classical statistical inference, the philosophy behind statistical learning is that a model's fit should be judged by its ability to predict new/hold-out data. With the ``big data'' age's explosion in data acquisition and complexity, which has required increasingly flexible modelling strategies \citep{hastie2009elements}, the statistical learning paradigm has become increasingly natural to the statistics community. 
Classical statistical learning, which dates back to the 1960s, is rooted in the setting where the data under study constitute a random sample of a fixed size, i.e.~a collection of $k\geq1$ independent and identically distributed (iid) random variables from some (unknown) probability distribution  \citep{hastie2009elements,james2013introduction,vapnik2013nature}. 
More specifically, following e.g.~the setting in \citet{vapnik2013nature}, which deals with what is commonly known as supervised learning, one assumes that there is an unknown distribution $P(\cdot)$ which governs the joint distribution of iid pairs $(y_i,z_i)$, $i=1,\ldots,k$, where $y_1,\ldots,y_k$ are referred to as training data and $z_1,\ldots,z_k$ as validation data. Loosely speaking, the aim is to  ``predict the validation data, using the training data, in an optimal way'': from a given class of functions, which use the training data as input, the objective is to find the candidate which predicts the validation data as well as possible, in the sense of minimised $P(\cdot)$-expected loss, given a suitable loss function. 

As the size and the complexity of the data increase, 
the risk that the underlying independence assumption is violated is increasing and, consequently, statistics in the context of dependent sampling is becoming increasingly relevant  \citep{christensen2019advanced}. In addition, for many datasets we do not typically know the total sample size a priori, i.e.~we do not deal with a controlled trial, and this suggests that the total number of observations should be treated as random. Typical examples of such data structures are given by spatially and/or temporally sampled data; a specific example is the dataset in \citet{bayisa2020large}, which consists of the space-time locations of roughly 500 000 Swedish ambulance calls. 
Such a {\em generalised random sample}, which may be described as a
collection $X=\{x_1,\ldots,x_N\}\subseteq S$, $0\leq N\leq\infty$, of
random points/variables in some general space $S$, where i) $N$ may be random and ii) the points
may be dependent, in essence constitutes what is referred to as a {\em point process} \citep{VanLieshoutBook,DVJ1,MW04,BenesRataj,DVJ2,CSKWM13,Diggle14Book,BRT15,last2017lectures,baccelli2020}; note the somewhat unusual convention that small
letters are used for random variables. Conditioning on $N=n$, 
when the members of $X$ are iid, we obtain the classical notion of a random
sample, which in the point process literature is referred to as a
Binomial point process \citep{VanLieshoutBook,MW04}. 
It is customary
to refer to each point as an {\em event}, since point processes often
are used to describe spatial and/or temporal locations of data which represent events. Typical examples include  astronomical objects  \citep{babu1996spatial,kerscher2000statistical}, 
climatic events \citep{toreticoncurrent}, 
crimes \citep{ABN12,MFJ18,som2021}, 
disease cases \citep{meyer2012space,Diggle14Book}, 
earthquakes \citep{ogata1998space,marsan2008extending,iftimi2019second}, 
farms \citep{biscio2019farming}, 
queuing events \citep{bremaud1981point,baccelli2013elements}, 
traffic accidents \citep{rakshit2019fast,moradispacetime,moradidirectional}, 
and 
trees (forestry)  \citep{stoyan2000recent,cronie2013spatiotemporal}. 
The term point process is rather unfortunate, we argue, seeing
as a point process in itself does not represent a stochastic process in the usual sense, but rather a random sample generalised by the two properties above. Some authors have suggested that a more suited name would be
{\em random point field} \citep{CSKWM13}. 
The historical reason for the name point process stems from the fact that when 
$S=\R$, or
$S=[0,\infty)$, we may view $S$ as a time axis and, consequently, we obtain a {\em temporal point process}
$X=\{t_i\}_{i=1}^N$, which in turn yields the cumulative stochastic
process $X(t)=\#\{t_i\in X:t_i\leq t\}\in\{0,1,\ldots\}$, $t\in S$  \citep{DVJ2}.

It is key to note that, in contrast to the classical setting, observed point process realisations, so-called {\em point patterns}, mostly do not come in the form of repeated samples. Instead, we observe only one realisation $\x=\{x_1,\ldots,x_n\}$ of the underlying point process $X$. This makes the statistical analysis more challenging since we essentially try to extract a large amount of information from only one realisation, where we cannot impose the fixed sample size iid assumption and, consequently, we cannot reduce the problem to one of repeated sampling. 

To the best of our knowledge, this paper introduces the first general statistical learning theory for point processes. This offers a new look on how statistics for point processes can be tackled and it rigorously brings the field into the contemporary era of statistical learning. The setting here is that we observe only one realisation $\x$ of a point process $X$ in some (complete separable metric) space $S$; our theory works equally well under repeated sampling of $X$. 

Our starting point is a family of integral formulas/theorems/relations, commonly referred to as the Campbell, Campbell-Mecke and Georgii-Nguyen-Zessin formulas. These all relate expectations of sums over the points of a point process to integrals with respect to various distributional characteristics of the point process, e.g.~factorial moment measures/densities and Papangelou conditional intensity functions, which characterise many point process models \citep{DVJ2,last2017lectures}. By altering these relations to represent the setting where one point process is ``predicted'' by another point process, we define what we call {\em bivariate innovations}, which essentially can be though of as measures of discrepancy between two point processes; the name {\em innovation} is motivated by the fact that in one particular setting, our bivariate innovations reduce to the ``classical'' innovations of  \citet{baddeley2005residual,baddeley2008properties}. 
Having established that our innovations generalise much of the previously developed statistical theory for point processes (see
e.g.~\citet{moller2017some,cronie2018bandwidth,coeurjolly2019understanding} and the references
therein), we proceed to study different distributional properties of our bivariate innovations. In particular, we arrive at conditions under which they may be exploited to carry out supervised learning. 

Cross-validation (CV) is ubiquitous in modern statistics and data science, and there is a vast literature dealing with CV in the classical iid setting; see e.g.~\citet{arlot2010survey} and the references therein. 
To make our statistical learning framework work (in the single sample setting), we combine the bivariate innovations framework with CV. This allows us to carry out the fitting by minimising the prediction error generated by predicting CV-generated validation sets from CV-generated training sets, by means of our bivariate innovations. 
Due to the underlying (potential)  dependence in a point process, it is not immediately clear how CV should be properly defined for point processes. We here present the first general and theoretically justified treatment of CV for point processes. Inspired by our previous work on point process subsampling \citep{Moradi2019}, we argue that CV in the point process setting should be defined by assuming that the validation sets $\x_i^V$, $i=1,\ldots,k$, are given by $k\geq1$ independently generated thinnings \citep[Section 5.1]{CSKWM13} of the observed point process, and that the training sets are given by $\x_i^T=\x\setminus\x_i^V$, $i=1,\ldots,k$. Formally, we allow $\x_i^V$ to be any kind of (in)dependent thinning, but 
due to many appealing properties of independent thinnings, where one independently retains each point $x\in\x$ with probability $p(x)$, according to some function $p(u)\in(0,1)$, $u\in S$, we mainly argue that CV for point processes should be based on independent thinning. 

Having studied in detail how our general framework can be applied in the settings of i) parametric (factorial) moment estimation, so-called product density/intensity function estimation, ii) Papangelou conditional intensity estimation and iii) non-parametric product density/intensity function estimation, we proceed by looking at specific instances of these. Most notably, through simulation studies we show that in each of these instances, our statistical learning framework outperforms the state of the art.

The paper is structured as follows.
In Section \ref{sec:Background} we give an overview of different point process characteristics, e.g.~product densities and Papangelou conditional intensities, as well as a few common point process models. In addition, we derive some basic, but for our purposes important, results on independent thinning. In Section \ref{s:generalised weigthed innovations} we first present some basics on parameter estimation and, more importantly, we define our bivariate innovations. We then proceed by deriving some distributional properties of our innovations. Section \ref{s:Learning} starts by defining and studying point process cross-validation and then proceeds to laying down our statistical learning framework. At the end of Section \ref{s:Learning} we study in detail the case where we consider independent thinning-based cross-validation. Section \ref{s:Applications} looks closer at a few applications of our approach and Section \ref{s:Discussion} contains a discussion.

\section{Point process preliminaries}\label{sec:Background}

We begin by providing an overview of, for our purposes, relevant point process theory.

\subsection{General notation and outcome spaces}
\label{s:Spaces}
Throughout, $S$ will be a general (complete separable metric) space
with distance metric $d(\cdot,\cdot)$ and $d$-induced Borel sets $\B$;
all subsets under consideration will be members of $\B$ so we
reserve the notation "$\subseteq$" for members of $\B$. A closed ball of
radius $r\geq0$ around a point $u\in S$ will be denoted by
$b(u,r)=\{v\in S:d(u,v)\leq r\}$.
We further endow $S$ with a notion of size in the form of a (locally-
and $\sigma$-finite Borel) reference measure $A\mapsto|A|$,
$A\subseteq S$, where the corresponding integration will be denoted by
$\int\de u$. Throughout, we will often consider functions without
explicitly stating that they are measurable/integrable.


The following examples, which have been illustrated in 
Figure \ref{f:ExamplesSpace}, are commonly encountered in the spatial statistical literature:
\begin{itemize}
\item The $d$-dimensional Euclidean space $S=\R^d$, $d\geq1$, with
  the Euclidean metric $d(u,v)=\|u-v\|_2$, $u,v\in\R^d$, where $\|u\|_2=\|(u_1,\ldots,u_d)\|_2=(\sum_{j=1}^d|u_j|^2)^{1/2}$ is the Euclidean norm, 
  and Lebesgue measure $|\cdot|$.

\item The $\alpha$-radius sphere
  $S=\alpha\mathbb{S}^{d-1}=\{\alpha x\in\R^d:\|x\|_2=1\}$, $\alpha>0$,
  in dimension $d\geq2$; here $d(\cdot,\cdot)$ is the great circle
  distance and $|\cdot|$ is the spherical surface measure \citep{robeson2014point,moller2016functional,lawrence2016point}.

\item A linear network $S=L=\bigcup_{i=1}^{k}l_i$, consisting of
  $k\in\{1,2,\ldots\}$ line segments
  $l_i=[u_i,v_i]=\{tu_i + (1-t)v_i:0\leq t\leq 1\}\subseteq\R^2$; here $L$ is assumed to be graph-connected. Often $d(u,v)$ is
  the shortest-path distance, giving the shortest length of any path
  in $L$ which joins $u,v\in L$ \citep{OS12,ABN12} or, more generally, a so-called 
  regular distance metric
  \citep{rakshit2017second,cronie2020inhomogeneous}. The measure
  $|\cdot|$ here corresponds to integration with respect to arc length
  (1-dimensional Hausdorff measure in $\R^2$).
  
  \item A spatio-temporal domain, where e.g.~one of the spaces $S$ above represents its spatial component, can be defined by $S\times T$, where $T$ is given by either a compact interval in $\R$, $[0,\infty)$ or $\R$  \citep{DVJ2,Diggle14Book,gonzalez2016spatio}. Here $d(\cdot,\cdot)$ is e.g.~given by the maximum of the spatial and the temporal distances and $|\cdot|$ is given by the product measure generated by the reference measure on $S$ and Lebesgue measure on $T$ \citep{Cronie2015}.
\end{itemize}

\begin{figure}[!htpb]
\centering
\includegraphics[scale=0.25]{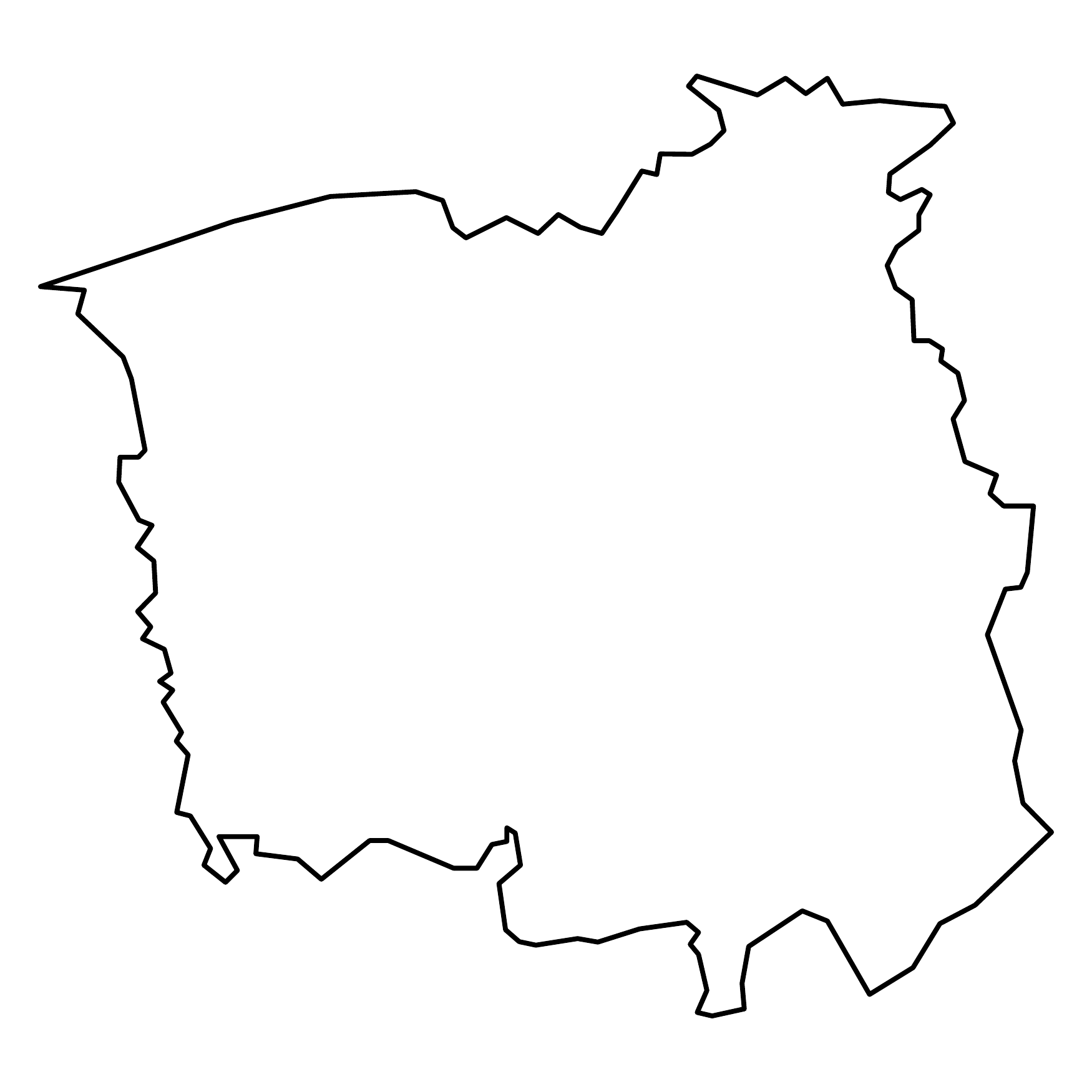}
\includegraphics[scale=0.25]{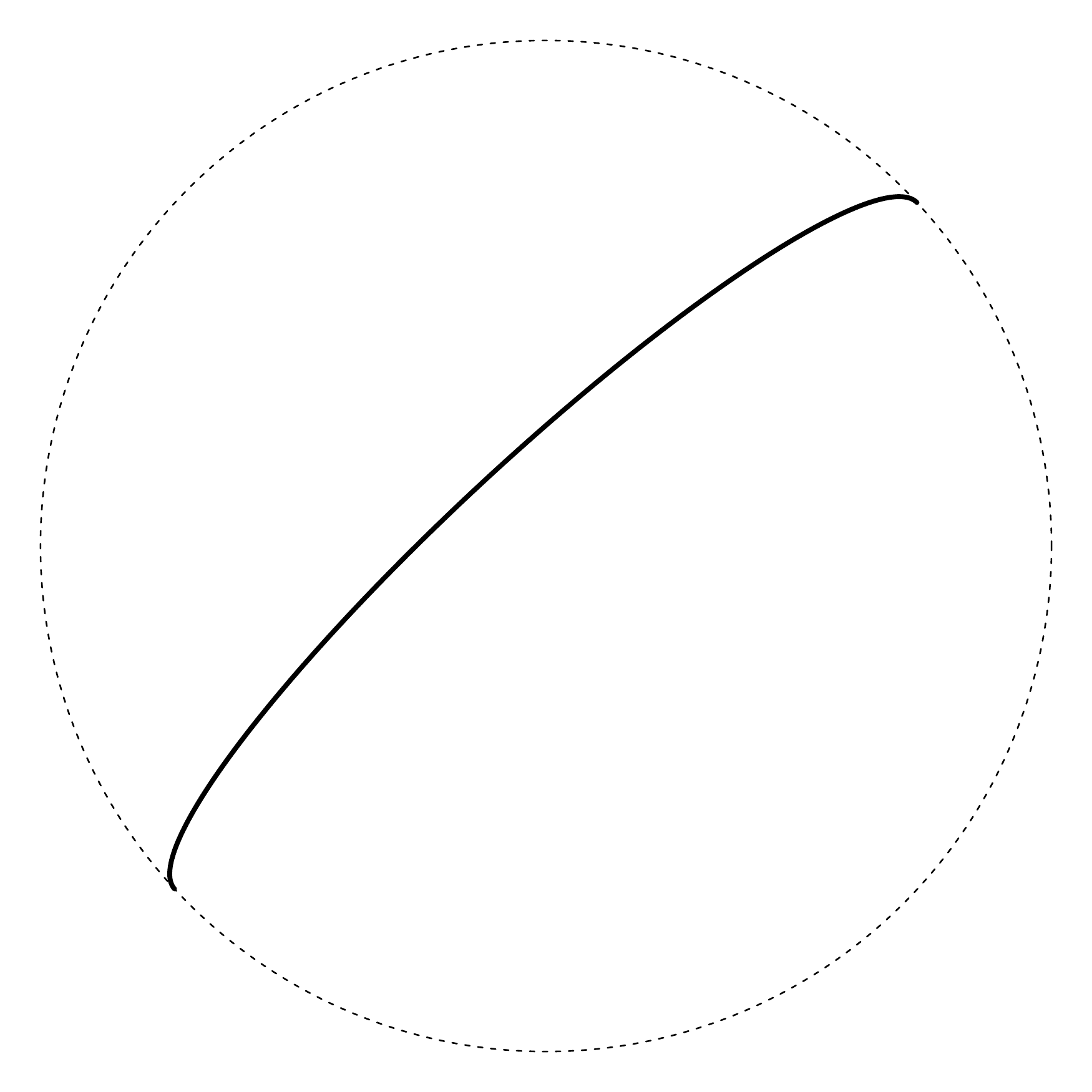}
\includegraphics[scale=0.25]{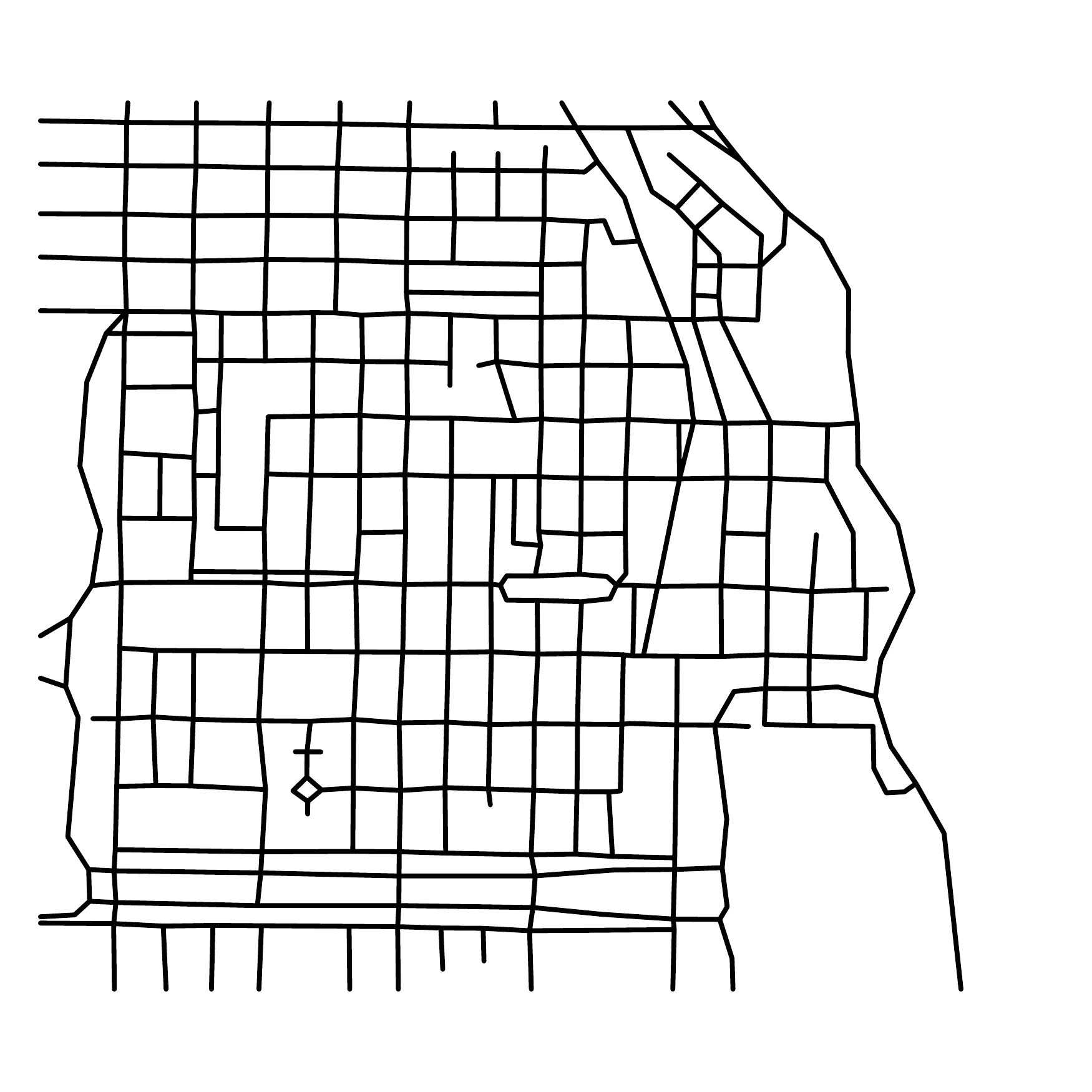}
\caption{Left: A subset of $\R^2$.  Middle: A sphere. Right: A linear network.}
\label{f:ExamplesSpace}
\end{figure}

Throughout, $\#$ will denote cardinality and $\1\{A\}$ will denote the
indicator function for $A$, which is $1$ when $A$ is satisfied and $0$
otherwise. Moreover, $(\Omega,\mathcal{F},\P)$ will be some suitable underlying
abstract probability space which generates the random elements under
consideration; e.g., a random variable/vector
$Y=(Y_1,\ldots,Y_n)\in S^n$, $n\geq1$, is formally a measurable
mapping from $(\Omega,\mathcal{F},\P)$ to $(S^n,\B^n)$.

\subsection{Point processes}

Formally, a (simple) point process $X=\{x_i\}_{i=1}^N$,
$0\leq N<\infty$, in $S$ may be defined as a random element/variable
in the measurable space $(\X,\mathcal{N})$,  where 
$\X=\X_S$ is the collection of point configurations
$\x=\{x_1,\ldots,x_n\}\subseteq S$, $0\leq n\leq\infty$, which are
locally/boundedly finite,
i.e.~$\#(\x\cap A)=\sum_{i=1}^n\1\{x_i\in A\}<\infty$ for any bounded
$A\subseteq S$ \citep{MW04}. Note that $\#(\x\cap S)=n=0$ yields that $\x=\emptyset$ and
if $S$ is bounded then we necessarily have that $n<\infty$. Moreover, $\mathcal{N}$
is the $\sigma$-algebra generated by
the cardinality mappings $\x\mapsto \#(\x\cap A)\in\{0,1,\ldots,\infty\}$,
$A\subseteq S$, $\x\in \X$, 
and it coincides with the Borel $\sigma$-algebra generated by a (modified) Prohorov
metric on $\X$ \citep{DVJ1,DVJ2}; the connection is made
by identifying the random set $X$ with its (discrete) random measure
representation $X(A)=\#(X\cap A)$, $A\subseteq S$. Note that the term
`simple' above refers to the fact that $\#(X\cap\{u\})\in\{0,1\}$
with probability one/almost surely (a.s.) for any $u\in S$; this
follows from the construction of a point process as a random subset of
$S$. Moreover, when $N<\infty$ a.s., which e.g.~is the case if $S$ is bounded, then we say
that $X$ is a {\em finite} point process. 
For general treatments, 
see
e.g.~\citet{VanLieshoutBook,DVJ1,MW04,BenesRataj,DVJ2,CSKWM13,Diggle14Book,BRT15,kallenberg2017random,last2017lectures,baccelli2020}.

As is often the case with distributions of random elements on
abstract spaces, we may here specify the distribution
$P(E)=P_X(E)=\P(X\in E)$, $E\in\mathcal{N}$, of a point process $X$ by
means of its finite dimensional distributions \citep{VanLieshoutBook},
i.e.~the distributions of all vectors $(X(A_1),\ldots,X(A_n))$,
$A_i\subseteq S$, $i=1\ldots,n\geq1$. The sets $E\in\mathcal{N}$ may
be thought of as point process features; we may e.g.~have
$E=\{\x\in\X:\#(\x\cap A)=0\}$ for some $A\subseteq S$. 
When $X$ is finite, its family of Janossy measures, which governs its finite dimensional distributions, sometimes admits densities $\{j_n\}_{n\geq0}$, where  $j_n(u_1,\ldots,u_n)\de u_1\cdots\de u_n$ gives the probability of $X$ having all its points in infinitesimal neighbourhoods of $u_1,\ldots,u_n\in S^n$ \citep{DVJ1}.

A dataset $\x=\{x_1,\ldots,x_n\}\in\X$, which we model/analyse
under the assumption that it has been generated by a point process, is
commonly referred to as a {\em point pattern} and the members of $\x$ and $X$ are often
called {\em events}.



\subsection{Point process characteristics}

Most of the relevant point process characteristics considered in the
literature can be obtained through (combinations of) expectations of
the kind
\begin{align}
\label{eq:sum}
  &\E\left[\mathop{\sum\nolimits\sp{\ne}}_{x_1,\ldots,x_n\in X}
    h(x_1,\ldots,x_n,X\setminus\{x_1,\ldots,x_n\})\right]
    ,
    \qquad
    n\geq1,
\end{align}
where $h:S^n\times\X\to\R$ is 
permutation invariant in its first $n$ arguments; unless $h$ is
non-negative (and possibly infinite), $h$ is assumed to be integrable.  The notation
$\sum^{\neq}$ is used to indicate that the summation is taken over
distinct $n$-tuples and it is noteworthy that \eqref{eq:sum}
corresponds to the expectation of a sum over the elements of the point
process
\begin{align}
\label{e:nTupleProcess}
X_{\neq}^n=\{(x_1,\ldots,x_n)\in X^n:x_i\neq x_j\text{ if }i\neq
j\}\subseteq S^n,
\end{align}
which consists of distinct $n$-tuples of elements of $X$,
i.e.~$\sum_{x_1,\ldots,x_n\in X}^{\neq}=\sum_{(x_1,\ldots,x_n)\in
  X_{\neq}^n}$ \citep{SchneiderWeil}. Note e.g.~that for any
$x_1,x_2\in X$, the points $(x_1,x_2),(x_2,x_1)$ are distinct elements
of $X_{\neq}^2$. Below we show that by considering different
subclasses of functions $h$, we obtain different integral identities
for \eqref{eq:sum}, which are based on different point process
characteristics; restricting such a subclass to non-negative $h$, the
corresponding identity becomes defining for the associated point
processes characteristic. Throughout, when we discuss such characteristics we implicitly assume that they exist.

\subsubsection{Factorial moment characteristics}
\label{s:ProductDensities}

The subclass of functions $h$ in \eqref{eq:sum} which are constant over
$\X$, i.e.~of the form $h(x_1,\ldots,x_n)$, defines the $n$th order
{\em product density/factorial moment density/intensity function} $\rho^{(n)}$ of $X$ through the
{\em Campbell formula/theorem}  \citep[Section 9.5]{DVJ2}, which states that \eqref{eq:sum} equals
\begin{align}
\label{eq:Campbell}
  \int_{S^n}
  h(u_1,\ldots,u_n) \rho^{(n)}(u_1,\ldots,u_n)\de u_1\cdots\de u_n.
\end{align}
Formally, $\rho^{(n)}$ is the Radon-Nikodym
derivative of the $n$th-order  factorial moment measure 
$$
(A_1\times\cdots\times A_n)\mapsto
\E\left[\mathop{\sum\nolimits\sp{\ne}}_{x_1,\ldots,x_n\in X}
    \prod_{i=1}^n\1\{u_i\in A_i\}\right]
, 
\quad A_1,\ldots,A_n\subseteq S,
$$
of $X$, with respect to the product measure $|\cdot|^n$. 
Heuristically, since $X$ is
simple, for infinitesimal neighbourhoods $A_i=du_i$, $\de u_i=|du_i|$, of the points $u_i\in S$, $i=1,\ldots,n$, we
obtain that 
$\P(X(du_1)=1,\ldots,X(du_n)=1)=\E[X(du_1)\cdots
X(du_n)]=\rho^{(n)}(u_1,\ldots,u_n)\de u_1\cdots\de u_n$.

The particular case $n=1$ gives us {\em the intensity function} function $\rho(u)=\rho^{(1)}(u)$, $u\in S$, of $X$, which thus satisfies
$$
\E[X(A)]=\int_A\rho(u)\de u, \qquad A\subseteq S.
$$
From the heuristics above, we see that the intensity function governs the univariate marginal distributional properties of $X$. 
Whenever $\rho(\cdot)\equiv\rho>0$ is constant we
say that $X$ is {\em homogeneous} and otherwise we say that $X$ is
{\em inhomogeneous}. 
Finally, it may be noted that $\rho^{(n)}(\cdot)$ is the
intensity function of the point process $X_{\neq}^n$ in \eqref{e:nTupleProcess}  \citep{SchneiderWeil}.

Clearly, we may have that $\rho^{(n)}(u_1,\ldots,u_n)$ is large
without points of $X$ around $u_1,\ldots,u_n\in S$ being dependent; e.g., 
under independence among the points we have that 
$\rho^{(n)}(u_1,\ldots,u_n)=\rho(u_1)\cdots\rho(u_n)$. Hence, in order to study $n$-point
dependencies among the points of $X$, it is more natural to consider
its $n$th {\em correlation function} (which does not actually
represent correlation in the usual sense):
\begin{equation}
\label{CorrelationFunctions}
g^{(n)}(u_1,\ldots,u_n)
=
\frac{\rho^{(n)}(u_1,\ldots,u_n)}
{\rho(u_1)\cdots\rho(u_n)}, 
\quad u_1,\ldots,u_n \in S.
\end{equation}
Note that $g^{(1)}(\cdot)=\rho(\cdot)/\rho(\cdot)=1$ and under
independence we obtain that $g^{(n)}(\cdot)=1$ for any $n\geq1$. Hence, when
$g^{(n)}(u_1,\ldots,u_n)>1$ we speak of {\em attraction/clustering/aggregation} between points of $X$ located around $u_1,\ldots,u_n\in S$ and when instead $g^{(n)}(u_1,\ldots,u_n)<1$, we speak of
{\em inhibition/regularity/repulsion}. The heuristic idea here is that
we measure joint probability effects after we have scaled away the
individual marginal ones. The archetype model for lack of interaction
is a Poisson process; see Section \ref{s:Poisson} for details.

In the case of $S=\R^d$, when $g^{(n)}(u_1,\ldots,u_n)$ only depends
on the separation vectors $u_i-u_j$, $i\neq j$, and the intensity
function is positive/bounded away from 0, the point process $X$ is
called $n$th-order intensity reweighted stationary
\citep{van11,Cronie2016,ghorbani2020functional}. When $n=2$ this is
referred to as second-order intensity reweighted stationarity (SOIRS)
\citep{InhomK2000}, and we write
$g^{(2)}(u_1,u_2)=g^{(2)}(u_1-u_2)$, whereas when this holds for any $n\geq1$, 
we say that $X$ is intensity reweighted moment stationary
\citep{van11}.
When $X$ is homogeneous, $n$th-order intensity reweighted stationarity
turns into the notion of $n$th-order (moment) stationarity, which, in
turn, is implied by stationarity (provided that all $\rho^{(n)}$,
$n\geq1$, exist); stationarity for a point process in $\R^d$ is
defined as $X$ having the same distribution as $\{y+x:x\in X\}$ for
any $y\in\R^d$. For non-Euclidean spaces $S$, things become more delicate,
however \citep{kallenberg2017random,rakshit2017second,cronie2020inhomogeneous}.

\subsubsection{Conditioning}
Turning to the general case, where $h$ is not necessarily constant over $\X$, 
we obtain that 
\eqref{eq:sum} equals 
\begin{align}
\label{e:CampbellMeasure}
\int_{S^n}
\int_{\X}
h(u_1,\ldots,u_n,\y)
\mathcal{C}_n^!(d((u_1,\ldots,u_n),\y)),
\end{align}
where $\mathcal{C}_n^!(A\times E)$, $A\subseteq S^n$,
$E\in\mathcal{N}$, is the $n$th-order reduced Campbell measure \citep[Section 13]{DVJ2}. Under
assumptions of absolute continuity with respect to the $n$th-order
factorial moment measure and the distribution of $X$, by e.g.~\citet[Sections 13 \& 15]{DVJ2} we obtain
that \eqref{e:CampbellMeasure}, and thereby \eqref{eq:sum}, equal 
\begin{align}
\label{eq:Campbell-Mecke}
&\int_{S^n}
\E_{u_1,\ldots,u_n}^!
\left[
h(u_1,\ldots,u_n,X)
\right]
\rho^{(n)}(u_1,\ldots,u_n)
\de u_1\cdots\de u_n,
\\
\label{eq:GNZ}
&\int_{S^n}
\E
\left[
h(u_1,\ldots,u_n,X)
\lambda^{(n)}(u_1,\ldots,u_n;X)
\right]
\de u_1\cdots\de u_n,
\end{align}
respectively,  
where the former relation is referred to as the {\em reduced Campbell-Mecke
formula/theorem} and the latter as the {\em Georgii-Nguyen-Zessin
(GNZ) formula/theorem.} 
The family $P_{u_1,\ldots,u_n}^!(E)$,
$u_1,\ldots,u_n\in S$, $E\in\mathcal{N}$, of regular conditional
probability distributions governing the expectations in \eqref{eq:Campbell-Mecke} are the so-called $n$th-order {\em reduced
  Palm distributions}, whereas $\lambda^{(n)}(u_1,\ldots,u_n;\x)$, $u_1,\ldots,u_n\in S$, $\x\in\X$, is referred to as the $n$th-order {\em Papangelou conditional intensity function} of $X$.



It follows that
$P_{u_1,\ldots,u_n}^!(\cdot)$ corresponds to a point process
$X_{u_1,\ldots,u_n}^!$, which may be interpreted as $X$ conditioned on having
points at the locations $u_1,\ldots,u_n$, which are removed upon
realisation. 
As one would hereby intuitively guess, the $k$th-order
product density of $X_{u_1,\ldots,u_n}^!$ is given by
\begin{align}\label{e:ReducedPalmIntensity}
&\rho^{!(k)}(v_1,\ldots,v_k|u_1,\ldots,u_n)
=
\frac{\rho^{(k+n)}(v_1,\ldots,v_k,u_1,\ldots,u_n)}
{\rho^{(n)}(u_1,\ldots,u_n)}
,\quad k,n\geq1,
\end{align}
when the $n$th-order product density of $X$ satisfies
$\rho^{(n)}(u_1,\ldots,u_n)>0$, otherwise it is $0$.

Point processes for which the relationship between \eqref{eq:sum} and \eqref{eq:GNZ} is well-defined are commonly referred to
as {\em Gibbs processes} \citep{MollerPalm}. Moreover, we heuristically have
that
\begin{align*}
  &\lambda^{(n)}(u_1,\ldots,u_n;\x)\de u_1\cdots\de u_n =\\
  =&\P(X(du_1)=1,\ldots,X(du_n)=1|X\cap S\setminus(du_1\cup\cdots\cup du_n)=\x\cap S\setminus(du_1\cup\cdots\cup du_n))
\end{align*}
for infinitesimal neighbourhoods $du_i\ni u_i\in S$, $|du_i|=\de u_i$, $i=1,\ldots,n\in S$. In words, this corresponds to 
the
probability of finding points of $X$ in infinitesimal regions around
$u_1,\ldots,u_n$, conditionally on $X$ agreeing with $\x$ outside these
infinitesimal regions. 
Moreover, recalling the Campbell formula and letting $h$ in \eqref{eq:GNZ} be of the
form $h(u_1,\ldots,u_n)$, we
immediately obtain that
$\rho^{(n)}(u_1,\ldots,u_n)=\E[\lambda^{(n)}(u_1,\ldots,u_n;X)]$. 
The first-order Papangelou conditional intensity, 
$\lambda(\cdot)=\lambda^{(1)}(\cdot)$, is commonly referred to as {\em the
  Papangelou conditional intensity} and it 
  is the central building block 
  here since \citep{MollerPalm}
\begin{align}
\label{e:nthPapangelou}
  \lambda^{(n)}(u_1,\ldots,u_n;\x)
  =
  \lambda(u_1;\x)
  \lambda(u_2;\x\cup\{u_1\})
  \cdots
  \lambda(u_n;\x\cup\{u_1,\ldots,u_{n-1}\})
,
\end{align}
as one would intuitively suggest based on the above infinitesimal
conditional probability interpretation. 
In particular, if $X$ is finite, with Janossy densities $j_n$, $n\geq1$,
then the heuristics are formalised by \citep[Section 15.5]{DVJ2}
\begin{align}
\label{e:PapangelouFinitePP}
\lambda(u,\x)
=
\left\{
\begin{array}{rl}
j_{n+1}(\x\cup\{u\})/j_{n}(\x),
& u\notin\x=\{x_1,\ldots,x_n\}\in\X, \\
j_{n}(\x)/j_{n-1}(\x\setminus\{u\}),
& u\in\x=\{x_1,\ldots,x_n\}\in\X,
\end{array}
\right.
\quad 
u\in S.
\end{align}
It should be noted
that this definition more commonly is given in terms of densities with
respect to Poisson process distributions \citep[Theorem
1.6]{VanLieshoutBook}. 

Finally, the connection between these
two notions of conditioning (interior vs exterior) is established
through the relation 
\citep{MollerPalm}
\begin{align}
\label{e:RelationPalmPapangelou}
P_{u_1,\ldots,u_n}^!(E)=\rho^{(n)}(u_1,\ldots,u_n)^{-1}\int_E
\lambda^{(n)}(u_1,\ldots,u_n;\x)P(d\x), \quad E\in\mathcal{N},
\end{align}
where $P(\cdot)$ is the distribution of $X$ on $(\X,\mathcal{N})$.

\subsection{Common point process models}
\label{s:Models}
Below we provide an overview of a few point process model families which are commonly encountered in the literature.

\subsubsection{Fixed size samples}
\label{s:RandomSamples}

The most basic example is the case where we condition on the total point count $X(S)=N\geq1$. This implies that the point process $X=\{x_1,\ldots,x_N\}\subseteq S^N$ is equivalent to an $N$-dimensional random vector where $x_1,\ldots,x_N\in S$ have the same marginal distribution. One may e.g.~think of a multivariate Gaussian random vector where $S=\R$, $\E[x_i]=\mu$, $\Var(x_i)=\sigma^2$ and  $\Cov(x_i,x_j)=\widetilde\sigma$, $i\neq j$, for any $i,j\in\{1,\ldots,N\}$. When we additionally assume that these are independent, so that $X$ is a random sample (iid), we speak of a Binomial point process \citep{VanLieshoutBook,MW04}. 

Assuming that $x_1,\ldots,x_N$ have a joint density $f_N(u_1,\ldots,u_N)$, $u_1,\ldots,u_N\in S$, with marginal densities $f_n(\cdot)$, $1\leq n<N$,
\[
f_n(u_1,\ldots,u_n)
=
\int\cdots\int
f_N(u_1,\ldots,u_n,v_1,\ldots,v_{N-n})
\de v_1\cdots\de v_{N-n},
\quad u_1,\ldots,u_n\in S
\]
the associated $N$th-order Janossy density satisfies  $j_N(u_1,\ldots,u_N)=N!f_N(u_1,\ldots,u_N)$, $u_1,\ldots,u_N\in S$ \citep[Section 5.3]{DVJ1}; the Janossy densities $j_n(\cdot)$ of orders $n\neq N$ are 0. By \citet[Lemma 5.4.III]{DVJ1}, it now follows that the corresponding product densities satisfy 
\begin{align*}
\rho^{(n)}(u_1,\ldots,u_n)
=&
\frac{1}{(N-n)!}
\int_{S^{N-n}}
j_N(u_1,\ldots,u_n,v_1,\ldots,v_{N-n})\de v_1\cdots\de v_{N-n}
\\
=&
\frac{N!}{(N-n)!}
f_n(u_1,\ldots,u_n)
,
\quad 
u_1,\ldots,u_N\in S, 
\quad 
1\leq n\leq N,
\end{align*}
whereby $\rho(u_1)\cdots\rho(u_n)=\rho^{(n)}(u_1,\ldots,u_n)=\frac{N!}{(N-n)!}
f_1(u_1)\cdots f_1(u_n)$ for a Binomial point process. 

Recalling \eqref{e:PapangelouFinitePP}, 
we may here consider the Papangelou conditional intensity given by
\begin{align*}
\lambda(u,\x)
=&\frac{n!f_n(u,x_1,\ldots,x_{n-1})}{(n-1)!f_{n-1}(x_1,\ldots,x_{n-1})}
=
n f_1(u|x_1,\ldots,x_{n-1}),
\qquad
\lambda(u,\emptyset)
=f_1(u),
\end{align*}
where $u\in S,
\quad 
\x=\{x_1,\ldots,x_{n-1}\}\subseteq S^{n-1}$, $2\leq n\leq N-1$, and  for a Binomial point process we have $f_1(u|x_1,\ldots,x_{n-1})=f_1(u)$.

\subsubsection{Poisson processes}
\label{s:Poisson}

A first step towards generalising classical (iid) random samples is to keep
the independence of the points but allow for a random total point
count. Such point process fall into the category of completely random measures \citep[Section 10.1]{DVJ2} and the archetype here, which is also the most prominent family of point
process models, is the family of Poisson processes. If a function
$\rho(u)\geq0$, $u\in S$, governs a well-defined point process $X$ in
the sense that i) $X(A)\sim Poi(\int_A\rho(u)\de u)$ for any
$A\subseteq S$ and ii) for any disjoint $A_1,\ldots,A_n\subseteq S$,
$n\geq1$, the discrete random variables $X(A_1),\ldots,X(A_n)\geq0$
are independent, then $X$ is a Poisson process in $S$ with intensity
function $\rho(u)$, $u\in S$. Consequently,
\begin{align}\label{e:Poisson}
  \lambda^{(n)}(u_1,\ldots,u_n;\x)
  =&\rho
    ^{(n)}(u_1,\ldots,u_n)=\rho(x_1)\cdots\rho(x_n),
  \quad
  g^{(n)}(u_1,\ldots,u_n)
  =
  1,
\\
\label{e:PoissonLogLikelihood}
  j_n(u_1,\ldots,u_n)
  =&
  \left(\int_S\rho(u)\de u\right)^n\exp\left\{-\int_S\rho(u)\de u\right\}
  \prod_{i=1}^n\rho(u_i)
  ,
\end{align}
for any $n\geq1$, and $j_0 = \exp\{-\int_S\rho(u)\de u\}$; recall that the Janossy densities $j_n$, $n\geq1$, refer to the finite case. 
Note further that a Binomial point
process may be defined as a Poisson process conditioned on $X(S)=N$. 
Moreover, 
the family of reduced Palm distributions satisfies 
$P_{u_1,\ldots,u_n}^!(\cdot)=P(\cdot)$, $n\geq1$, i.e.~reduced Palm conditioning
has no effect. Also, an independent thinning (see Section \ref{s:MPPs}) of a Poisson process is again a Poisson process.




\subsubsection{Cox processes}
\label{s:Cox}
A Cox process is essentially the mixed model version of a Poisson
process. More specifically, consider a stochastic/random process/field
$\Lambda(u)$, $u\in S$, which a.s.~is non-negative and satisfies
$\int_A\Lambda(u)\de u<\infty$ for bounded $A\subseteq S$. If, conditional on $\Lambda(\cdot)=\rho(\cdot)$, $X$ is a Poisson process
with intensity $\rho(\cdot)$, then $X$ is said to be a Cox process with random intensity function/driving random field $\Lambda$. It follows that
\begin{align*}
  \rho^{(n)}(u_1,\ldots,u_n)
  =&\E[\Lambda(u_1)\cdots\Lambda(u_n)],
  \quad
  g^{(n)}(u_1,\ldots,u_n)
  =
  \frac{\E[\Lambda(u_1)\cdots\Lambda(u_n)]}
     {\E[\Lambda(u_1)]\cdots\E[\Lambda(u_n)]},
  \\
  \lambda(u;\x)
  =&
    \frac{\E[\exp\{-\int_S\Lambda(v)\de v\}\prod_{x\in\x}\Lambda(x)\Lambda(u)]}
    {\E[\exp\{-\int_S\Lambda(v)\de v\}\prod_{x\in\x}\Lambda(x)]}
                  .
\end{align*}
By Jensen's inequality, $g^{(n)}(u_1,\ldots,u_n)\geq1$ (with equality if $\Lambda$ is
completely independent/noise) for any $n$, whereby a Cox process is clustering.

A particularly tractable and well
studied family of Cox processes is the family of log-Gaussian Cox processes
\citep{moller1998lgcp}. Here the random intensity function is given by
$\Lambda(u)=\exp\{Z(u)\}$, $u\in S$, for a Gaussian random
field $Z$ on $S$. The product densities of such a model can
readily be derived using moment generating functions of Gaussian
random vectors: 
\(
\rho^{(n)}(u_1,\ldots,u_n)
  =\exp\{\sum_{i=1}^n\E[Z(u_i)] + \sum_{i=1}^n\sum_{j=1}^n\Cov(Z(u_i),Z(u_j))/2\}. 
\)
In particular, when $S=\R^d$, if the covariance function
$C(u_1,u_2)=\Cov(Z(u_i),Z(u_j))$ is translation invariant in the sense
that $C(u_1,u_2)=C(u_1-u_2)$, $u_1,u_2\in S$, then $X$ is intensity
reweighted moment stationary and we note that the
variance $\Var(Z(u))=\sigma^2$, $u\in S$, is constant.

\subsubsection{Exponential family Gibbs processes}
\label{s:ExpFamily}

Many common model families fit into the framework of exponential family Gibbs models \citep{VanLieshoutBook,MW04,BRT15}. Such models
have Papangelou conditional intensities of the form
\[
  \lambda_{\theta}(u;\x) = \beta(u;\x) \e^{\theta^{\top}D(u;\x)} =
  \beta(u;\x)
\e^{\theta^{\top}(T(\x\cup\{u\})-T(\x))},
\quad u\in S, 
\x\in\X, 
\theta\in\Theta\subseteq\R^l,
\]
where $\beta:S\times\X\to(0,\infty)$, which we assume to be bounded, 
and $T:\X\to[0,\infty)$ is a canonical sufficient statistic.
Considering a fixed function $\beta(u;\cdot)=\beta(u)>0$, $u\in S$,
specific examples include Poisson processes, area-interaction processes and Strauss processes. 
The quantities above are required to be such that the Papangelou
conditional intensity in question is locally integrable. We stress
that product densities for exponential family models are generally
not available in closed form. Note further that by letting
$\beta(\cdot)\equiv\beta>0$, one obtains a homogeneous version of the process in question. 
Moreover, the function $\beta(\cdot)$ may also
itself belong to some parametric family of functions, in which case
the model in question would be reparametrised so that the
parameter vector $\theta$ would include the parameters of $\beta(\cdot)$.

    

  Letting
    $\theta=\log\eta$ and
    $T(\x)=T_R(\x)=\frac{1}{2}\sum_{x_1,x_2\in\x}^{\neq}\1\{d(x_1,x_2)\leq
    R\}$, we obtain an inhomogeneous Strauss process with Papangelou
    conditional intensity
    $$
    \lambda_{\theta}(u;\x)= \beta(u)\exp\{
    \log\eta(T_R(\x\cup\{u\})-T_R(\x))\} =\beta(u)\eta^{D_R(u;\x)},
    \quad u\in S,\x\in\X,
    $$
    where $R>0$ is called the interaction radius, $\eta\in[0,1]$ is called the
    interaction parameter and
    $D_R(u;\x)=\sum_{x\in\x\setminus\{u\}}\1\{d(u,x)\leq
    R\}=\#\{x\in\x\setminus\{u\}:d(u,x)\leq R\}$, $u\in S$, $\x\in\X$, where we use the convention that $0^0=1$.
    Strauss processes form a basic family of inhibiting point process,
    where $\eta=1$ corresponds to a Poisson process, $\eta\in(0,1)$ corresponds to the family of Strauss soft-core models and $\eta=0$ corresponds to the classical
    hard-core model, which does not allow points to be within distance $R$ from one
    another. Note that the hard-core model's Papangelou conditional intensity may be expressed as
    \begin{align}
    \label{e:HardCorePapangelou}
    \lambda_{\theta}(u;\x)=\beta(u)\1\left\{u\notin\bigcup_{x\in\x}b(x,R)\right\};
    \end{align}
    recall that $b(x,R)$ denotes a closed $R$-ball around $x$. 

    

\subsubsection{Determinantal point processes}
\label{s:DPP}

Determinantal point processes (DPPs) are models which give rise to inhibition among their points. 
They were introduced to statistics in
their current form by \citet{macchi75dpp} and have since been applied in numerous spatial statistical settings  \citep{lavancier2015dpp} 
as well as in machine learning \citep{kulesza2012dpp}. 
Known examples of DPPs when $S=\R^d$ include 
the Poisson and Ginibre point processes.

A point process $X$ is a DPP on $S$ if there exists a complex-valued function $C: S \times S \rightarrow \mathbb{C}$, 
called the kernel of $X$, such 
that for all $n\geq 1$, the product densities are given by 
\begin{equation}\label{eq:def dpp}
    \rho^{(n)}(u_1,\ldots,u_n)
    = 
    \det[C](u_1,\ldots,u_n),
    \qquad u_1,\ldots,u_n \in S,
\end{equation}
where $\det$ denotes the determinant and $[C](u_1,\ldots,u_n)$ denotes the matrix with entry $C(u_i,u_j)$
on the $i$-th row and the $j$-th column, $1\leq i,j \leq n$. 
In order to ensure the existence of a DPP 
with product densities given by~\eqref{eq:def dpp}, 
several conditions need to be enforced on $C$. 
Consider the integral operator defined for all square integrable function $f: S \rightarrow \mathbb{C}$ by
\begin{equation}\label{eq:int operator dpp}
    f \rightarrow \int_S C(x,y) f(y) \de y.
\end{equation}
According to~\citet{hough2009zeros}, if $C$ is hermitian, locally
square integrable and all the eigenvalues of the integral
operator~\eqref{eq:int operator dpp} are in $[0,1]$, then $C$ defines
one and only one DPP. 

DPPs have various appealing properties for statistical applications. 
For instance, it follows directly from~\eqref{eq:def dpp} that most
moment-based summary statistics 
have closed form expressions, which are governed by 
$C$. 
Moreover, according to \citet{shirai2003fermion}, if $X$ is a DPP with kernel $C$ such that all its eigenvalues are in $[0,1)$, then for any $n\geq 1$ and $u_1,\ldots u_n \in S$, $X^!_{u_1,\ldots, u_n}$ is also a DPP with kernel $C^{!}_{u_1,\ldots,u_n}(x,y)$ given by
\begin{equation*}
    C^{!}_{u_1,\ldots,u_n}(x,y) = 
     \frac{\det M}{\det [C](u_1,\ldots,u_n)},
\end{equation*}
where $M$ is the matrix with entry on the $i$-th row and $j$-th column given by $C(u_i,u_j)$ for $2\leq i,j\leq (n+1)$, $C(u_i,y)$ for $j=1$ and $2\leq i \leq (n+1)$,  $C(x,u_j)$ for $i=1$ and $2\leq j \leq (n+1)$, and $C(x,y)$ for $i=j=1$. 
Hence, since by~\eqref{eq:def dpp}  the (factorial) moments of a DPP are known in closed form, with respect to its kernel, also the moments of $X^!_{u_1,\ldots, u_n}$
as well as the $n$th-order Papangelou conditional intensity of $X$ are available in closed form. 
In particular, for any 
$\x= \{x_1,\ldots,x_n\}\in\X$ and $u\in S$,
\begin{equation*}
    \lambda(u;\x) = 
    \frac{\det [C](u,x_1,\ldots,x_n)}{\det [C](x_1,\ldots,x_n)}.
\end{equation*}
Finally, according to~\citet{lavancier2015dpp}, 
a simple and convenient choice of kernel when $S=\R^d$ 
is a real-valued continuous covariance function verifying $C(u,v)=C_0(\|u-v\|)$, $u,v\in\R^d$, where the Fourier transform of $C_0$ belongs to $[0,1]$.
This highlights the fact that dealing with DPPs in non-Euclidean spaces $S$ can be quite challenging \citep[cf.][]{anderes2020isotropic}. 

\subsection{Marked point processes and thinning}
\label{s:MPPs}

We next look closer at marked point processes, which are particular instances of point processes on product spaces. These are usually of interest when each event carries some additional piece of
information, which is not directly connected to $S$, e.g.~a label, some
quantitative measurement or more abstract objects such as functions
and sets \citep{CSKWM13,ghorbani2020functional}. Our main interest in using marking here is related to the fact that so-called thinnings of point processes may be obtained through a particular kind of marking; our cross-validation approaches presented in Section \ref{s:CVgeneral} are based on thinning.

Given two general spaces $S$ and $\M$, with associated reference measures $|A|$, $A\subseteq S$ and $\nu_{\M}(B)$,
$B\subseteq\M$, consider the product space $\breve{S}=S\times\M$. The space $\breve{S}$ is itself a general space which we endow with the product reference measure $\breve{\nu}(A\times B)=|A|\nu_{\M}(B)$,
$A\times B\subseteq S\times\M$. 
Moreover, denote
the space of locally finite point configurations
$\breve{\x}=\{(x_1,m_1),\ldots,(x_n,m_n)\}$ in $\breve{S}$ by
$\breve{\X}$ and the corresponding point configuration $\sigma$-algebra by $\breve{\mathcal{N}}$. A point process
$\breve{X}=\{(x_i,m_i)\}_{i=1}^N$ on $\breve{S}$, i.e.~a random
element in $(\breve{\X},\breve{\mathcal{N}})$, is called a {\em marked point process (MPP)} with {\em marks} $m_i\in\M$, $i=1,\ldots,N$, if the projection $X=\{x_i\}_{i=1}^N$, which is a random element in
$(X,\mathcal{N})$, exists as a well-defined point process on $S$.  In keeping with \citet{DVJ2}, we call $X$ the {\em ground process}, $S$ the {\em ground space} and $\M$ the {\em mark space}.

\begin{remark}
  By letting the mark space be given by $\M=(0,1)$, we may treat each
  mark as an "arrival time" and transition to so-called {\em sequential point processes}, which have the same construction as point processes but with the difference that the elements of $\X$ instead are ordered tuples $\x=(x_1,\ldots,x_n)$
  \citep{LieshoutSequential}.
\end{remark}
 
The product densities 
of $\breve{X}$ satisfy
\citep{Cronie2016}
\begin{align}
\label{e:MarkedProdDens}
\breve{\rho}^{(n)}((u_1,m_1),\ldots,(u_n,m_n))
=&
f_{\M}^{(n)}(m_1,\ldots,m_n|u_1,\ldots,u_n)
\rho_X^{(n)}(u_1,\ldots,u_n)
,
\end{align}
where $(u_i,m_i)\in S\times\M$, $i=1,\ldots,n$, $\rho_X^{(n)}(\cdot)$ is
the $n$th-order product density of the ground process,
$f_{\M}^{(n)}(\cdot|u_1,\ldots,u_n)$, $u_1,\ldots,u_n\in S$, is a
family of density functions on $\M^n$ and $g_X^{(n)}(\cdot)$ is the
$n$th-order correlation function of $X$; we write
$f_{\M}(\cdot|\cdot)=f_{\M}^{(1)}(\cdot|\cdot)$.  This highlights that
the joint distributions of the marks are specified conditionally on
the ground process.

A particular kind of marking which will be of interest to us is
(location-dependent) independent marking, where the marks are
independent conditional on the ground process. Here
$f_{\M}^{(n)}(m_1,\ldots,m_n|v_1,\ldots,v_n)=f_{\M}(m_1|v_1)\cdots
f_{\M}(m_n|v_n)$ and if $X$ is randomly labelled, i.e.~if the marks
are iid conditional on the ground process then
$f_{\M}^{(n)}(m_1,\ldots,m_n|v_1,\ldots,v_n)=f_{\M}(m_1)\cdots
f_{\M}(m_n)$ for a common density $f_{\M}(m)$, $m\in\M$. Note that for
a stationary MPP the latter is the density of what is
commonly referred to as {\em the} mark distribution
\citep{CSKWM13,baccelli2020}. A particular instance of an independently marked point process is a Poisson process $\breve{X}=\{(x_i,m_i)\}_{i=1}^N$ on $S\times\M$, where $X=\{x_i\}_{i=1}^N$ is well-defined (a Poisson process); note e.g.~that a homogeneous Poisson process on $S\times\M=\R^d\times\R^{d'}$, $d,d'\geq1$, with intensity $\breve{\rho}(\cdot)\equiv\rho>0$ is not an MPP with ground space $\R^d$ since the local finiteness of $X=\{x_i\}_{i=1}^N$ is violated \citep{VanLieshoutBook}. 

\subsubsection{Thinning}

Heuristically, a thinning $Z\subseteq X$ of a point process $X$ is
generated by applying some rule/mechanism to $X$ which either retains
or deletes each $x\in X$ \citep{CSKWM13}. We next provide a definition of thinning
which is based on bivariate marking of a point process.

\begin{definition}
\label{def:ZuYmpp}
  Given a point process $X=\{x_i\}_{i=1}^N\subseteq S$, a {\em
  thinning} $Z$ of $X$ with retention probability 
  $p: S \times \X \rightarrow [0,1]$ 
  may be defined as the marginal process
  $Z=\{x:(x,m)\in\breve{X}\cap S\times\{1\}\}$ of a bivariate marking of $X$,
  \begin{align}
  \label{e:ZuYmpp}
   \breve{X}
  =
  \{(x_i,m_i)\}_{i=1}^N
  \subseteq S\times \M,
  \quad \M=\{0,1\},
\end{align}
where 
$m_i=m(x_i)\in \M$, $i=1,\ldots,N$, 
for some (possibly random) marking function $m(\cdot)$, which governs the retention probability function.

When $\breve{X}$ is independently marked, i.e.~the retention
probability $p(u)$, $u\in S$, does not depend on $X$, we say that
$Z$ is an {\em independent thinning}. If, in addition,
$p(\cdot)\equiv p\in[0,1]$, so that the marks are independent and
Bernoulli distributed with parameter $p$, we say that $Z$ is a
{\em $p$-thinning}.
\end{definition}

Here the reference measure on $\M$ is given by $\nu_{\M}(\{i\})=1$,
$i\in\M=\{0,1\}$, i.e.~the counting measure on $\M$. 
Note that the
retention probability function governs how points of $X$ are assigned
to $Z$. Moreover, the complement/remainder
$Y=X\setminus Z$ is a thinning with retention probability function
$1-p(u,\x)$. 

Under independent thinning, which is one of the cornerstones of this paper, 
we independently retain each point $x\in X$ according to $p(u)$, $u\in S$. 
This should be contrasted to the case where the
retention (marking) of a point depends on whether other specific
(e.g.~nearby) points have been retained. 
Independent thinnings are particularly tractable and below we provide a central result on important distributional properties of independent thinnings; its proof can be found in Section \ref{s:Proofs}. We will use this result to establish certain properties of the bivariate innovations presented in Section~\ref{s:generalised weigthed innovations}

\begin{thm}\label{lemma:Thinning}
Let $Z$ be a p-thinning of a point process $X$ on $S$, with retention probability $p(u)\in(0,1)$, $u\in S$, let $Y=X\setminus Z$, let $\breve{X}$ be the associated MPP representation in \eqref{e:ZuYmpp} and consider some $n\geq1$. 

For any non-negative or integrable
$h:S^n\times\X\to\R$, 
\begin{align}
    \label{eq:thinning Z vs Y}
    &\E\left[\mathop{\sum\nolimits\sp{\ne}}_{x_1,\ldots,x_n\in Z}
    h(x_1,\ldots,x_n, Y)
    \prod_{i=1}^n(1-p(x_i))
    \right] 
    =\nonumber
    \\
    =&
    \E\left[\mathop{\sum\nolimits\sp{\ne}}_{x_1,\ldots,x_n\in Y}
    h(x_1,\ldots,x_n, Y\setminus \{x_1,\ldots,x_n\})
    \prod_{i=1}^n p(x_i)
    \right].
  \end{align}
Moreover, provided that they exist, the $n$th-order Papangelou conditional intensity
and the $n$th-order product density of $Z$ a.e.~satisfy
\begin{align}
\lambda_Z^{(n)}(u_1,\ldots,u_n,Z)
\stackrel{a.s.}{=}&
p(u_1)\cdots p(u_n)\E[\lambda_X^{(n)}(u_1,\ldots,u_n;X)|Z],
\nonumber
\\
\label{e:ThinningProdDens}
\rho_Z^{(n)}(u_1,\ldots,u_n)
=&
p(u_1)\cdots p(u_n)
\rho_X^{(n)}(u_1,\ldots,u_n),
\end{align}
where $\lambda_X^{(n)}$ and $\rho_X^{(n)}$ are the $n$th-order Papangelou conditional intensity and product density of $X$. 
In addition, when the $n$th-order Papangelou conditional intensities of $\breve{X}$ and $Y$ exist, they satisfy 
    \begin{equation*}
        \E[\breve\lambda^{(n)}((u_1,1),\ldots,(u_n,1);\breve{X})|Y]
        =
        \frac{\prod_{i=1}^n p(u_i)}{\prod_{i=1}^n (1-p(u_i))}
        \lambda_Y^{(n)}(u_1,\ldots,u_n;Y)
    \end{equation*}
    for almost all $u_1,\ldots, u_n \in S$. 
    In particular, for a $p$-thinning with retention probability $p\in(0,1)$ we set $p(\cdot)\equiv p$ in all the expressions above. 

\end{thm}

\begin{remark}
\label{rem:ClassicalConditionalIntesnsity}
Certain marked temporal point processes, e.g.~Hawkes processes, are often specified through a "classical" conditional intensity function $\lambda^{\dagger}(t,m)$, $t\geq0$, $m\in\M$, which heuristically gives the probability of finding an event with mark $m$ in an infinitesimal future time interval $(t,t+dt)$, given the history of events in $[0,t]$. Given a suitable filtration, such a conditional intensity may be defined through an integral relationship of the same form as the GNZ formula, but with $h$ being a predictable stochastic process and $\lambda$ replaced by the predictable stochastic process $\lambda^{\dagger}$ \citep{DVJ1,DVJ2,flint2019functional}. Under integrability conditions on  $\lambda^{\dagger}$, the existence of a Papangelou conditional intensity implies the existence of $\lambda^{\dagger}$ \citep[Lemma 2.7]{flint2019functional}. Consequently, we expect results similar to the ones provided above to hold for classical conditional intensities. 
\end{remark}

\section{Innovations}
\label{s:generalised weigthed innovations}

In this section we define what we refer to as bivariate innovations. These are tools which may be used  e.g.~to predict properties of one point process from another point process. Together with our cross-validation approaches in Section \ref{s:CVgeneral}, they are one of the building blocks of our statistical learning framework.

\subsection{General parametrised estimator families}
Assume that we observe/sample a point pattern
$\x=\{x_1,\ldots,x_n\}$ within some (bounded) study region/domain
$W\subseteq S$, $|W|>0$, which we assume has been generated by some
unknown point process $X$ (restricted to $W$). Broadly speaking, 
statistics here concerns itself with extracting
information about the underlying point process $X$ through $\x$. 

As we shall see, most of the statistical settings which we will
encounter here deal with estimation/modelling of some particular characteristic of
$X$. It turns out that the associated estimators can be characterised
by a {\em general parametrised estimator family}
$\Xi_{\Theta}^n=\{\xi_{\theta}^n: \theta\in\Theta\}$,
$\Theta\subseteq\R^l$, $l\geq1$, $n\geq1$, where
\begin{align}
\label{e:GeneralEstimator}
\xi_{\theta}^n(u_1,\ldots,u_n;\y),
\quad u_1,\ldots,u_n\in S, 
\quad \y\in\X, 
\quad \theta\in\Theta,
\end{align}
are real-valued and, for any $\y$, $\xi_{\theta}^n(\cdot;\y)$ is
either non-negative or integrable. When each
$\xi_{\theta}^n$ is constant over $\y\in\X$, i.e.~it does not depend on $\y$,
we set
\begin{align}
\label{e:GeneralEstimatorNoPP}
\xi_{\theta}^n(\cdot;\y)
\equiv
\xi_{\theta}^n(\cdot)
\text{ for any }
\y\in\X,
\quad
\theta\in\Theta.
\end{align}
The underlying assumption here is that
$\xi_{\theta_0}^n\in\Xi_{\Theta}^n$, for some unknown
$\theta_0\in\Theta$, represents the true characteristic of interest of the
underlying point process $X$. In other words, we assume that there is
no model miss-specification and our aim is to estimate $\theta_0$ based on $\x$, using $\Xi_{\Theta}^n$. To carry out the estimation of
$\theta_0$, one would need to find a minimiser of some loss function,
$\mathcal{L}(\theta)$, $\theta\in\Theta$, which also depends on the
data, i.e.~the point pattern $\x$. Such a minimiser
$\widehat\theta=\widehat\theta_W(\x)\in\Theta$ is referred to as an
estimate and the random version $\widehat\theta_W(X)$ is referred to
as an estimator.

\subsection{Bivariate and univariate innovations}

We next introduce our bivariate innovations which, as previously mentioned, are one of the main components of the statistical learning framework presented in this paper. 
The essential idea behind them is that they predict properties of one point process from another point process. 
Moreover, as is indicated in Section \ref{s:UnivariateInnovations}, these tools may be used to summarise many existing statistical estimation approaches, both parametric and non-parametric ones. 




\begin{definition}\label{def:Innovations}
  Consider two general parametrised estimator families, 
  $\Xi_{\Theta}^n=\{\xi_{\theta}^n: \theta\in\Theta\}$ and
  $\mathcal{H}_{\Theta}=\{h_{\theta}:\theta\in\Theta\}$, 
  both of either the form \eqref{e:GeneralEstimator} or \eqref{e:GeneralEstimatorNoPP}. 
  We refer to the members of $\mathcal{H}_{\Theta}$ as
  {\em test functions}.

  The associated families of {\em ($n$th-order $\mathcal{H}_{\Theta}$-weighted) bivariate innovations}
  $\{\mathcal{I}_{\xi_{\theta}^n}^{h_{\theta}}(\cdot;\z,\y):
  \y,\z\in\X\}_{\theta\in\Theta}$ and {\em univariate innovations} $\{\mathcal{I}_{\xi_{\theta}^n}^{h_{\theta}}(\cdot;\y):
  \y\in\X\}_{\theta\in\Theta}$ are defined as the signed Borel 
  measures 
\begin{align}
\label{e:Innovations}
\mathcal{I}_{\xi_{\theta}^n}^{h_{\theta}}(A;\z,\y)
=&
\sum_{(x_1,\ldots,x_n)\in\z_{\neq}^n\cap A}
h_{\theta}(x_1,\ldots,x_n;\y\setminus\{x_1,\ldots,x_n\})
\\
&-
\int_{A}
h_{\theta}(u_1,\ldots,u_n;\y)
\xi_{\theta}^n(u_1,\ldots,u_n;\y)
\de u_1\cdots\de u_n,
&
\nonumber
\\
\mathcal{I}_{\xi_{\theta}^n}^{h_{\theta}}(A;\y)
=&
\mathcal{I}_{\xi_{\theta}^n}^{h_{\theta}}(A;\y,\y),
&
A\subseteq S^n,
\nonumber
\end{align}
where
$\z_{\neq}^n=\{(x_1,\ldots,x_n)\in\z^n:x_i\neq x_j\text{ if }i\neq
j\}$ and 
$h_{\theta}(x_1,\ldots,x_n;\y\setminus\{x_1,\ldots,x_n\})=h_{\theta}(x_1,\ldots,x_n;\y)$
if $\y\cap\z=\emptyset$. 
In particular, when $\Xi_{\Theta}^n$ and $\mathcal{H}_{\Theta}$ are of the form
\eqref{e:GeneralEstimatorNoPP} then 
$\mathcal{I}_{\xi_{\theta}^n}^{h_{\theta}}(\cdot;\z,\y)=\mathcal{I}_{\xi_{\theta}^n}^{h_{\theta}}(\cdot;\z)$ for any $\y,\z\in\X$.

\end{definition}

We first note that due to the assumed measurability of all involved quantities, each innovation has the following property: 
for fixed $\y,\z\in\X$ it is a (signed Borel) measure on $S^n$, and
for a fixed $A$ it is a measurable function of $\y\in\X$ and
$\z\in\X$. This is often referred to as being a kernel
\citep{kallenberg2017random}.

\begin{remark}
  At times one has to require that $\y\in\X$ in \eqref{e:Innovations}
  is contained in some (possibly) bounded $W\subseteq S$. This
  restriction can be included by replacing the test function by either
  $h_{\theta}(u_1,\ldots,u_n;W\cap\y\setminus\{u_1,\ldots,u_n\})$ or
  $\1\{\y\in
  W\}h_{\theta}(u_1,\ldots,u_n;\y\setminus\{u_1,\ldots,u_n\})$.
\end{remark}

Turning to the heuristics, for two point processes $Z$ and $Y$, the random
signed measure $\mathcal{I}_{\xi_{\theta}^n}^{h_{\theta}}(A;Z,Y)$ is
an empirical measure of how well $Y$ predicts distinct $n$-tuples of
$Z$ in $A$ via $\xi_{\theta}^n$ and $h_{\theta}$; the test function
$h_{\theta}$ weights the associated contributions of $n$-subsets of
distinct points, $\xi_{\theta}^n$ 
is intended to
describe the distributional properties of the superposition $Z\cup Y$
and $\mathcal{I}_{\xi_{\theta}^n}^{h_{\theta}}(A;Z,Y)$ estimates
how well the specific choice $\theta\in\Theta$ does in predicting $n$-tuples of $Z$ from $Y$.

The name {\em innovation} has been chosen to be in keeping with
\citet{baddeley2005residual,baddeley2008properties} and the related
estimating equation approaches considered in the literature (see
\citet{moller2017some,coeurjolly2019understanding} and the references
therein). 
\citet{baddeley2005residual,baddeley2008properties} used the term
innovation for the univariate innovation $\mathcal{I}_{\lambda_{\theta}}^{h_{\theta}}(A;X)$,
where $X$ is some point process and
$\xi_{\theta}^1(\cdot;\cdot)=\lambda_{\theta}(\cdot;\cdot)$ belongs to
a parametric family of Papangelou conditional intensity functions; here, it is natural to refer to such innovations as {\em classical} innovations. 
From the heuristics above,
conceptually we may view
$\mathcal{I}_{\xi_{\theta}^n}^{h_{\theta}}(W;X)$ as a
measure of how well we can 
predict $n$ distinct hold-out-points of $X$, by means of the remaining points of $X$. Moreover, instead of estimation, they considered point process residuals, $R_{\widehat\theta}(W)=\mathcal{I}_{\xi_{\widehat\theta}^n}^{h_{\widehat\theta}}(W;X)$, obtained by plugging a separately generated estimate $\widehat\theta\in\Theta$ into the innovation.

The next straightforward result, which is proved in Section
\ref{s:Proofs}, indicates that univariate innovations make sense as loss functions/estimating
equations. Also, given the setting of Lemma \ref{lemma:InnovationUnbiased} below, in the case of Papangelou conditional intensities, expressions
for the variance (when $n=1$) can be found in
\citet{baddeley2008properties} and \citet[Lemma 15.5.III]{DVJ2}, and covariance expressions can be found in \citet{coeurjolly2013fast}.

\begin{lemma}\label{lemma:InnovationUnbiased}
  If $X$ has $n$th-order Papangelou conditional intensity
  $\lambda_{\theta_0}^{(n)}\in\Xi_{\Theta}^n$, $n\geq1$, where
  $\Xi_{\Theta}^n$ is of the form \eqref{e:GeneralEstimator}, then
  $\E[\mathcal{I}_{\lambda_{\theta_0}}^{h}(W;X)]=0$ for any test
  function $h:S^n\times\X\to\R$. If $X$ has $n$th-order product
  density $\rho_{\theta_0}^{(n)}\in\Xi_{\Theta}^n$, $n\geq1$, where
  $\Xi_{\Theta}^n$ is of the form \eqref{e:GeneralEstimatorNoPP}, then
  $\E[\mathcal{I}_{\rho_{\theta_0}^{(n)}}^{h}(W;X)]=0$ for any test
  function satisfying $h(\cdot;\y)=h(\cdot)$ for any $\y\in\X$.
\end{lemma}


One approach to using univariate innovations as loss functions for the estimation of the underlying parameter $\theta_0$ is to find a minimiser
$\widehat\theta=\widehat\theta_W(\x)$ of
\[
\mathcal{L}(\theta;\x)
=
\mathcal{I}_{\xi_{\theta}^n}^{h_{\theta}}(W;\x)^2,
\quad \theta\in\Theta.
\]
A particularly interesting choice for $\mathcal{H}_{\Theta}$ is given by
$h_{\theta}(\cdot;\y)=f(\xi_{\theta}^n(\cdot;\y))$, for some suitable function $f:\R\to\R$. 
In Section \ref{s:UnivariateInnovations} 
we look closer at how 
univariate innovations 
summarise many (if not most) existing statistical inference approaches for point processes. In particular, we highlight (non-)parametric intensity estimation, Papangelou conditional intensity fitting and $K$-function based minimum contrast estimation. Consequently, we also indicate different test function choices; these have been summarised in Section \ref{s:TestFunctions}.
The main aim with this exposition is to highlight various contexts in which bivariate innovations may be used for prediction-based inference.

\subsection{Properties of bivariate innovations}


In Theorem
\ref{thm:InnovationsMean} below, which is proved in Section
\ref{s:Proofs}, we derive expressions for the expectation and the variance of 
a bivariate innovation, together with necessary and
sufficient condition to ensure that the expectation is null.



\begin{thm}\label{thm:InnovationsMean}
Given a point process $X$ in $S$, let $Z$ be an arbitrary thinning of $X$, $Y=X\setminus Z$, and 
    $\breve{X}$ the associated bivariate 
    point process representation in Definition~\ref{def:ZuYmpp}. 
    Consider further some fixed $n\geq1$, and let $\Xi_{\Theta}^n=\{\xi\}$ and
  $\mathcal{H}_{\Theta}=\{h\}$ consist of one element each.

    When $\xi,h:S^n\to\R$ are of the form \eqref{e:GeneralEstimatorNoPP}, the univariate innovation  $\mathcal{I}_{\xi}^{h}(\cdot;Z,Y)=\mathcal{I}_{\xi}^{h}(\cdot;Z)$ satisfies 
    \begin{align}
    \label{e:ExpectationInnovationProdDens}
    \E[\mathcal{I}_{\xi}^{h}(A;Z)]
    =&
     \int_A
     h(u_1,\ldots,u_n)
     \left(
     \rho_{Z}^{(n)}(u_1,\ldots,u_n)
     -
     \xi(u_1,\ldots,u_n)
     \right)
     \de u_1\cdots\de u_n
     ,
     \\
\Var(\mathcal{I}_{\xi}^{h}(A;Z))
=&
\sum_{j=0}^n
j!
\binom{n}{j}^2
\int_{S^{2n-j}}
h(u_1,\ldots,u_n)
h(u_1,\ldots,u_j,u_{n+1},\ldots,u_{2n-j})
\times
\nonumber
\\
&\times
\1\{(u_1,\ldots,u_n)\in A\}
\1\{(u_1,\ldots,u_j,u_{n+1},\ldots,u_{2n-j})\in A\}
\times
\nonumber
\\
&\times
\rho_Z^{(2n-j)}(u_1,\ldots,u_{2n-j})
\de u_1\cdots\de u_{2n-j}
\nonumber
\\
&- \left(\int_A
     h(u_1,\ldots,u_n)
     \rho_{Z}^{(n)}(u_1,\ldots,u_n)
     \de u_1\cdots\de u_n\right)^2
,
\nonumber
\end{align}
for any $A\subseteq S^n$, 
where $\rho^{(n)}_Z(\cdot)$, $n\geq1$, denote the product densities of $Z$; here we have that 
$j=0$ yields that $\{u_1,\ldots,u_j,u_{n+1},\ldots,u_{2n-j}\}=\{u_{n+1},\ldots,u_{2n}\}$ and $j=n$ yields that  $\{u_1,\ldots,u_j,u_{n+1},\ldots,u_{2n-j}\}=\{u_1,\ldots,u_n\}$.
    Moreover, the expectation in  \eqref{e:ExpectationInnovationProdDens} is 0 for any $A\subseteq S^n$ and any test function $h$ of the form \eqref{e:GeneralEstimatorNoPP} if and only if 
\begin{align}
  \label{e:ConditionProdDens}
  \xi(u_1,\ldots,u_n)
  \stackrel{a.e.}{=}&
  \rho_Z^{(n)}(u_1,\ldots,u_n).
\end{align}

    If, instead, $\xi,h:S^n\times\X\to\R$ are of the form \eqref{e:GeneralEstimator}, 
    when $\breve{X}$ admits an 
    $n$th-order Papangelou conditional intensity 
    $\breve\lambda^{(n)}(\cdot;\breve{X})$ then, for any $A\subseteq S^n$, 
    the bivariate innovation $\mathcal{I}_{\xi}^{h}(\cdot;Z,Y)$ satisfies 
    \begin{align}
    \label{e:ExpectationInnovation}
    &\E[\mathcal{I}_{\xi}^{h}(A;Z,Y)] 
      = \\
    =&
       \int_A
       \E\left[
       h(u_1,\ldots,u_n;Y)
       \left(
       \breve\lambda_1^{(n)}(u_1,\ldots,u_n;\breve X)
       -
       \xi(u_1,\ldots,u_n;Y)
       \right)
       \right]
       \de u_1\cdots\de u_n
       ,
       \notag
  \end{align}
  where 
  $$
  \breve\lambda_1^{(n)}(u;\breve{X})
  =
  \breve\lambda^{(n)}((u_1,1),\ldots,(u_n,1);\breve{X}),
  \quad 
  u=(u_1,\ldots,u_n)\in S^n, n\geq1,
  $$
  and 
  \begin{align*}
  &\E[\mathcal{I}_{\xi}^{h}(A;Z,Y)^2]
    =
  \\
  =& 
    \sum_{j=0}^n
    j! \binom{n}{j}^2
    \int_{S^{2n-j}}
    \1\{(u_1,\ldots,u_n),(u_1,\ldots,u_j,u_{n+1},\ldots,u_{2n-j})\in A\}
    \E\Big[
    h(u_1,\ldots,u_n;Y)
  \\
  &\times
    h(u_1,\ldots,u_j,u_{n+1},\ldots, u_{2n-j};Y)
    \breve\lambda_1^{(2n-j)}(u_1,\ldots, u_{2n-j};\breve{X})
    \Big]
    \de u_1\cdots \de u_{2n-j} 
  \\
  &+ 
     \int_{S^n}\int_{S^n}
     \1\{u,v\in A\}
     \E\left[
     h(u;Y)
     h(v;Y)
     \xi(u;Y)
     \xi(v;Y)
     \right]
     \de u
     \de v
    \\
    &- 2
    \int_{S^n}\int_{S^n}
     \1\{u,v\in A\}
     \E\left[
     h(u;Y)
     h(v;Y\cup\{u\})
    \xi(v;Y\cup\{u\})
    \breve\lambda_1^{(n)}(u;\breve{X})
    \right]
    \de u
    \de v
    .
    \end{align*}
    Assume further that 
    $\E[\breve\lambda_1^{(n)}(u_1,\ldots,u_n;\breve{X})^2]<\infty$ 
    for $|\cdot|^n$-almost any $(u_1,\ldots,u_n)\in S^n$. 
    Then, for any $A\subseteq S^n$ and any test function $h$ such that 
    $\E[h(u_1,\ldots,u_n;Y)^2]<\infty$ we have that 
    $\E[\mathcal{I}_{\xi}^{h}(A;Z,Y)]=0$ 
    if and only if
    \begin{equation}
    \label{e:ConditionPapangelou}
    \xi(u_1,\ldots,u_n;Y)
    \stackrel{a.e.}{=}
    \left.\E\left[\breve\lambda_1^{(n)}(u_1,\ldots,u_n;\breve{X})\right|Y\right]
    .
    \end{equation}

\end{thm}

Theorem \ref{thm:InnovationsMean} is general in the sense that we have not imposed any specific conditions
on the dependence structure between the two point processes $Z$ and
$Y=X\setminus Z$. 
Things become explicit, and for our purposes particularly interesting, when we require that $Z$ is an independent thinning of some point process $X$, in particular a $p$-thinning of $X$. 
The result below is a direct consequence of combining Theorem \eqref{lemma:Thinning} with the conditions in \eqref{e:ConditionProdDens} and  \eqref{e:ConditionPapangelou}. 

\begin{cor}
Assume the setting in Theorem \eqref{thm:InnovationsMean}. When $Z$ is an independent thinning of a point process $X$ in $S$, based on some retention probability function $p(u)\in(0,1)$, $u\in S$, it follows that the conditions in \eqref{e:ConditionProdDens} and  \eqref{e:ConditionPapangelou} translate to 
\begin{align}
\label{e:WeightsProdDens}
\xi(u_1,\ldots,u_n)
\stackrel{a.e.}{=}&
\prod_{i=1}^n p(u_i)
\rho_{X}^{(n)}(u_1,\ldots,u_n)
,
\\
\label{e:WeightsPapangelou}
\xi(u_1,\ldots,u_n;Y)
\stackrel{a.e.}{=}&
\frac{\prod_{i=1}^n p(u_i)}{\prod_{i=1}^n (1-p(u_i))}\lambda_Y^{(n)}(u_1,\ldots,u_n;Y),
\end{align}
respectively. When $Z$ is a $p$-thinning with retention probability $p\in(0,1)$, we set $p(\cdot)\equiv p$ in \eqref{e:WeightsProdDens} and \eqref{e:WeightsPapangelou}. 
\end{cor}

These observations will play a crucial role in the development of our statistical learning approach.

\section{Cross-validation and 
point process 
learning}
\label{s:Learning}

Having introduced the first building block of our statistical learning approach, namely the bivariate innovations, we next turn to the second building block, which is the notion of cross-validation for point processes. Once we have these two tools in hand, we combine them to define our (supervised) statistical learning/estimation approach.

\subsection{Cross-validation}
\label{s:CVgeneral}

Broadly speaking, cross-validation (CV) refers to a family of
techniques which are based on the idea of using one set of data to
test a model's predictive performance with respect to 
additional/new/incoming data  \citep{arlot2010survey}. In addition to testing a model's
ability to predict new data, objectives of CV include avoiding
overfitting and balancing bias and variance.

CV is essentially carried out by splitting/partitioning the full
dataset into a {\em training} dataset, on which the model is
fitted/trained, and a {\em validation} dataset, on which the
performance of the model is validated/evaluated/measured. Commonly,
this procedure is repeated a number of times, according to some
scheme/structure, yielding $k\geq1$ pairs
\begin{align*}
(\x_i^T,\x_i^V),
\quad
\x_i^V=\x\setminus\x_i^T,
\quad
i=1,\ldots,k, 
\end{align*}
of training and validation datasets; here,
$\x=\{x_1,\ldots,x_n\}\subseteq W\subseteq S$ is our sampled/observed
point pattern. For point processes, this may be formalised using thinning.

\begin{definition}
Given $k\geq1$ independently generated thinnings $Z_1,\ldots,Z_k$ of a
point process $X\subseteq S$, we refer to the collection of pairs $(X_i^T,X_i^V)=(Y_i,Z_i)$, $Y_i=X\setminus Z_i$, $i=1,\ldots,k$, as a cross-validation (CV) splitting/partitioning.

\end{definition}

Identically, given $k\geq1$ thinnings $\z_1,\ldots,\z_k$ of a point pattern $\x\subseteq S$, we will refer to the collection  $(\x_i^T,\x_i^V)=(\y_i,\z_i)$, $\y_i=\x\setminus\z_i$, $i=1,\ldots,k$, as a cross-validation (CV) splitting/partitioning.

\begin{remark}
If we have access to $k$ independent copies $X_1,\ldots,X_k$ of a point process $X$ (such repeated sampling is quite uncommon in practice), then we may naturally either let each $X_i^V$ be a thinning of $X_i$, $i=1,\ldots,k$, or ii) consider some $k'\geq1$ CV-splittings for each $X_i$, whereby we have $k$ different CV rounds and $kk'$ training-validation pairs.  
\end{remark}

Different CV procedures essentially provide different
ways of creating these pairs, i.e.~splitting the point process/pattern.  
Looking at our general definition for CV above, it immediately becomes clear that we essentially may consider an infinite number of ways to carry out CV-partitioning. We next list a couple of approaches commonly encountered in the literature:
\begin{itemize}
    \item {\em Classical $k$-fold CV:} Split $\x$ into $k$ folds/pieces of fixed equal/similar cardinality and, in each round $i=1,\ldots,k$, the $i$th fold plays the role of $\x_i^V$, while the union of the remaining $k-1$ folds plays the role of $\x_i^T$. 
    
    This sequential algorithm does not result in independent thinning since the assignment of a given point of $\x$ to a given fold depends on how points have previously been assigned to folds; a given fold runs full once it has a given number of points in it.
    
    \item {\em Leave-one-out CV:} This is just classical $k$-fold CV with $k=\#\x$, so that $\x_i^V=\{x_i\}$ and $\x_i^T=\x\setminus\{x_i\}$, $i=1,\ldots,\#\x$. 
    
    It has been argued that leave-one-out CV ensures a lower bias but a larger variance than classical $k$-fold CV with $k<\#\x$. 
    
\end{itemize}

Dependent thinnings are 
hard to work with since, for arbitrary point processes, it is generally hard to derive distributional properties for them -- we essentially have no control over the 
dependence structures between the training and validation sets. For this reason we will not look closer at classical $k$-fold, or leave-one-out CV for that matter. 
In Section \ref{s:Discussion} we discuss a couple of additional CV approaches. 
One of them, which we refer to as domain partitioning CV, is both quite natural and appealing. We do, however, show that this approach cannot be properly combined with our innovation-based statistical learning approach.

We argue that CV procedures for point processes should be based on independent thinning, where the retention probability function is bounded away from 0 and 1. 
The main argument is
that when we apply
such thinning, then, as we saw in Theorem \ref{lemma:Thinning}, we have control over distributional 
properties of most characteristic of interest, e.g.~product densities.


We further see no particular reason for choosing a specific form for the retention probability function used and, consequently, we impose the stronger argument that CV procedures for point processes should be based on $p$-thinning.

\subsubsection{CV based on {\em p}-thinning}
\label{s:PP CV}
We next propose two CV procedures, where the former essentially is
what in the literature is referred to as Monte-Carlo CV, or repeated random
sub-sampling validation, and the latter
is a variant of classical $k$-fold CV.

\begin{definition}[Monte-Carlo CV]
\label{def:MCCV}
Given $k\geq1$ $p$-thinnings $\z_1,\ldots,\z_k$, $p\in(0,1)$, of a
point pattern $\x$, we define {\em Monte-Carlo CV (MCCV)} as
generating the validation and training data splittings as 
$\x_i^V=\z_i$ and $\x_i^T=\x\setminus\z_i$, $i=1,\ldots,k$.
\end{definition}

Note that we in the case of MCCV may have that $\x_i^V\cap \x_j^V \neq \emptyset$,
$i\neq j$. One particularly appealing property of MCCV is that for any point pattern $\x$, by the law of large numbers and the central limit theorem, for a suitable function $f$ on $\X$, the mean $\frac{1}{k}\sum_{i=1}^k f(\x_i^V)$ converges a.s.~to
$\E[f(\x_1^V)]$ and weakly to a Gaussian random variable. 
Note further
that when we consider MCCV with $p\approx0$ we obtain something similar to classical leave-one-out CV, where  
the advantage of the former over the latter clearly is that we have theoretical control over things such as moments
characteristics; recall that Theorem \ref{lemma:Thinning} indicates different
distributional properties of $X_i^T$ and $X_i^V$.

As an alternative, where the training sets are not allowed to overlap,
we also propose a variant of classical $k$-fold CV. It should be noted that for large datasets the two types of $k$-fold CV should yield very similar results.

\begin{definition}[Multinomial CV]\label{def:kFoldMCCV}
Given some $k\geq2$, randomly label the point pattern $\x$ with iid marks $m(x)\in\{1,\ldots,k\}$, $x\in\x$, from a multinomial distribution with parameters $k$ and $p_1=\ldots=p_k=1/k$. We define {\em ($k$-fold) multinomial CV} as generating the validation and training data splittings as $\x_i^V=\{x\in\x:m(x)=i\}$ and $\x_i^T=\x\setminus\x_i^V$, $i=1,\ldots,k$.
  %
\end{definition}

First note that the main difference between this approach and MCCV is
that here $\x_i^V\cap \x_j^V = \emptyset$, $i\neq j$. In addition, each validation set is a $p$-thinning with retention probability $1/k$ and each training set is a $p$-thinning with retention probability $1-1/k=(k-1)/k$, whereby various distributional properties are known, e.g., the product densities of $X_i^V$ are given by $\rho_p^{(n)}(\cdot)=\rho^{(n)}(\cdot)/k^n$,
$n\geq1$; recall Theorem \ref{lemma:Thinning}.

We further note that an alternative (algorithmic) construction of multinomial CV is obtained by letting $\x_0^V=\emptyset$ and,  sequentially, letting $\x_i^V$ be a $p$-thinning of $\x\setminus\bigcup_{j=0}^{i-1}\x_j^V$ with retention probability $p_i=1/(k-(i-1))$ and $\x_i^T=\x\setminus\x_i^V$, $i=1,\ldots,k$.

Comparing the two approaches, aside from the validation sets not
overlapping in the multinomial CV approach, an upside to multinomial CV is that it only requires the specification of $k$, as opposed to the
pair $(p,k)$ in the MCCV case. Moreover, since we only have to deal with a
total of $k$ training-validation pairs in multinomial CV, it may be viewed as a computationally efficient version of MCCV. On the other
hand, one may argue that a drawback of multinomial MCCV is that we
(subjectively) should choose some fixed value $k$, where the "optimal"
choice very well may depend on e.g.~the samples size and/or the degree
of dependence in the underlying point process. Also, it may be that multinomial MCCV with $k=k'\geq2$ has worse statistical properties (bias/variance) than MCCV with $p=1/k'$ and $k$ (much) larger than $k'$. Moreover, in the case of MCCV, although $p$ is being kept fixed, $k$ is allowed to
be sequentially increased, which should result in sequentially
increased performance.

Concerning classical $k$-fold CV in the context of regression
analysis, the proper value for $k$ has been debated in the literature
and it has been argued that it should be chosen based on the sample
size $n$; e.g., \citet{james2013introduction} suggest that $k$ should
be chosen to be between 5 and 10. Finally, (for large samples) the point counts of the folds in classical $k$-fold CV and $k$-fold multinomial MCCV are approximately the same.

\subsubsection{Additional partitioning}
\label{s:TestSet}

Certain situations call for an additional layer of CV, so that we have triples  $(\x_i^T,\x_i^V,\x_i^E)$, $i=1,\ldots,k$,
where $\x_i^E=\x\setminus(\x_i^V\cup\x_i^T)$ is used for
one additional operation of evaluation of the fit; as we will see,
$\x_i^V$ is mostly used in the actual fitting process.  E.g., a
test set $\x_i^E$ is a subset of the data which is taken out before
the training-validation splitting takes place, and it is used to
evaluate the goodness of fit of the final fitted model by evaluating
how well the model predicts the test data. 
If warranted, we may naturally let $\x_i^E=\x_j^E$ for all $i,j=1,\ldots,k$, and if this is not the case we may generate $\x_i^E$, $i=1,\ldots,k$, as follows: apply some CV approach to $\x$ and denote the corresponding validation sets by $\x_i^E$, $i=1,\ldots,k$; then use another round of CV to generate $(\x_i^T,\x_i^V)$ from $\x\setminus\x_i^E$, $i=1,\ldots,k$. A possible extension here is to generate a CV-partitioning  $\{(\x_{ij}^T,\x_{ij}^V)\}_{j=1}^k$ from each $\x\setminus\x_i^E$,  $i=1,\ldots,k$. 
When we consider $p$-thinning-based CV here, each of $\x_i^T$, $\x_i^V$ and $\x_i^E$ will be a $p$-thinning of $\x$, but each with a different retention probability.

\subsection{Point process learning}
\label{s:CV based bivariate innovation statistics}

The heuristic/philosophical argument behind the 
approach laid out below is that a good estimation approach should result in a model which
does well in predicting "new" (validation) data, given the "current" 
(training) data. As we shall see, this is indeed the case in most settings. 

Recall from Section \ref{s:generalised weigthed innovations} that we consider an observed point pattern $\x=\{x_1,\ldots,x_n\}\subseteq W$, where the study region/domain $W\subseteq S$ is mostly bounded, and we assume that $\x$ is 
generated by (the restriction to $W$ of) some unknown point process $X$. 
Further, consider training-validation pairs $(\x_i^T,\x_i^V)$,
$i=1,\ldots,k$, which have been generated in accordance with Section \ref{s:CVgeneral}.



Our (supervised) learning/estimation approach is based on the idea of finding a minimiser $\widehat\theta\in\Theta$ of
some loss function $\mathcal{L}(\theta)$, $\theta\in\Theta$, which is based on some combination of either of
\begin{align}
  \label{e:InnovationsCV}
  \mathcal{I}_{\xi_{\theta}^n}^{h_{\theta}}(W^n;\x_i^V,\x_i^T)
  =&
     \mathop{\sum\nolimits\sp{\ne}}_{x_1,\ldots,x_{n}\in\x_i^V\cap W}
     h_{\theta}(x_1,\ldots,x_{n};\x_i^T)
  \\
   &-
     \int_{W^n}
     h_{\theta}(u_1,\ldots,u_{n};\x_i^T)
     \xi_{\theta}^n(u_1,\ldots,u_{n};\x_i^T)
     \de u_1\cdots\de u_{n},\nonumber
     \\
       \label{e:InnovationsCVNoPP}
  \mathcal{I}_{\xi_{\theta}^n}^{h_{\theta}}(W^n;\x_i^T)
  =&
  \mathop{\sum\nolimits\sp{\ne}}_{x_1,\ldots,x_{n}\in\x_i^T\cap W}
     h_{\theta}(x_1,\ldots,x_{n})
  \\
  &-
     \int_{W^n}
     h_{\theta}(u_1,\ldots,u_{n})
     \xi_{\theta}^n(u_1,\ldots,u_{n})
     \de u_1\cdots\de u_{n},
     \qquad i=1,\ldots,k.
     \nonumber
\end{align}
We employ the bivariate innovations \eqref{e:InnovationsCV} if the parametrised estimator family
$\Xi_{\Theta}^n=\{\xi_{\theta}^n: \theta\in\Theta\}$ and test function
family $\mathcal{H}_{\Theta}=\{h_{\theta}:\theta\in\Theta\}$ are of the form \eqref{e:GeneralEstimator}, and the univariate innovations  \eqref{e:InnovationsCVNoPP} if $\Xi_{\Theta}^n$ and $\mathcal{H}_{\Theta}$ are of the form \eqref{e:GeneralEstimatorNoPP}.
The choice of general innovation family, i.e.~parametrised estimator family and test function
family, is both context dependent and the important item here; e.g., as we shall see, non-parametric intensity estimation will require a different form for $\Xi_{\Theta}^n$ than parametric intensity estimation. Once these choices have been made, we proceed by
making a choice for the loss function $\mathcal{L}(\theta)$, $\theta\in\Theta$, to be employed.

\begin{remark}
  It should be emphasised that we here may combine several different
  collections of innovations, where in each we use different test
  functions and possibly also different training-validation pair generation approaches. 
\end{remark}

We emphasise that the use of \eqref{e:InnovationsCVNoPP} in fact results in a {\em point process subsampling} approach, very much akin to the one proposed in \citet{Moradi2019}. More specifically, we do not make explicit use of $\x_i^V$, $i=1,\ldots,k$, which is different from CV-based approaches. 

\begin{remark}
Regarding subsampling, more generally, one could plug $\x_i^T$, $i=1,\ldots,k$, into a family of univariate innovations 
and combine the resulting  subsample-innovations into a loss function.
\end{remark}


\subsubsection{Loss functions and point process learning}
\label{s:FixedTestFunctionLoss}

When specifying a loss function, 
once we have made a choice for $\Xi_{\Theta}^n=\{\xi_{\theta}^n: \theta\in\Theta\}$, which governs what we are interested in fitting, 
there are a few choices left to be made: the test function family $\mathcal{H}_{\Theta}$ to be considered in the bivariate innovations, the way the bivariate innovations are combined to form the loss function and the CV parameters used; regarding the latter, clearly, different CV partitioning approaches may yield completely varying results.

Before we look closer at specific loss function choices, we note that if $\#\x_i^V< n$, or equivalently $\#\x_i^T = \#\x-\#\x_i^V > \#\x-n$, then the sum in \eqref{e:InnovationsCV}
will be 0 and, consequently, it may be the case that
$\mathcal{I}_{\xi_{\theta}^n}^{h_{\theta}}(W^n;\x_i^V,\x_i^T)$  is
given by
$-\int_{W^n} h_{\theta}(\mathbf{u};\x_i^T)
\xi_{\theta}^n(\mathbf{u};\x_i^T) \de\mathbf{u}$. 
Similarly, in certain cases it is problematic to have $\#\x_i^T = 0$ both in \eqref{e:InnovationsCV} and \eqref{e:InnovationsCVNoPP}, where in the former case this would mean that we would try to predict
$\x_i^V=\x$ using $\x_i^T=\emptyset$. In some cases this is not a problem (see e.g.~Section \ref{s:ConstantIntensity}) but in other cases this will be problematic.  
When we do not want such a pair
$(\x_i^T,\x_i^V)$ to contribute to the final loss function, we may multiply the corresponding innovation by an indicator function $I_i$, $i=1,\ldots,k$, for bounds on the cardinality of the training set, 
and let the loss function be given by a combination of the $k$ resulting terms. 
Hence, typically we would let 
\begin{align}
\label{e:IndicatorInnovation}
I_i
=
\left\{
\begin{array}{ll}
\1\{1\leq\#\x_i^T\} &
\text{if we consider \eqref{e:InnovationsCVNoPP} ($\Xi_{\Theta}^n$ and $\mathcal{H}_{\Theta}$ are of form \eqref{e:GeneralEstimatorNoPP})},
\\
\1\{1\leq\#\x_i^T\leq\#\x-n\} & 
\text{if we consider \eqref{e:InnovationsCV} ($\Xi_{\Theta}^n$ and $\mathcal{H}_{\Theta}$ are of form \eqref{e:GeneralEstimator})},
\end{array}
\right.
\end{align}
and let the loss function be given by a combination of the terms
$$
\widetilde{\mathcal{I}}_{\xi_{\theta}^n}^{h_{\theta}}(W^n;\x_i^T)
=
I_i
\mathcal{I}_{\xi_{\theta}^n}^{h_{\theta}}(W^n;\x_i^T)
\text{ or }
\widetilde{\mathcal{I}}_{\xi_{\theta}^n}^{h_{\theta}}(W^n;\x_i^V,\x_i^T)
=
I_i
\mathcal{I}_{\xi_{\theta}^n}^{h_{\theta}}(W^n;\x_i^V,\x_i^T)
,
\quad i=1,\ldots,k,
$$
where the former is used when we consider \eqref{e:InnovationsCVNoPP} and the latter is used when we consider \eqref{e:InnovationsCV}. 
Defining $\widetilde h_{\theta}(u_1,\ldots,u_{n};\x_i^T) 
= \1\{1\leq\#\x_i^T\leq\#\x-n\}h_{\theta}(u_1,\ldots,u_{n};\x_i^T)$, 
$i=1,\ldots,k$, note that the latter satisfies 
\begin{align*}
\widetilde{\mathcal{I}}_{\xi_{\theta}^n}^{h_{\theta}}(W^n;\x_i^V,\x_i^T)
=
\1\{1\leq\#\x_i^T\leq\#\x-n\}
\mathcal{I}_{\xi_{\theta}^n}^{h_{\theta}}(W^n;\x_i^V,\x_i^T)
=
\mathcal{I}_{\xi_{\theta}^n}^{\widetilde h_{\theta}}(W^n;\x_i^V,\x_i^T),
\end{align*}
i.e.~the indicator function may be absorbed into the test function and thereby into the innovation. 
When we consider \eqref{e:InnovationsCVNoPP} we see no reason for including the event $\{\#\x_i^T\leq\#\x-n\}$ in the indicator function, since we here essentially deal with a subsampling approach, rather than a prediction approach; if required, one may here naturally choose to multiply the innovation by  $I_i=\1\{1\leq\#\x_i^T\leq\#\x-n\}$ instead. 
Similarly, we do not want to rule out the possibility that there may be situations where one would want $I_i=\1\{1\leq\#\x_i^T\}$ when considering \eqref{e:InnovationsCV}. 
However, as we will see (e.g.~Section \ref{s:ConstantIntensity}), there are situations where the preferred choice is to set $I_i=1$ for all $i=1,\ldots,k$, i.e.~to not include indicator functions for the cardinalities of the training sets in the loss function. 

\begin{definition}
Consider the supervised statistical learning framework above, where the objective is to find a minimiser of a loss function 
\begin{align}
\label{e:GeneralCVLoss}
\mathcal{L}(\theta)
=
\mathcal{L}(\theta;\{(\x_i^T,\x_i^V)\}_{i=1}^k,n,p,k,\Xi_{\Theta}^n,\mathcal{H}_{\Theta}), 
\quad \theta\in\Theta,
\end{align}
which is generated by a combination of $\widetilde{\mathcal{I}}_{\xi_{\theta}^n}^{h_{\theta}}(W^n;\x_i^V,\x_i^T)$ or $\widetilde{\mathcal{I}}_{\xi_{\theta}^n}^{h_{\theta}}(W^n;\x_i^T)$, $i=1,\ldots,k$. We refer to this approach as {\em point process learning} and we refer to any minimiser $\widehat\theta$ of $\mathcal{L}(\theta)$, $\theta\in\Theta$, as a {\em 
point process-learned (PPL) estimate}. 
\end{definition}

Defining 
\[
\mathcal{T}_k=\{j\in\{1,\ldots,k\}:I_j=1\},
\]
where $\mathcal{T}_k=\{1,\ldots,k\}$ when we set $I_i=1$ for all $i=1,\ldots,k$, 
we here in particular see the following three loss function candidates as particularly interesting/natural:
\begin{align}
  \label{e:EstFunGeneralMedian}
  \mathcal{L}_1(\theta)
  =&
     \frac{1}{k}
     \sum_{i=1}^k
     |\widetilde{\mathcal{I}}_{\xi_{\theta}^n}^{h_{\theta}}(W^n;\x_i^V,\x_i^T)|
     \propto
     \frac{1}{\#\mathcal{T}_k}
     \sum_{i\in\mathcal{T}_k}
     |\mathcal{I}_{\xi_{\theta}^n}^{h_{\theta}}(W^n;\x_i^V,\x_i^T)|
     ,
  \\
  \label{e:EstFunGeneral}
  \mathcal{L}_2(\theta)
  =&
     \frac{1}{k}
     \sum_{i=1}^k
     \widetilde{\mathcal{I}}_{\xi_{\theta}^n}^{h_{\theta}}(W^n;\x_i^V,\x_i^T)^2
     \propto
     \frac{1}{\#\mathcal{T}_k}
     \sum_{i\in\mathcal{T}_k}
     \mathcal{I}_{\xi_{\theta}^n}^{h_{\theta}}(W^n;\x_i^V,\x_i^T)^2
     ,
  \\
  \label{e:EstFunGeneralMean}
  \mathcal{L}_3(\theta)
  =&
     \left(\frac{1}{k}
     \sum_{i=1}^k
     \widetilde{\mathcal{I}}_{\xi_{\theta}^n}^{h_{\theta}}(W^n;\x_i^V,\x_i^T)
     \right)^2
     \propto
     \left(\frac{1}{\#\mathcal{T}_k}
     \sum_{i\in\mathcal{T}_k}
     {\mathcal{I}}_{\xi_{\theta}^n}^{h_{\theta}}(W^n;\x_i^V,\x_i^T)
     \right)^2
     ,
\end{align}
where we replace $\widetilde{\mathcal{I}}_{\xi_{\theta}^n}^{h_{\theta}}(W^n;\x_i^V,\x_i^T)$ by $\widetilde{\mathcal{I}}_{\xi_{\theta}^n}^{h_{\theta}}(W^n;\x_i^T)$ when $\xi_{\theta}^n$ and $h_{\theta}$ are of the form
\eqref{e:GeneralEstimatorNoPP}. 
The loss function \eqref{e:EstFunGeneral} is probably the most natural
one and essentially corresponds to an $L_2$-loss. Similarly,
\eqref{e:EstFunGeneralMedian} is a robust version which essentially
corresponds to an $L_1$-loss.  Employing
\eqref{e:EstFunGeneralMean} on the other hand, means finding a parameter $\theta$ such
that the mean of all innovation terms is (close to) 0.  By Hölder's and
Jensen's inequalities,
$\mathcal{L}_1(\theta)^2\leq\mathcal{L}_2(\theta)\leq
k\mathcal{L}_3(\theta)$, and by setting $k=1$ and squaring
\eqref{e:EstFunGeneralMedian}, these three losses coincide.

As an alternative approach here, one may exploit
the empirical distribution 
of
\begin{align}
\label{e:IndividualEstimates}
\widehat\theta_i^{\mathcal{H}_{\Theta}}
=
\widehat\theta((\x_i^T,\x_i^V),n,p,W,\Xi_{\Theta}^n,\mathcal{H}_{\Theta})
\in\Theta,
\quad 
i\in\mathcal{T}_k,
\end{align}
which are minimisers of $\theta\mapsto\widetilde{\mathcal{I}}_{\xi_{\theta}^n}^{h_{\theta}}(W^n;\x_i^V,\x_i^T)^2$ or $\theta\mapsto\widetilde{\mathcal{I}}_{\xi_{\theta}^n}^{h_{\theta}}(W^n;\x_i^T)^2$,  $\theta\in\Theta$, $i\in\mathcal{T}_k$. 
The
corresponding sample median or mean,
\begin{align}
\label{e:IndividualEstimatesMedian}
    &\med\{\widehat\theta_i^{\mathcal{H}_{\Theta}}:i\in\mathcal{T}_k\}
    ,
    \\
    \label{e:IndividualEstimatesMean}
    &\frac{1}{\#\mathcal{T}_k}\sum_{i\in\mathcal{T}_k}\widehat\theta_i^{\mathcal{H}_{\Theta}},
\end{align}
may serve as alternatives to the estimates obtained through \eqref{e:GeneralCVLoss}. 

\begin{remark}
Empirical quantiles of \eqref{e:IndividualEstimates} 
may serve as
confidence/uncertainty regions for $\theta_0$ and $\widehat\xi^n=\frac{1}{\#\mathcal{T}_k}\sum_{i\in\mathcal{T}_k}\xi_{\widehat\theta_i}^n$ as a final point estimate; note the connections to the so-called resample-smoothing approach of \citet{Moradi2019}. 
A further loss function alternative is to minimise the maximum of the individual squared/absolute innovations. 
These alternatives are currently being explored in a parallel paper. 
\end{remark}

\subsubsection{Test function choices, hyperparameters and indentifiability}
\label{s:TestFunctions}

As can be seen in Section \ref{s:UnivariateInnovations}, the
literature offers a few suggestions on suitable test functions
$h_{\theta}$, $\theta\in\Theta$, to be employed.
Most notably, when $\xi_{\theta}^n$ is differentiable in $\theta$, in the univariate setting the
test function
$h_{\theta}(\cdot)=\partial\xi_{\theta}^n(\cdot)/\partial\theta=\nabla\xi_{\theta}^n(\cdot)$
turns $\theta\mapsto\mathcal{I}_{\xi_{\theta}^n}^{h_{\theta}}(\cdot)$
into a Poisson process
likelihood score-type function. A further group of candidates
which is encountered may be summarised as
$h_{\theta}(\cdot)=f(\xi_{\theta}^n(\cdot))$, where $f(x)=x^{\gamma}$, $\gamma\in\R$
\citep{baddeley2005residual,cronie2018bandwidth}. E.g., $\gamma=0$
corresponds to so-called raw innovations, $\gamma=-1/2$ corresponds to
so-called Pearson innovations and $\gamma=-1$ corresponds to so-called
Stoyan-Grabarnik/inverse innovations \citep{baddeley2005residual}.
The interesting thing with $\gamma=-1$ is that we obtain
\begin{align*}
  \mathcal{I}_{\xi_{\theta}^n}^{h_{\theta}}(A;\z,\y)
  =&
     \sum_{(x_1,\ldots,x_n)\in\z_{\neq}^n\cap A}
     h_{\theta}(x_1,\ldots,x_{n};\y)
     -
     \left|A\cap\supp(\xi_{\theta}^n(\cdot;\y))\right|,
\end{align*}
where the size of the support
$A\cap\supp(\xi_{\theta}^n(\cdot;\y))=\{\mathbf{u}\in
A:\xi_{\theta}^n(\mathbf{u};\y)>0\}\subseteq A$, which may vary depending on $\y$, is given by $|A|$ if
$\xi_{\theta}^n(\cdot;\y)$ is strictly positive. When 
$\xi_{\theta}^n(\cdot;\y)$ is not necessarily strictly positive, a
convenient approximation could be to simply replace the support by $A$ in
the expression above and proceed with the minimisation of the
corresponding loss function. This is convenient from a computational
point of view, since we can omit computing the (possibly hard to deal
with) integrals in the innovations.

Here, as a proof of concept, we have considered the scenario where we fix the test function family (and the CV parameters) a priori, whereby we proceed by specifying how to combine the innovations to obtain the loss function. 
Note, however, that the CV parameters as well as the test function family (or more generally the loss function) considered may each be viewed a hyperparameter. 
In Section \ref{s:Discussion} we discuss different strategies to finding "optimal" test functions and, more generally, hyperparameters.


A closely related and quite important issue, which we have avoided
mentioning up to this point, is identifiability. Ideally we should
have that
$\mathcal{L}(\theta_1)=\mathcal{L}(\theta_2)$
implies that $\theta_1=\theta_2$ for any $\theta_1,\theta_2\in\Theta$. We note e.g.~that both $x\mapsto|x|$ and $x\mapsto x^2$,
$x\in\R$, are convex and that sums of convex functions are convex so
if all 
the $k$ innovations 
are convex then so will the loss functions be and,
consequently, 
any local minimiser will be a global minimiser. 
Moreover, 
differentiation of the innovations  
may reveal whether convexity holds, but looking closer at the innovations we note that tractable derivatives may be a bit too much to hope for in many situations. Moreover, the fact that $f(x)=x^{-\gamma}$,
$x>0$, $\gamma>0$, is convex may possibly also be exploited to show
that the innovations are convex. We finally note that an innovation
does not necessarily attain the value 0 (see e.g.~the
"leave-one-out"-discussion in \citet{cronie2018bandwidth}), whereby it cannot be viewed as an estimating equation, and this makes the derivation of closed form estimators based on loss functions involving combinations of innovations hard.



\subsection{Point process learning with {\em p}-thinning-based CV}

Assuming that we use either MCCV or multinomial CV in our point process learning approach, we next look closer at how the different components of the bivariate
innovations
$\mathcal{I}_{\xi_{\theta}^n}^{h_{\theta}}(W^n;\x_i^V,\x_i^T)$,
$i=1,\ldots,k$, are specified in three of the most typical estimation
settings, namely parametric product density/intensity estimation, non-parametric product density/intensity estimation and Papangelou conditional intensity fitting; Section \ref{s:UnivariateInnovations} indicates how a few other common estimation approaches may be combined with our point process learning framework. Here we focus on the choice of $\Xi_{\Theta}^n=\{\xi_{\theta}^n: \theta\in\Theta\}$, since we view the choice of test function family as something related to the choice of loss function.

\subsubsection{Parametric product density/intensity estimation}
\label{s:ProductDensity}

When we carry out parametric $n$th-order product density estimation,
$\Xi_{\Theta}^n$ and
$\mathcal{H}_{\Theta}$ are of the form \eqref{e:GeneralEstimatorNoPP}, whereby we use the innovations in  \eqref{e:InnovationsCVNoPP} and set
$$
    \xi_{\theta}^n(u_1,\ldots,u_n)
    =
    w\rho_{\theta}^{(n)}(u_1,\ldots,u_n)
    =
    (1-p)^n\rho_{\theta}^{(n)}(u_1,\ldots,u_n)
    , 
    \quad u_1,\ldots,u_n\in W, 
    \quad \theta\in\Theta,
$$
for some parametric family $\rho_{\theta}^{(n)}$, $\theta\in\Theta$,
of $n$th-order product densities. 
Consequently,  
the innovations 
$\mathcal{I}_{\xi_{\theta}^n}^{h_{\theta}}(W^n;\x_i^T) =
\mathcal{I}_{(1-p)^n\rho_{\theta}^{(n)}}^{h_{\theta}}(W^n;\x_i^T)$ 
in \eqref{e:InnovationsCVNoPP} become 
\begin{align}
  \label{e:InnovationParametric}
  &
    \mathcal{I}_{(1-p)^n\rho_{\theta}^{(n)}}^{h_{\theta}}(W^n;\x_i^T)
    =
  \\
  =&
     \mathop{\sum\nolimits\sp{\ne}}_{x_1,\ldots,x_{n}\in\x_i^T\cap W}
     h_{\theta}(x_1,\ldots,x_{n})
     -
     (1-p)^n
     \int_{W^n}
     h_{\theta}(u_1,\ldots,u_{n})
     \rho_{\theta}^{(n)}(u_1,\ldots,u_{n})
     \de u_1\cdots\de u_{n}.
     \nonumber
\end{align} 
We make the choice $w=(1-p)^n$ since when $X$ has $n$th-order
product density $\rho_{\theta_0}^{(n)}$, $\theta_0\in\Theta$, 
equation~\eqref{e:WeightsProdDens}, which here is equivalent to the condition in  \eqref{e:ConditionProdDens}, yields that this is a necessary and sufficient condition for 
$\E[\mathcal{I}_{\xi_{\theta_0}^n}^{h_{\theta_0}}(W^n;X_i^T)]=0$ to hold 
for any training-validation pair $(X_i^T,X_i^V)$, $i=1,\ldots,k$, and
any test function $h_{\theta_0}$. 

  \begin{remark}
    As we shall see in Section \ref{s:PapangelouEstimation}, when
    we let $\xi_{\theta}^n(\cdot;\cdot)=w\lambda_{\theta}^{(n)}(\cdot;\cdot)$
    for a family of $n$th-order Papangelou conditional intensities
    $\lambda_{\theta}^{(n)}$, $\theta\in\Theta$, the weight should be
    $w=p^n/(1-p)^n$. Since the $n$th order product densities and
    Papangelou conditional intensities coincide in the case of a
    Poisson process, one could get the impression that there is a
    contradiction. To see that this is not the case, note that the conclusions of 
    \eqref{e:ConditionPapangelou} and \eqref{e:ConditionProdDens},
    which determine how we specify the weights, coincide for a
    Poisson process so the weight should be $w=p^n$.
  \end{remark}

It is further worth noting that in the MCCV case, by the law
of large numbers, we 
obtain that 
$\lim_{k\to\infty}\frac{1}{k}\sum_{i=1}^k\mathcal{I}_{\xi_{\theta}^n}^{h_{\theta}}(W^n;\x_i^T)
\stackrel{a.s.}{=}
(1-p)^n\mathcal{I}_{\xi_{\theta}^n}^{h_{\theta}}(W^n;\x)$ and, by the central limit theorem,  
$(\frac{1}{k}\sum_{i=1}^k\mathcal{I}_{\xi_{\theta}^n}^{h_{\theta}}(W^n;\x_i^T)-(1-p)^n\mathcal{I}_{\xi_{\theta}^n}^{h_{\theta}}(W^n;\x))/(p^n(1-p)^n\sum_{x_1,\ldots,x_{n}\in\x\cap
  W}^{\neq} h_{\theta}(x_1,\ldots,x_{n}))^{1/2}$ tends weakly to a standard
normal distribution. 



In Section \ref{s:ConstantIntensity} we look closer at intensity estimation in the homogeneous
case: 
$\rho_{\theta}(\cdot)\equiv\theta$, $\theta\in\Theta=(0,\infty)$, where the point process $X$ has constant intensity
$\rho_{\theta_0}(\cdot)\equiv\theta_0\in\Theta$, and we estimate $\theta_0$ based on \eqref{e:InnovationParametric} with $n=1$.

\subsubsection{Parametric Papangelou conditional intensity estimation}
\label{s:PapangelouEstimation}

In the case of Papangelou conditional intensities, the formulation of
prediction is straightforward, since Papangelou conditional intensities may be interpreted as a
conditional densities; recall \eqref{e:PapangelouFinitePP}. 
To carry out parametric $n$th-order Papangelou
conditional intensity estimation, we set
\begin{align*}
\xi_{\theta}^n(u_1,\ldots,u_n;\y)
=&
w\lambda_{\theta}^{(n)}(u_1,\ldots,u_n;\y)
\\
=&\frac{p^n}{(1-p)^n}\lambda_{\theta}^{(n)}(u_1,\ldots,u_n;\y),
\quad u_1,\ldots,u_n\in W, 
\quad \y\in\X,
\quad \theta\in\Theta,
\end{align*}
for some parametric family $\lambda_{\theta}^{(n)}$,
$\theta\in\Theta$, of $n$th-order Papangelou conditional intensities; here $\Xi_{\Theta}^n$ and $\mathcal{H}_{\Theta}$
are of the form \eqref{e:GeneralEstimator} so we use the innovations in \eqref{e:InnovationsCV} for the estimation. 
Regarding the choice $w=p^n/(1-p)^n$ for the weight, equation \eqref{e:WeightsPapangelou} tells us that if $X$ has Papangelou
conditional intensity $\lambda_{\theta_0}$ for some $\theta_0\in\Theta$, 
making this choice is equivalent to having 
$\E[\mathcal{I}_{\xi_{\theta_0}^n}^{h_{\theta_0}}(W^n;X_i^V,X_i^T)]=0$
for any training-validation pair $(X_i^T,X_i^V)$, $i=1,\ldots,k$, and
any test function $h_{\theta_0}$. 
The innovations $\mathcal{I}_{\xi_{\theta}^n}^{h_{\theta}}(W^n;\x_i^T,\x_i^V)=\mathcal{I}_{(p/(1-p))^n\lambda_{\theta}^{(n)}}^{h_{\theta}}(W^n;\x_i^V,\x_i^T)$ in \eqref{e:InnovationsCV} here become
  \begin{align}
    \label{e:PapangelouInnovation}
    \mathcal{I}_{(p/(1-p))^n\lambda_{\theta}^{(n)}}^{h_{\theta}}(W^n;\x_i^V,\x_i^T)
    =&
       \mathop{\sum\nolimits\sp{\ne}}_{x_1,\ldots,x_{n}\in\x_i^V\cap W}
       h_{\theta}(x_1,\ldots,x_{n};\x_i^T)
    \\
     &-
       \int_{W^n}
       h_{\theta}(u_1,\ldots,u_{n};\x_i^T)
       \frac{p^n
       \lambda_{\theta}^{(n)}(u_1,\ldots,u_{n};\x_i^T)
       }{(1-p)^n}
       \de u_1\cdots\de u_{n}
       .
       \nonumber
\end{align}
Due to 
the relationship between
the $n$th-order and the first-order Papangelou conditional intensities
in expression \eqref{e:nthPapangelou}, we henceforth focus on the case
$n=1$, i.e.
\begin{align}
\label{e:PapangelouInnovation1}
  \mathcal{I}_{\xi_{\theta}^1}^{h_{\theta}}(W;\x_i^V,\x_i^T)
  =&
     \mathcal{I}_{p(1-p)^{-1}\lambda_{\theta}}^{h_{\theta}}(W;\x_i^V,\x_i^T)
  \\
  =&
     \sum_{x\in\x_i^V\cap W}
     h_{\theta}(x;\x_i^T)
     -
     \frac{p}{1-p}
     \int_{W}
     h_{\theta}(u;\x_i^T)
     \lambda_{\theta}(u;\x_i^T)
     \de u
     .
     \nonumber
\end{align}

\begin{remark}
Recalling the connection between Papangelou conditional intensities and classical conditional intensities mentioned in Remark \ref{rem:ClassicalConditionalIntesnsity}, when $\xi_{\theta}^1$ is based on a parametric family $\lambda_{\theta}^{\dagger}$, $\theta\in\Theta$, of classical conditional intensities, we strongly suspect that the weight used should be $w=p/(1-p)$. This is not something we will look closer at in the current paper. On the other hand, it may be noted that e.g.~Hawkes processes, which are extensively applied temporal point processes, have known closed forms for their Papangelou conditional intensities \citep{yang2019stein}, so \eqref{e:PapangelouInnovation1} may be used to fit Hawkes processes.

\end{remark}

For illustrational purposes, in Section \ref{s:HardCore} we look closer at how
\eqref{e:PapangelouInnovation1} can be used to fit a hard-core process
(recall Section \ref{s:ExpFamily}).

\begin{remark}
  As an aside, given an estimate $\widehat\theta$ of $\theta_0$, note
  that the bivariate residual 
  $\mathcal{I}_{\lambda_{\widehat\theta}}^{h_{\widehat\theta}}(W;\z,\y)$
  may be used as an estimate of dependence between two point processes
  $Y$ and $Z$, with realisations $\y$ and $\z$, respectively.
Similarly, 
  $\frac{1}{k}\sum_{i=1}^k\lambda_{\widehat\theta}^{\#\x_i^V}(\x_i^V;\x_i^T)$
  may be used to measure the goodness of fit and consequently compare
  the performance of competing Papangelou conditional intensity models
  (given a fixed collection of training-validation sets), since this
  essentially is an average of conditional densities. The higher the
  value, the better the model predicts validation data from training
  data. A caveat here though: the training and validation sets used here should
  be other than the ones used in the original fitting.
\end{remark}

\subsubsection{Non-parametric product density/intensity estimation}
\label{s:NonParametricsCV}

The non-parametric product density/intensity 
estimation setting is a bit more delicate than the parametric
estimation setting. Here $\Xi_{\Theta}^n$ and $\mathcal{H}_{\Theta}$ are of the
form \eqref{e:GeneralEstimator} and we set
    $$
    \xi_{\theta}^n(u_1,\ldots,u_n;\y)=w\widehat\rho_{\theta}^{(n)}(u_1,\ldots,u_n;\y),
    \quad u_1,\ldots,u_n\in W, 
    \quad \y\in\X,
    \quad \theta\in\Theta,
    $$
    where $\widehat\rho_{\theta}^{(n)}$ is some non-parametric
    product density estimator with tuning/smoothing
    parameter $\theta\in\Theta$.  Recall the setting in Section
    \ref{s:NonParametrics}, where we consider a non-parametric
    intensity estimator $\widehat\rho_{\theta}(u,\y)$, $u\in W$,
    $\y\in\X$, $\theta$, i.e.~$n=1$, and our aim is to choose the tuning parameter
    $\theta$ optimally.
    
    We set the weight to $w=p^n/(1-p)^n$, whereby the innovations  $\mathcal{I}_{\xi_{\theta}^n}^{h_{\theta}}(W^n;\x_i^V,\x_i^T)
        =
        \mathcal{I}_{p^n(1-p)^{-n}\widehat\rho_{\theta}^{(n)}}^{h_{\theta}}(W^n;\x_i^V,\x_i^T)$ in \eqref{e:InnovationsCV}
    become
    \begin{align*}
      \mathcal{I}_{p^n(1-p)^{-n}\widehat\rho_{\theta}^{(n)}}^{h_{\theta}}(W^n;\x_i^V,\x_i^T)
        =&
        \mathop{\sum\nolimits\sp{\ne}}_{x_1,\ldots,x_{n}\in\x_i^V\cap W}
        h_{\theta}(x_1,\ldots,x_{n};\x_i^T)
      \\
      &-
        \int_{W^n}
        h_{\theta}(u_1,\ldots,u_n;\x_i^T)
        \frac{p^n \widehat\rho_{\theta}^{(n)}(u_1,\ldots,u_n;\x_i^T)}{(1-p)^n}
        \de u_1\cdots\de u_n
        .
\end{align*}
The heuristic motivation for this choice is the following. If $X$ has $n$th-order
product density $\rho^{(n)}(\cdot)$, Theorem \ref{lemma:Thinning} yields that
$\rho^{(n)}(\cdot)=\rho_{1-p}^{(n)}(\cdot)/(1-p)^n$, where
$\rho_{1-p}^{(n)}(\cdot)$ is the $n$th-order product density of
$X_i^T$, which is a $p$-thinning of $X$ with retention probability
$1-p$. Hence, $\widehat\rho_{\theta}^{(n)}(\cdot,X_i^T)/(1-p)^n$ is a
sensible choice for the estimation of the product density of $X$ and,
consequently,
$\widetilde\rho_{\theta}^{(n)}(\cdot,\x_i^T)=\widehat\rho_{\theta}^{(n)}(\cdot,\x_i^T)p^n/(1-p)^n$
is a sensible choice for the estimation of the product density of
$X_i^V$, which is an independent thinning with retention probability $p$. 
More formally, \eqref{e:WeightsPapangelou} tell us that for the 
  corresponding bivariate innovation to have expectation 0, we must
  have that $\xi_{\theta}^n(u_1,\ldots,u_n;X_i^T)$ should coincide with 
  $\lambda_{X_i^T}^{(n)}(u_1,\ldots,u_n;X_i^T)p^n/(1-p)^n$, where
  $\lambda_{X_i^T}^{(n)}(\cdot;X_i^T)$ is the $n$th-order Papangelou conditional
  intensity of $X_i^T$. 
  Taking expectations, we would thus (at least) need
  that
    $$
    \E[w\widehat\rho_{\theta}^{(n)}(u_1,\ldots,u_n;X_i^T)]
    =
    p^n\rho_X^{(n)}(u_1,\ldots,u_n)=\rho_{X_i^V}^{(n)}(u_1,\ldots,u_n),
    $$
    where $\rho_X^{(n)}$ and $\rho_{X_i^V}^{(n)}$ are the (true)
    $n$th-order product densities of $X$ and $X_i^V$,
    respectively. This in turn would require that
    $\widehat\rho_{\theta}^{(n)}$ is unbiased for arbitrary point
    processes and, to the best of our knowledge, no such estimators
    exist. However, minimising the squared innovation would on average
    "force" the estimator to resemble an unbiased one.

  In Section \ref{s:Kernel} we look closer at kernel intensity
  estimation, in particular optimal bandwidth selection.

\section{Applications}
\label{s:Applications}
As a proof of concept, we next look closer at a few special scenarios, which deal with some of the most common estimation settings encountered in the literature.

\subsection{Parametric intensity estimation: constant intensity estimation}
\label{s:ConstantIntensity}

Consider a homogeneous point process $X$ with unknown constant intensity 
$\theta_0 \in (0,\infty)$. 
To carry out parametric intensity estimation, we consider the setting 
of Section~\ref{s:ProductDensity}, where in  \eqref{e:InnovationParametric}
we 
set $\Theta = (0,\infty)$, $n=1$, 
$\rho_{\theta}(\cdot)\equiv\theta\in(0,\infty)$, and 
$\xi_{\theta}^1(\cdot)=(1-p) \theta$. 
We further assume that the test function family  $\mathcal{H}_{\Theta}=\{h_{\theta}:\theta\in\Theta\}$ contains only one element $h$, which consequently does not depend on $\theta$.
We may then estimate $\theta_0\in (0,\infty)$ by 
minimising e.g.~one of the loss
functions \eqref{e:EstFunGeneralMedian}, \eqref{e:EstFunGeneral} or
\eqref{e:EstFunGeneralMean}, using the
univariate innovations
\begin{align}
\label{e:InnovationParametricIntensity}
    \mathcal{I}_{\xi_{\theta}^1}^{h}(W;\x_i^T)
    =
    \sum_{x\in\x_i^T\cap W}
     h(x)
     -
     \theta (1-p)
     \int_{W}
     h(u)
     \de u
     , 
     \quad i=1,\ldots,k.
\end{align} 
Below, 
we show that
for specific choices of $h$, our parametric intensity estimation approach outperforms the classical estimator
\begin{align}
  \label{e:ClassicalEstimator}
  \widetilde\theta=\widetilde\theta(X,W)=\frac{X(W)}{|W|}
\end{align}
in certain cases. 

Central to the results in this section is the following generalisation of the classical intensity estimator in  \eqref{e:ClassicalEstimator}.
\begin{definition}
Given some $W\subseteq S$ and a test function $h:W\to\R$ such that $0<|\int_{W} h(u) \de u|<\infty$, the {\em $h$-weighted intensity estimator} of the intensity $\theta_0>0$ of a homogeneous point process $X$ on $S$ is given by
\begin{align}
\label{e:hWeightedIntensityEst}
\widetilde\theta_h(X,W)
  =
  \frac{ \sum_{x\in X\cap W} h(x) }
{\int_{W} h(u) \de u }
.
\end{align}
In particular, $\widetilde\theta_h(X,W)=\widetilde\theta(X,W)$ for any constant test function $h(\cdot)\equiv c\neq0$. 

\end{definition}

Given training-validation pairs $\{(\x_i^T,\x_i^V)\}_{i=1}^k$ in accordance with Section~\ref{s:PP CV}, below we consider scaled versions 
\begin{equation*}
\widehat\theta((\x_i^T,\x_i^V),p,W,h)
=
    \frac{
  \sum_{x\in\x\cap W} h(x)  \1\{x\in\x_i^T\} }{ (1-p) \int_{W} h(u) \de u}
  =
  \frac{ \sum_{x\in\x_i^T\cap W} h(x) }
{ (1-p) \int_{W} h(u) \de u }
=
  \frac{\widetilde\theta_h(\x_i^T,W)}{1-p}
\end{equation*}
of \eqref{e:hWeightedIntensityEst}, for which we use the short notation $\widehat\theta_i^h$, $i=1,\ldots,k$. In particular, if the test function $h$ is given by a non-null constant then $\widehat\theta_i^h=\widetilde\theta(\x_i^T,W)/(1-p)=\x_i^T/((1-p)|W|)$. As we shall see in the next result, which is proved in Section \ref{s:Proofs}, these are the estimates \eqref{e:IndividualEstimates} in the current context.


\begin{thm}
\label{thm:ConstantIntensity}
Let $\x$ be a realisation of a homogeneous point process $X$, with constant intensity $\theta_0>0$, which is observed within $W\subseteq S$. 
Let $\{(\x_i^T,\x_i^V)\}_{i=1}^k$ be the associated CV-partitioning, as presented in Section~\ref{s:PP CV}, and let $h$ be a test function satisfying  $0<|\int_{W} h(u) \de u|<\infty$. 

Let the indicator functions in \eqref{e:IndicatorInnovation} be given by $I_i=\1\{1\leq\#\x_i^T\}$, $i=1,\ldots,k$, whereby $\mathcal{T}_k = \{i\in \{1,\ldots,k\}; \#\x_i \geq 1\}$. It then follows that the estimates in \eqref{e:IndividualEstimates} are given by $\widehat\theta_i^h$, $i=1,\ldots,k$, whereby 
the estimate in  \eqref{e:IndividualEstimatesMedian} is given by
\begin{align*}
\widehat{\theta}_1(\{(\x_i^T,\x_i^V)\}_{i=1}^k,p,W,h) 
&=
\med\{\widehat\theta_i^h : 
i \in \mathcal{T}_k\}
=
\med\{(1-p)^{-1}\widetilde\theta_h(\x_i^T,W) : 
i \in \mathcal{T}_k\}
,
\end{align*}
where $\med\{\cdot\}$ denotes the sample median, and this coincides with the estimate obtained by minimising the loss function $\mathcal{L}_1$ in \eqref{e:EstFunGeneralMedian} with respect to $\theta\in\Theta$.
Moreover, the estimate in  \eqref{e:IndividualEstimatesMean} here takes the form 
\begin{align*}
\widehat{\theta}_j(\{(\x_i^T,\x_i^V)\}_{i=1}^k,p,W,h) 
=&
\frac{1}{\#\mathcal{T}_k}
\sum_{i \in \mathcal{T}_k} 
\widehat\theta_i^h
=
\frac{1}{(1-p)\#\mathcal{T}_k}
\sum_{i \in \mathcal{T}_k} 
\widetilde\theta_h(\x_i^T,W)
,
\quad j=2,3,
\end{align*}
and it coincides with the estimates obtained by minimising
any of the loss functions $\mathcal{L}_j$, $j=2,3$, in \eqref{e:EstFunGeneral} and \eqref{e:EstFunGeneralMean} with respect to $\theta\in\Theta$. 
In addition,  
\begin{align*}
\E[\widehat\theta_j(\{(\x_i^T,\x_i^V)\}_{i=1}^k,p,W,h)]
  =&
  \widetilde\theta_h(\x,W)
  =
\frac{\sum_{x\in\x\cap W} h(x)}{
\int_{W}h(u) \de u}, 
\quad j=2,3.
\end{align*}
Here we use the conventions that empty sums are 0 and $0/0=0$. We thus obtain the median and the mean of the estimates \eqref{e:IndividualEstimates}, and in the case of multinomial CV we set $p=1/k$  throughout.

If we instead let the indicator functions in \eqref{e:IndicatorInnovation} be given by $I_i=1$, $i=1,\ldots,k$, whereby $\widetilde{\mathcal{I}}_{\xi_{\theta}^1}^{h}(W;\x_i^T)
=\mathcal{I}_{\xi_{\theta}^1}^{h}(W;\x_i^T)$ and $\widehat\theta_i^h
=0$ if $\x_i^T=\emptyset$, then the results above remain the same but with $\mathcal{T}_k = \{1,\ldots,k\}$.

\end{thm}

In Theorem~\ref{thm:ConstantIntensity}, the first thing we note is that if the test function is given by a non-null constant then 
the expectation above reduces to $\widetilde\theta(\x,W)$. 
We emphasise that if we let $I_i
=\1\{1\leq\#\x_i^T\}$ then we remove any term $\widehat\theta_i^h$ where $\x_i^T=\emptyset$ from consideration, as opposed to letting $\widetilde{\mathcal{I}}_{\xi_{\theta}^1}^{h}(W;\x_i^T)
=\mathcal{I}_{\xi_{\theta}^1}^{h}(W;\x_i^T)$ where we instead include  $\widehat\theta_i^h=0$. 
Aside from the challenge of deriving closed form variance expressions when $I_i=\1\{1\leq\#\x_i^T\}$, in the case of $\widehat{\theta}_2=\widehat{\theta}_3$, it further turns out that the choice  $I_i=\1\{1\leq\#\x_i^T\}$ yields a higher variance than the the choice $I_i=1$, whereby $I_i=1$ will be the preferred choice. This is summarised in  Lemma~\ref{lemma:constant intensity majoration variance} below, which is proved in Section \ref{s:Proofs}.

\begin{lemma}
\label{lemma:constant intensity majoration variance}
  Let the situation be as in Theorem~\ref{thm:ConstantIntensity}. Then, 
  \begin{equation*}
  \Var\left(\frac{1}{\#\mathcal{T}_k}
\sum_{i \in \mathcal{T}_k} 
\widehat\theta_i^h\right)
\geq
  \Var\left( 
      \frac{1}{k}
      \sum_{i=1}^k 
      \widehat\theta_i^h
      \right)
     =\frac{p}{k(1-p)} \frac{\sum_{x\in \x\cap W}h(x)^2}
  {(\int_{W} h(u) \de u)^2},
  \end{equation*}
  where the right hand side tends to 0 if $k\to\infty$ or/and $p\to0$. 
\end{lemma}

Theorem~\ref{thm:ConstantIntensity} gives the expectation of one of the estimators, conditionally on $X\cap W=\x$. The unconditional case is treated in Lemma \ref{lemma:ConstantIntensityX} below, which is proved in Section \ref{s:Proofs}, and in particular it tells us that $\widehat\theta_2=\widehat\theta_3$ is unbiased for arbitrary $h$. 

\begin{lemma}
\label{lemma:ConstantIntensityX}
  Let the situation be as in Theorem~\ref{thm:ConstantIntensity} 
  and let $g_X^{(2)}(\cdot)$ be the pair correlation function of $X$.
  
  Irrespective of whether 
  the indicator functions in \eqref{e:IndicatorInnovation} are set to i) $I_i=\1\{1\leq\#\x_i^T\}$ 
or ii) 
$I_i=1$, $i=1,\ldots,k$, 
the expectation
  of 
  $\widehat\theta_j(\{(X_i^T,X_i^V)\}_{i=1}^k,p,W,h)$, $j=2,3$,  
  is given by 
  \begin{equation*}
  \E[\widehat\theta_j(\{(X_i^T,X_i^V)\}_{i=1}^k,p,W,h)]
  =\theta_0.
  \end{equation*}
  Moreover, the variance of $\widehat\theta_j(\{(X_i^T,X_i^V)\}_{i=1}^k,p,W,h)$ which corresponds to the choice i) is larger than or equal to the variance corresponding to the choice ii), and the variance corresponding to ii) satisfies  
 \begin{align*}
 \Var\left(\widehat\theta_j(\{(X_i^T,X_i^V)\}_{i=1}^k,p,W,h)\right)
 =&
   \Var\left( 
    \frac{1}{k}
    \sum_{i=1}^k 
    \widehat\theta((X_i^T,X_i^V),p,W,h)
    \right)
    \\
    =&
    \left(\frac{p}{(1-p)k}
     +
     1 \right) \theta_0
     \frac{ \int_W h(u)^2\de u }
     {(\int_{W} h(u) \de u)^2} 
     \\
     &+
     \theta_0^2
     \left( \frac{
     \int_W\int_W h(u_1)h(u_2)g_X^{(2)}(u_1,u_2)\de u_1\de u_2
     }
     {( \int_{W} h(u) \de u)^2} - 1\right).
 \end{align*}
 Recall that in the multinomial CV case we have $p=1/k$.
\end{lemma}

Since we have unbiasedness for any choice of $k$ and $p$ in the MCCV case, we would like to see how $k$ and $p$ should be chosen in order to minimise the variance. 
We start by noting that for $j=2,3$, when
$k\to\infty$ or/and $p\to0$ we have that
\begin{align*}
\Var(\widehat\theta_j(\{(X_i^T,X_i^V)\}_{i=1}^k,p,W,h))
  \to&
     \theta_0
     \frac{ \int_W h(u)^2\de u }
     {(\int_{W} h(u) \de u)^2} 
     \\
     &+
     \theta_0^2
     \left( \frac{
     \int_W\int_W h(u_1)h(u_2)g_X^{(2)}(u_1,u_2)\de u_1\de u_2
     }
     {( \int_{W} h(u) \de u)^2} - 1\right)
\end{align*}
monotonically in $k$ and $p$. Moreover, since by the law of large numbers we have that $\lim_{k\to\infty}\widehat\theta_j(\{(\x_i^T,\x_i^V)\}_{i=1}^k,p,W,h)=\widetilde\theta_h(\x,W)$ a.s.~for any point configuration $\x\in\X$, it also follows that 
$$
\lim_{k\to\infty}\widehat\theta_j(\{(X_i^T,X_i^V)\}_{i=1}^k,p,W,h)
\stackrel{a.s.}{=}
\widetilde\theta_h(X,W). 
$$
In other words, the $h$-weighted intensity estimator in  \eqref{e:hWeightedIntensityEst} corresponds to the unbiased minimum-variance case here. 

Turning to consistency under an increasing-domain regime, it follows that if $h$ and $g_X^{(2)}$ are such that
\begin{align*}
&\int_{W_n} \int_{W_n} h(u_1) h(u_2) g_X^{(2)}(u_1,u_2) \de u_1 \de u_2
\left/ \left( \int_{W_n} h(u) \de u \right)^{2}\to 1\right.,
\\
&\int_{W_n} h(u)^2 \de u \left/
\left( \int_{W_n} h(u) \de u \right)^{2}\to0\right.,
\end{align*}
for some
increasing sequence $W_n\subseteq S$, $n\geq1$, then, for any $k$ and
$p$ we have that 
$\E[(\widehat\theta_j(\{(X_i^T,X_i^V)\}_{i=1}^k,p,W,h)-\theta_0)^2]\to0$,
whereby $\widehat\theta_j(\{(X_i^T,X_i^V)\}_{i=1}^k,p,W,h)\to\theta_0$ in
probability when $n\to\infty$. 

The increasing-domain asymptotics above are satisfied e.g.~when $X$ is a homogeneous Poisson process and $h$ is constant; when $X$ is a homogeneous Poisson process, the pair correlation function is given by $g_X^{(2)}(\cdot)\equiv1$, which in turn implies that the second term of the variance in Lemma \ref{lemma:ConstantIntensityX} vanishes. 

Depending on the underlying point process $X$, finding variance-optimal choices for $h$ may be a challenge. 
Note first that Jensen's inequality tells us that $\int_W h(u)^2\de u \geq (\int_{W} h(u) \de u)^2$, with equality if $h$ is linear. This yields that the first of the variance terms is minimised when $h$ is linear. 
Hence, for a homogeneous Poisson process we have that the classical estimator \eqref{e:ClassicalEstimator}, which is obtained by letting $h$ be constant, is both variance-optimal and consistent. 
Dealing with the combination of the two variance terms simultaneously is a more delicate matter, which depends on the (unknown) dependence structure of the underlying point process $X$.

\subsubsection{Numerical evaluations}

Next, we evaluate our constant intensity estimators numerically, and we do so by considering the following models observed on $[0,1]^2=W\subseteq S=\R^2$ (see Section \ref{s:Models} for details):
\begin{itemize}
    \item A homogeneous Poisson process with intensity $\rho=250$.
    
    \item A homogeneous log-Gaussian Cox process (LGCP) with driving random field $\Lambda(u)=\exp\{Z(u)\}$, where $Z(u)$, $u\in W$, is a Gaussian random field with constant mean function $u\mapsto 3.5$ and exponential covariance function
    $(u,v)\mapsto \sigma^2\exp\{-r\|u-v\|_2\}$, $u,v\in W$, with $\sigma^2=4$ and $r=0.1$. Its intensity is given by $\rho=\e^{3.5 + 4/2}\approx245$.
    
    \item A homogeneous determinantal point process (DPP) with kernel given by $(u,v)\mapsto \sigma^2\exp\{-r\|u-v\|_2\}$, $u,v\in W$. We here set $(\sigma^2,r)=(250,50)$, whereby the intensity is given by $\rho=\sigma^2=250$.
\end{itemize}

Drawing inspiration from Section \ref{s:TestFunctions}, we will here consider the family of test functions given by $h^{\gamma}(u)=h^{\gamma}(u_1,u_2)=u_1^{\gamma}u_2^{\gamma}$, $u=(u_1,u_2)\in W$, $\gamma\in\R$. Note that $\gamma=0$ yields $h^{\gamma}(\cdot)\equiv1$, which in turn yields $\widehat{\theta}_2(\{(\x_i^T,\x_i^V)\}_{i=1}^k,p,W,h^{\gamma})=\widehat{\theta}_3(\{(\x_i^T,\x_i^V)\}_{i=1}^k,p,W,h^{\gamma})\approx\widetilde\theta(\x,W)=\#\x/|W|$ for large $k$, i.e.~the classical estimator; in practice, we consider each $\gamma$ in the sequence $\Theta_{\gamma}=\{-1,\ldots,-0.1,0,0.1,\ldots,1\}$. The classical estimator is also the estimator which we will compare our newly derived estimators to, so it is sufficient to look at which choice of $\gamma$ is optimal. To study the performance of the estimators, we report 
the mean squared error ($\mathrm{MSE}$), which theoretically is the same as the variance since the bias is 0 in the case of $\widehat\theta_2=\widehat\theta_3$. Due to the higher variance for multinomial CV, we here only consider MCCV and we let $p=0.1,0.2,\ldots,0.9$ and $k=400$. In Figure \ref{f:constantmse} we find the results for the three models.


\begin{figure}[!htpb]
 \centering
\includegraphics[scale=0.16]{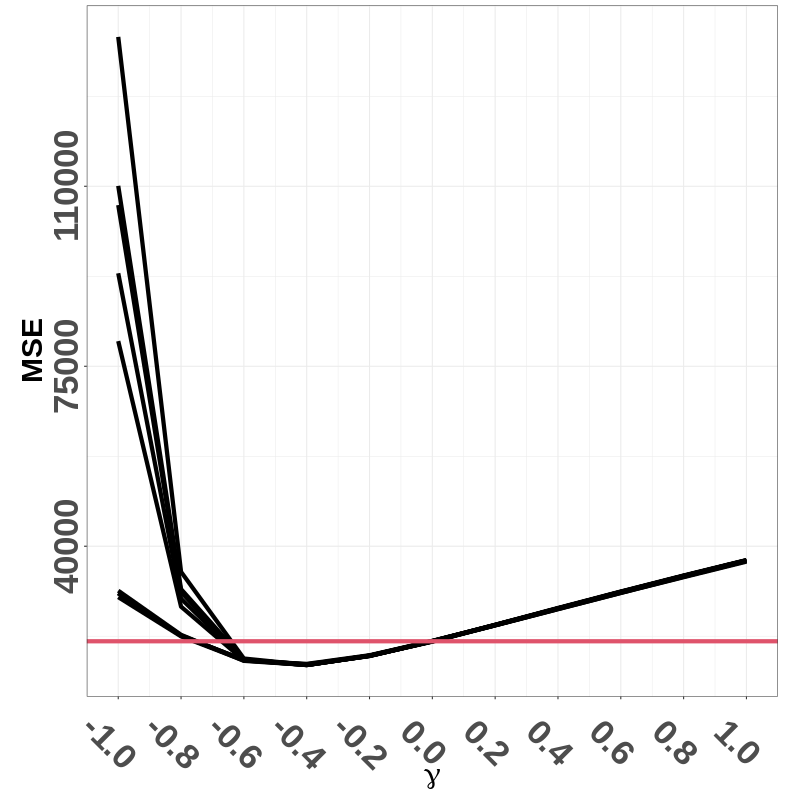}
\includegraphics[scale=0.16]{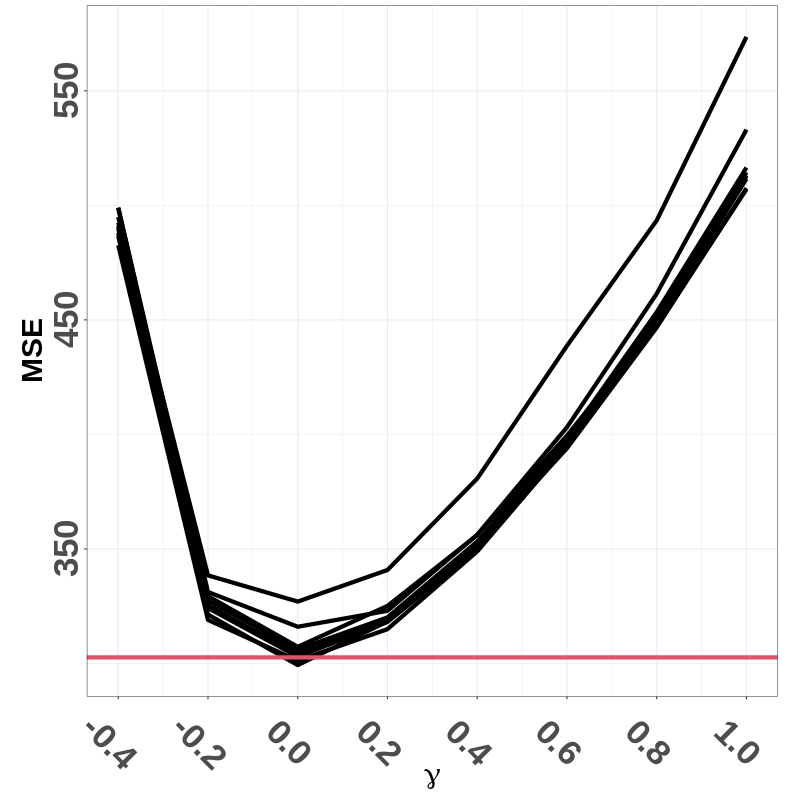}
\includegraphics[scale=0.16]{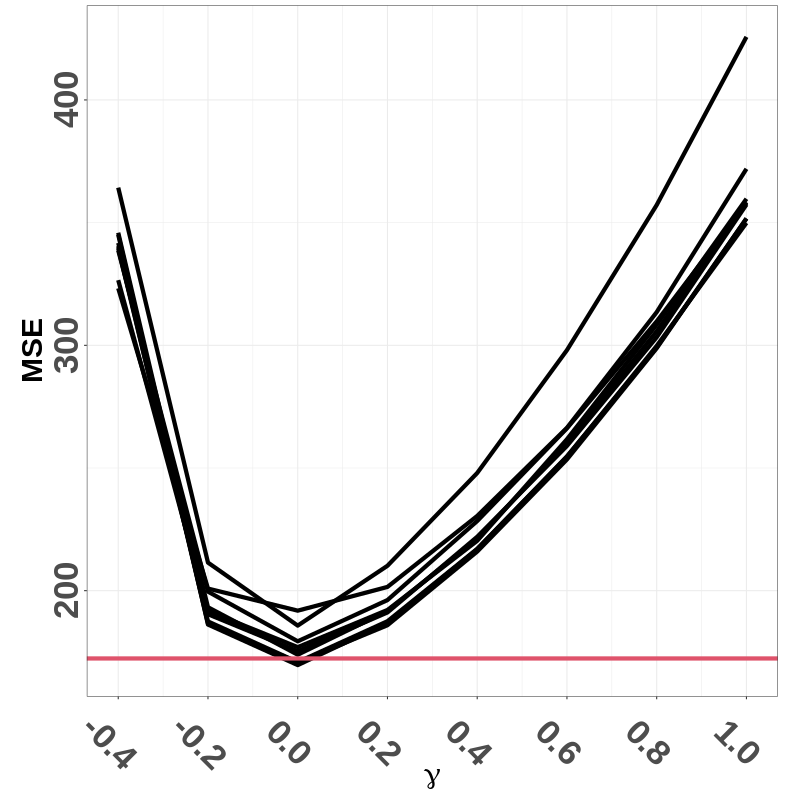}

\includegraphics[scale=0.16]{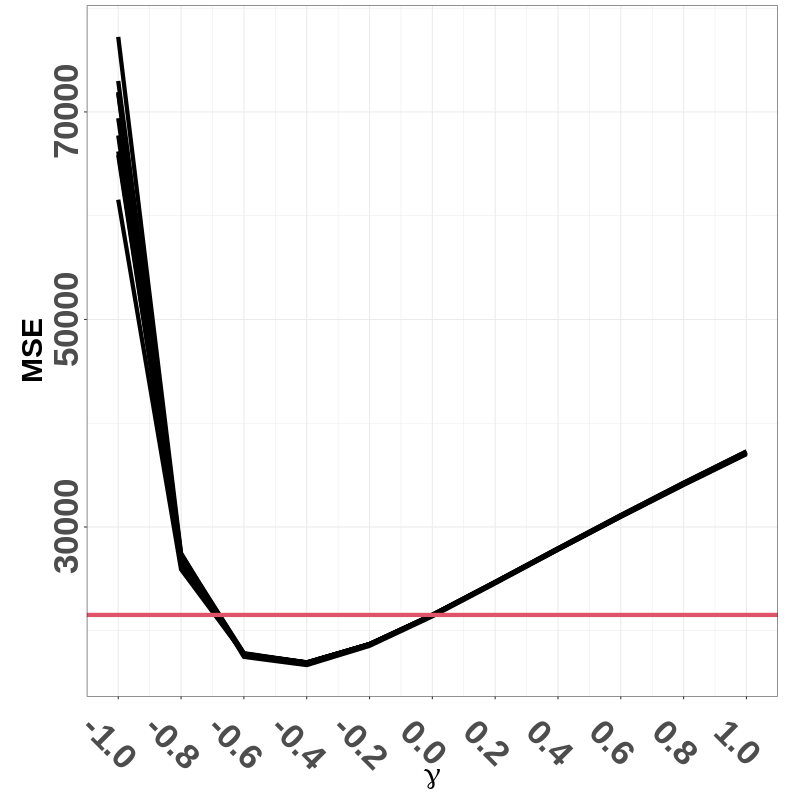}
\includegraphics[scale=0.16]{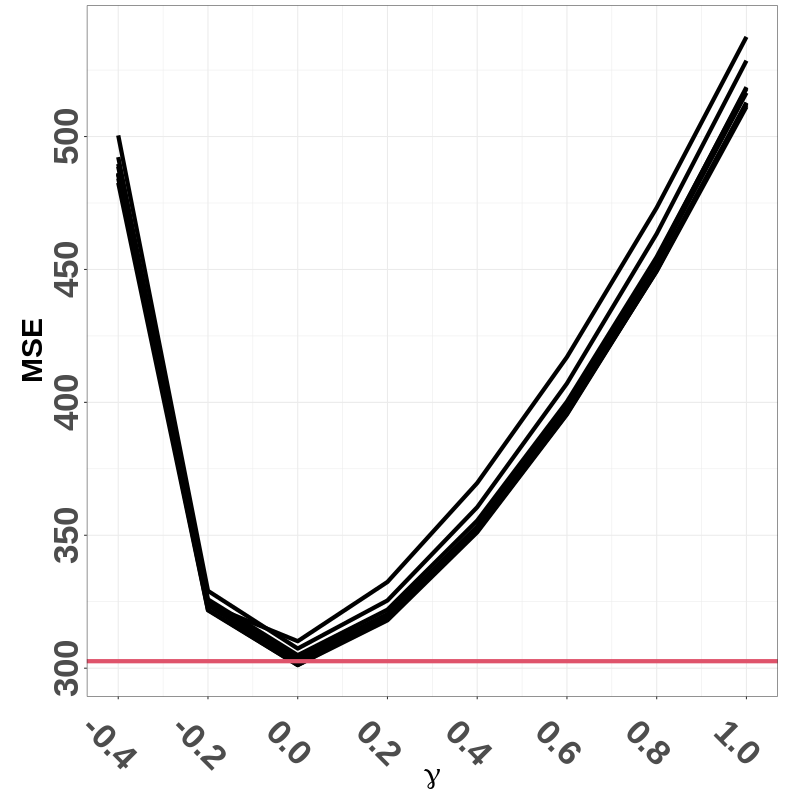}
\includegraphics[scale=0.16]{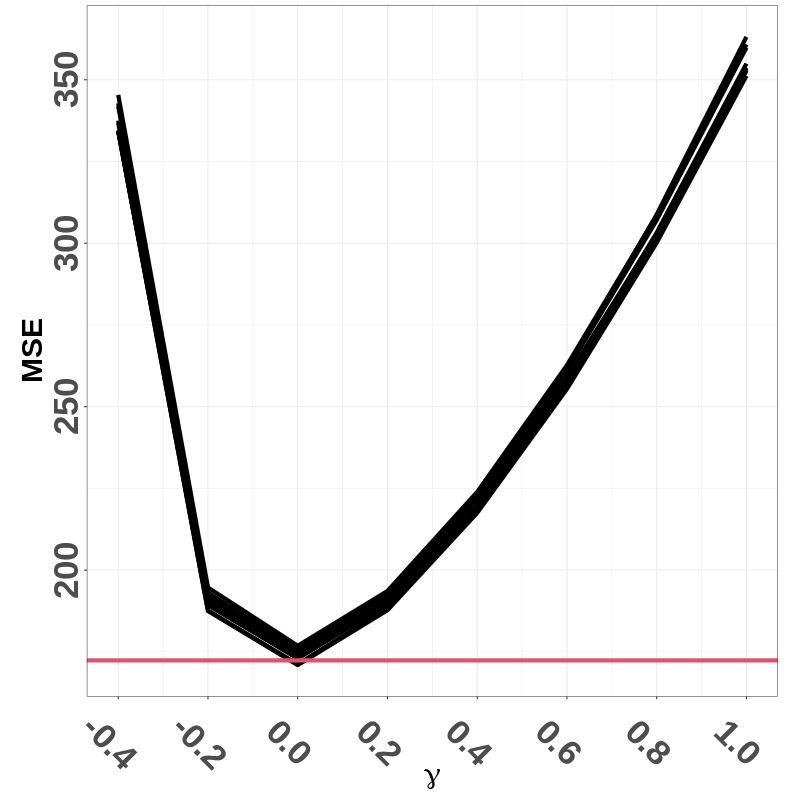}

\caption{
Mean squared error ($\mathrm{MSE}$). 
Rows: $\widehat\theta_1$ (upper) and $\widehat\theta_2=\widehat\theta_3$ (lower). 
Columns: 
LGCP (left), Poisson (middle) and DPP (right). 
The estimates have been obtained using MCCV with $p=0.1,0.2,\ldots,0.9$ and $k=400$. 
The different curves in each plot represent the different choices for $p$.
}
\label{f:constantmse}
\end{figure}

We see that regardless of whether we employ $\widehat\theta_1$ or  $\widehat\theta_2=\widehat\theta_3$, for the Poisson process and the DPP the $\mathrm{MSE}$-optimal choice is $\gamma=0$, i.e.~the classical estimator, and for the LGCP the $\mathrm{MSE}$-optimal choice is $\gamma=-0.4$. Morover, we see that $\widehat\theta_2=\widehat\theta_3$ yields a lower $\mathrm{MSE}$ than $\widehat\theta_1$. 
Note further that the tightness of the different curves in the case of $\widehat\theta_2=\widehat\theta_3$ reflects the variance asymptotics when $k$ tends to infinity; $\widehat\theta_1$ seems more sensitive to the choice of $p$ since the curves are more sparse in the plots  associated to $\widehat\theta_1$. We further note that the choice $\gamma\in[-0.2,0]$ puts us, relatively speaking and  $\mathrm{MSE}$-wise, within a short range of the optimal choice of $\gamma$ for each model. Since the gain of letting $\gamma$ be only slightly smaller than 0 in the LGCP (clustering) is quite large, compared to what we would lose in $\mathrm{MSE}$ in the case of the Poisson process and the DPP (regular), our general suggestion is to choose $\gamma$ slightly smaller than 0, say $\gamma\in[-0.2,0]$. However, if there are clear signs of regularity in an observed pattern then one should naturally choose the classical estimator, i.e.~$\gamma=0$; note that it is generally hard to characterise complete randomness (Poisson process) by means of visual inspection of only one point pattern. If one is convinced that the point pattern comes from a clustered model, then one should clearly choose $\gamma<0$, say, $\gamma$ between $-0.4$ and $0$.

\subsection{Papangelou conditional intensity fitting: hard-core processes}
\label{s:HardCore}
We next turn to one of the most common statistical settings, which is fitting a Papangelou conditional intensity model to data. 
Recall from Section \ref{s:PapangelouEstimation} that in the context of Papangelou conditional intensity-based fitting, our point process learning approach uses the bivariate innovations in \eqref{e:PapangelouInnovation1}.

We here choose to illustrate our approach in the context of Strauss processes (recall
Section \ref{s:ExpFamily}). Since the Poisson process case ($\eta=1$)
reduces to parametric intensity estimation, which has been covered in
Section \ref{s:ConstantIntensity}, to retain tractability, we here focus on the hard-core
process ($\eta=0$). More specifically, we consider a hard-core process $X$ in $W=S$ with Papangelou conditional intensity belonging to the family 
  $$
  \lambda_{\theta}(u;\x)
  =\beta_{\theta'}(u)\1\left\{u\notin\bigcup_{x\in\x}b(x,R)\right\}, 
  \qquad u\in W, 
  \x\in\X, 
  $$
  where $\theta=(\theta',R)\in\Theta=\Theta'\times(0,\infty)$,  $\Theta'\subseteq\R^{l'}$, $l'\geq1$, and the true parameters of $X$ are denoted by $\theta_0=(\theta_0',R_0)$. 
  
  It is noteworthy that the likelihood estimate of $R_0$ is given by \citep[Example 3.17]{VanLieshoutBook}
  $$
\bar R=\min_{x,y\in\x,x\neq y} d(x,y)
$$
  but, to the best of our knowledge, a closed form likelihood estimator for $\theta_0'$ is not available in the literature \citep{VanLieshoutBook}.

\subsubsection{Pseudolikelihood}
Before we proceed, we will have a brief look at the state of the art, namely pseudolikelihood estimation (recall Section \ref{s:PapangelouUnivariateInnovation}). Here we maximise
\begin{align*}
  \theta
  \mapsto&
     \sum_{x\in\x\cap W}
     \log\lambda_{\theta}(x;\x\setminus\{x\})
     -
     \int_{W}
     \lambda_{\theta}(u;\x)
     \de u
     \\
     =&
     \sum_{x\in\x\cap W}
     \log
     \beta_{\theta'}(x)
     +
     \sum_{x\in\x\cap W}
     \log\left(
     \1\left\{x\notin\bigcup_{y\in\x\setminus\{x\}}b(y,R)\right\}
     \right)
     -
     \int_{W\setminus\bigcup_{x\in\x}b(x,R)}
     \beta_{\theta'}(u)
     \de u
\end{align*}
and from the second term we see that the estimate of $R_0$ satisfies $\widehat R_{PL}\in(0,\bar R)$. By fixing $R\in(0,\bar R)$, the second term vanishes and to carry out the estimation of $\theta'$, we may set the gradient of the resulting expression to 0; we see that the model is not identifiable under the pseudolikelihood regime. Assuming that $\beta_{\theta'}(\cdot)$ is such that integration and differentiation may be interchanged, we thus solve
\begin{align*}
  \sum_{x\in\x\cap W}
     \frac{\nabla_{\theta'} \beta_{\theta'}(x)}{\beta_{\theta'}(x)}
     -
     \int_{W\setminus\bigcup_{x\in\x}b(x,R)}
     \nabla_{\theta'}\beta_{\theta'}(u)
     \de u=0\in\R^{l'},
     \quad R\in(0,\bar R),
\end{align*}
which is a vector of univariate innovations set to 0 (estimating equations). 
In particular, when  $\beta_{\theta'}(\cdot)$ is constant, i.e.~$\beta_{\theta'}(\cdot)\equiv\beta\in\Theta'=(0,\infty)$, where $\beta_{\theta_0'}(\cdot)=\beta_0\in\Theta'$, this reduces to setting 
\(
\beta\mapsto
  \sum_{x\in\x\cap W}
     \frac{1}{\beta}
     -
     |W\setminus\bigcup_{x\in\x}b(x,R)| 
\)
to 0, whereby we obtain the estimates 
\begin{align*}
  \widehat\beta_{PL}=&\widehat\beta_{PL}(\widehat R_{PL})=\frac{|W|}{|W\setminus\bigcup_{x\in\x}b(x,\widehat R_{PL})|}
  \widetilde\theta(\x,W),
  \quad
  \widehat R_{PL}\in(0,\bar R)
  ,
\end{align*}
and we see that the former is an adjusted version of the classical intensity estimator $\widetilde\theta$ in \eqref{e:ClassicalEstimator}. 
This estimate makes sense since having a smaller hard core range means that we can squeeze in more points into $W$; note that $\widehat\beta_{PL}(\widehat R_{PL})$ decreases as $\widehat R_{PL}$ increases. 
Pseudolikelihood estimation for a hard core model in $\R^2$ may be practically carried out by means of the function \verb|ppm| in the \textsf{R} package \textsf{spatstat} \citep{BRT15}; the function \verb|ppm| uses the choice $\widehat R_{PL}=\bar R\#\x/(\#\x+1)$. 

\subsubsection{Point process learning}

Turning to our point process learning approach, we here consider a class of test functions 
$h_{\theta}(\cdot)=f(p(1-p)^{-1}\lambda_{\theta}(\cdot))$, where $f:\R\to\R$ is such that $\lim_{x\to0}|f(x)|=\infty$. 
This 
includes e.g.~$f(x)=x^{-\gamma}$, $\gamma>0$. 

Given $\x_i^T,\x_i^V\neq\emptyset$, the innovations in  \eqref{e:PapangelouInnovation1} here become 
\begin{align*}
    \mathcal{I}_{\xi_{\theta}^1}^{h_{\theta}}(W;\x_i^V,\x_i^T)
    =&
       \sum_{x\in\x_i^V\cap W}
       f\left(
       \frac{p\lambda_{\theta}(x;\x_i^T)}{1-p}\right)
       -
       \frac{p}{1-p}
       \int_{W}
       f\left(
       \frac{p\lambda_{\theta}(u;\x_i^T)}{1-p}\right)
       \lambda_{\theta}(u;\x_i^T)
       \de u
    \\
    =&
       \sum_{x\in\x_i^V\cap W}
       f\left(
       \frac{p\beta(x)\1\{x\notin\bigcup_{y\in\x_i^T}b(y,R)\}}{1-p}\right)
    \\
     &-
       \frac{p}{1-p}
       \int_{W\setminus\bigcup_{x\in\x_i^T}b(x,R)}
       \beta(u)
       f\left(
       \frac{p\beta(u)\1\{u\notin\bigcup_{x\in\x_i^T}b(x,R)\}}{1-p}\right)
       \de u
       ,
\end{align*}
where for the right hand side to be finite, we need that $\1\{x\notin\bigcup_{y\in\x_i^T}b(y,R)\}=1$ for each $x\in\x_i^V$, i.e.~$\x_i^V\cap\bigcup_{y\in\x_i^T}b(y,R)=\emptyset$, which is to say that $R$-balls around the points of $\x_i^T$ cannot contain any points of $\x_i^V$. Since $(W\setminus\bigcup_{x\in\x_i^T}b(x,R))\cap\bigcup_{x\in\x_i^T}b(x,R)=\emptyset$, the integral is finite. 
Hence, we see that the innovations are finite only if $R$ belongs to
\begin{align*}
\mathcal{R}
=
\mathcal{R}_p(\{(\x_i^V,\x_i^T)\}_{i=1}^{k})
=&\left\{r>0:\x_i^V\cap\bigcup_{x\in\x_i^T}b(x,r)=\emptyset \text{ for
    all }i\in\mathcal{T}_k\right\}
    \\
    =&\bigcap_{i\in\mathcal{T}_k}\left\{r>0:\x_i^V\cap\bigcup_{x\in\x_i^T}b(x,r)=\emptyset\right\}
    ,
\end{align*}
where 
\begin{align*}
\mathcal{T}_k=&\{i\in\{1,\ldots,k\}:I_i=1\}=\{i\in\{1,\ldots,k\}:1\leq\#\x_i^T\leq\#\x-1\}.
\end{align*}
In other words, the estimate $\widehat R=\widehat R_p(\{(\x_i^V,\x_i^T)\}_{i=1}^{k})$ of the interaction/hard-core range $R_0$ belongs to $\mathcal{R}$ and in the MCCV case we obtain
that
$$
\lim_{k\to\infty}\mathcal{R}_p(\{(\x_i^V,\x_i^T)\}_{i=1}^{k})=(0,\bar R),
$$
i.e.~the upper bound is given by the likelihood estimate of $R_0$.  This suggests a data-driven lower bound for $k$ in the MCCV case: 
sequentially increase $k$ at least until $\mathcal{R}_p(\{(\x_i^V,\x_i^T)\}_{i=1}^{k})=(0,\bar R)$. 

By imposing the restriction that 
$\theta=(\theta',R)\in\Theta'\times\mathcal{R}$, the innovations 
\eqref{e:PapangelouInnovation1} reduce to 
\begin{align}
\label{e:ReducedInnovationHardCore}
\mathcal{I}_{\xi_{\theta}^1}^{h_{\theta}}(W;\x_i^V,\x_i^T)
=&
     \mathcal{I}_{p(1-p)^{-1}\lambda_{\theta}}^{f(\lambda_{\theta}(\cdot))}(W;\x_i^V,\x_i^T)
     \\
     =&
  \sum_{x\in\x_i^V\cap W} f\left(
    \frac{p\beta_{\theta'}(x)}{1-p}\right) - 
  \int_{W\setminus\bigcup_{x\in\x_i^T}b(x,R)} f\left(
    \frac{p\beta_{\theta'}(u)}{1-p}\right) 
    \frac{p\beta_{\theta'}(u)}{1-p} \de u
    ,
    \nonumber
\end{align}
i.e. the loss function for estimating $\theta_0'$ is given by a combination of  $\mathcal{I}_{p(1-p)^{-1}\lambda_{\theta}}^{f(\lambda_{\theta}(\cdot))}(W;\x_i^V,\x_i^T)$, $\theta=(\theta',R)\in\Theta'\times\mathcal{R}$, $i\in\mathcal{T}_k$. 
Since $\mathcal{I}_{\xi_{\theta_1}^1}^{h_{\theta_1}}(W;\x_i^V,\x_i^T)=\mathcal{I}_{\xi_{\theta_2}^1}^{h_{\theta_2}}(W;\x_i^V,\x_i^T)$ does not imply that $\theta_1=\theta_2$ (if $\theta_1=(\theta',R_1)$ and $\theta_2=(\theta',R_2)$ these two innovations are the same
for any $R_1,R_2\in\mathcal{R}$), the loss function $\theta\mapsto\mathcal{I}_{\xi_{\theta}^1}^{h_{\theta}}(W;\x_i^V,\x_i^T)$ is not identifiable for a fixed  $i\in\mathcal{T}_k$. One would typically deal with this by fixing a point estimate $\widehat R$ of $R_0$, most naturally $\widehat R=\widehat R_p(\{(\x_i^V,\x_i^T)\}_{i=1}^{k})=\sup\mathcal{R}$, and then proceed by exploiting the innovations in  \eqref{e:ReducedInnovationHardCore} for the estimation of $\theta'$. However, we have seen that, numerically, this is not necessary when employing any of the loss functions $\mathcal{L}_1$, $\mathcal{L}_2$ and $\mathcal{L}_3$ in \eqref{e:EstFunGeneralMedian}, \eqref{e:EstFunGeneral} and \eqref{e:EstFunGeneralMean}, i.e.~we may let both $R\in\mathcal{R}$ and $\theta'\in\Theta'$ be free parameters to be estimated. In other words, the component-wise unidentifiability seems to not spill over on the loss functions. 

\begin{remark}
Regarding the test function family considered here, note in contrast e.g.~that a test function for which $\lim_{x\to0}f(x)=0$ will instead
minimise the squared/absolute innovations when the hard-core
constraint is violated, which is clearly not what we want here. 
\end{remark}



We next turn to the special case where $\beta_{\theta'}(\cdot)\equiv\beta\in\Theta'=(0,\infty)$ and $\beta_{\theta_0'}(\cdot)=\beta_0\in\Theta'$. Here the loss function terms become 
\begin{align}
\label{e:ReducedInnovationHardCoreConstantBeta}
\widetilde{\mathcal{I}}_{\xi_{\theta}^1}^{h_{\theta}}(W;\x_i^V,\x_i^T)
=&
     \widetilde{\mathcal{I}}_{p(1-p)^{-1}\lambda_{\theta}}^{f(\lambda_{\theta}(\cdot))}(W;\x_i^V,\x_i^T)
     \\
     =&
    I_i
    f\left(
    \frac{p\beta}{1-p}\right)
    \left(\#\x_i^V 
    - 
    \frac{p\beta}{1-p}\left|W\setminus\bigcup_{x\in\x_i^T}b(x,R)\right|\right),
    \quad R\in\mathcal{R}.
    \nonumber
\end{align}
If we impose that $|f(x)|>0$, $x>0$, which e.g.~holds for $f(x)=x^{-\gamma}$, $\gamma>0$, then this is 0 if either $I_i=\1\{1\leq\#\x_i^T\leq\#\x-1\}=0$ or if $\beta$ is given by 
\[
\widehat\beta_i|R
=
\frac{(1-p)\#\x_i^V}
{p|W\setminus\bigcup_{x\in\x_i^T}b(x,R)|}
=
\frac{|W|}{|W\setminus\bigcup_{x\in\x_i^T}b(x,R)|/(1-p)}
\overbrace{
\frac{\widetilde\theta(\x_i^V,W)}{p}
}^{\approx\widetilde\theta(\x,W)}
,
\quad R\in\mathcal{R},
\]
which essentially is equivalent to a CV-based version of $\widehat\beta_{PL}(R)$, $R\in\mathcal{R}$. 
It should be emphasised that $|W\setminus\bigcup_{x\in\x_i^T}b(x,R)|/(1-p)$ is not linear in $p$ (from a distributional point of view) so we expect the choice of $p$ to be of significance here. 
Recalling Theorem \ref{thm:ConstantIntensity}, we see that the estimate obtained by minimising either $\mathcal{L}_1$ in \eqref{e:EstFunGeneralMedian} or $\mathcal{L}_2$ in \eqref{e:EstFunGeneral} here tries to find a pair $(\beta,R)$ such that, on average (in a median sense in the former case and in a mean sense in the latter case), we estimate $\#\x_i^V$ as well as possible by means of $\widehat\beta_i|R$, $i=1,\ldots,k$.

\begin{remark}

As an alternative to optimising with respect to $\beta$ and $R$ jointly, one could consider the profile alternative where one fixes $R\in\mathcal{R}$, e.g.~$R=\sup\mathcal{R}$. When this is the case, minimising  \eqref{e:EstFunGeneralMedian} yields the estimate 
\(
\widehat\beta 
=
\med\{\widehat\beta_i|R : 
i \in \mathcal{T}_k\}; 
\) 
cf.~the estimate in expression \eqref{e:IndividualEstimatesMedian}. 
Similarly, the estimates obtained using $\mathcal{L}_2$ and $\mathcal{L}_3$ in \eqref{e:EstFunGeneral} and \eqref{e:EstFunGeneralMean} are here given by 
\(
\widehat\beta 
=
\frac{1}{\#\mathcal{T}_k}
\sum_{i\in\mathcal{T}_k}
\widehat\beta_i|R; 
\) 
cf.~the estimate in  \eqref{e:IndividualEstimatesMean}.

\end{remark}

\subsubsection{Numerical evaluations}

We next evaluate our approach numerically in the case where $\beta_{\theta'}(\cdot)\equiv\beta\in\Theta'=(0,\infty)$ and $\beta_{\theta_0'}(\cdot)=\beta_0\in\Theta'$. More specifically, we consider 100 realisations of a hard core model on $W=[0,1]^2$ with parameters $R_0=0.05$ and $\beta_0=100$; this particular choice of parameters, which give rise to an average point count of 58.51, was made completely arbitrarily.

We here consider the loss functions $\mathcal{L}_1$, $\mathcal{L}_2$ and $\mathcal{L}_3$ in \eqref{e:EstFunGeneralMedian}, \eqref{e:EstFunGeneral} and \eqref{e:EstFunGeneralMean}, in combination with \eqref{e:ReducedInnovationHardCoreConstantBeta}, $I_i=\1\{1\leq\#\x_i^T\leq\#\x-1\}$ and MCCV, where $k=400$ and $p \in \{ 0.05, 0.07, 0.1, 0.15, 0.2, 0.3, 0.5, 0.7, 0.9 \}$. We further focus on the test functions $f(x)=1/x$ and $f(x)=1/\sqrt{x}$, $x\in\R$. 
To compare our approaches to the state of the art, we additionally carry out pseudolikelihood estimation. 

In Figure \ref{f:Strauss} we report 
the absolute bias, the variance and 
the mean squared error ($\mathrm{MSE}$) for each estimator; 
as anticipated, the choice of $p$ plays a significant role here. 
All estimators, irrespective of the choice of $p$, yield estimated biases of $R_0$ which are close to 0; this is also the the case for the pseudolikelihood estimator. 
More interestingly, our point process learning approach does not require that a fixed estimate of $R_0$ is plugged into \eqref{e:ReducedInnovationHardCoreConstantBeta} in order to obtain a good estimate of $\beta_0$. In the case of $\beta_0$, compared to the pseudolikelihood estimator ($\mathrm{MSE}=410.26$), we can find choices for $p$ such that we always achieve either a smaller bias ($p=0.07$ with $f(x)=1/\sqrt{x}$, using either $\mathcal{L}_1$ or $\mathcal{L}_2$, and $p=0.1$ with $f(x)=1/x$, using $\mathcal{L}_2$), variance (essentially any $p$; note that it decreases with $p$) or $\mathrm{MSE}$. The superior performance in $\mathrm{MSE}$ with respect to pseudolikelihood estimation holds for $f(x)=1/\sqrt{x}$ with $p\in\{0.05,0.07,0.1\}$ and either of $\mathcal{L}_1$ or $\mathcal{L}_2$; it is minimised with $p=0.07$ and $\mathcal{L}_2$, giving $\mathrm{MSE}=225.14$. Similarly, when
$f(x)=1/x$ this is true for  $p\in\{0.07,0.1,0.15\}$ and either of $\mathcal{L}_1$ or $\mathcal{L}_2$; it is minimised with $p=0.1$ and $\mathcal{L}_2$, yielding $\mathrm{MSE}=202.23$. Hence, it seems that $\mathcal{L}_1$ and $\mathcal{L}_2$, which perform next to identically, are the preferred loss functions here. Finally, we have observed that increasing $\beta_0$, while keeping $R_0=0.05$ fixed, results in the $\mathrm{MSE}$-optimal $p$ being slightly smaller; $p=0.07$ in combination with $\mathcal{L}_2$ still performs better than pseudolikelihood estimation in terms of $\mathrm{MSE}$.

\begin{figure}[!hbtp]
    \centering
\includegraphics[scale=0.15]{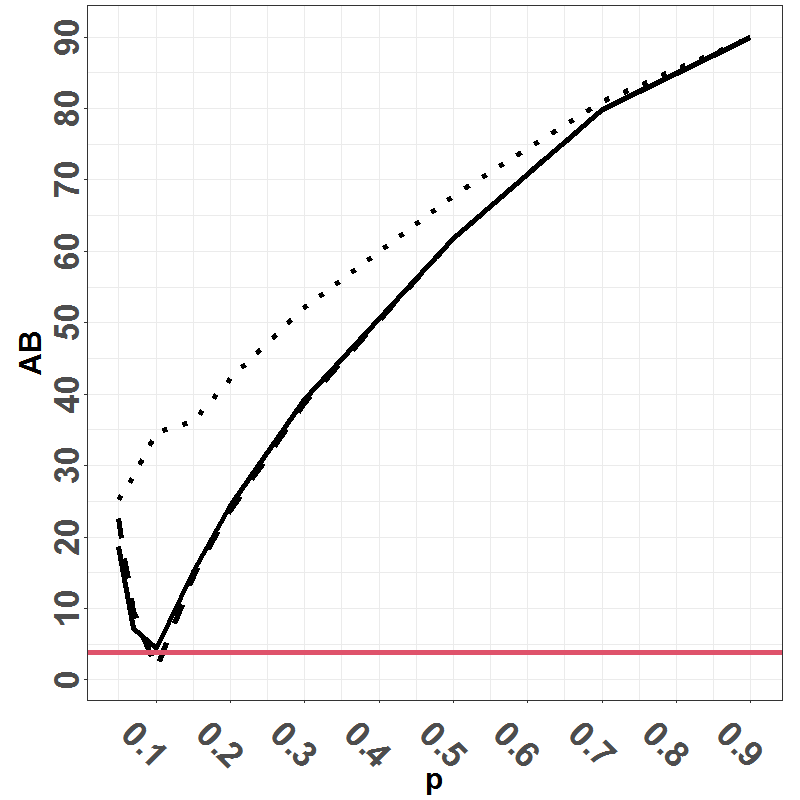}
\includegraphics[scale=0.15]{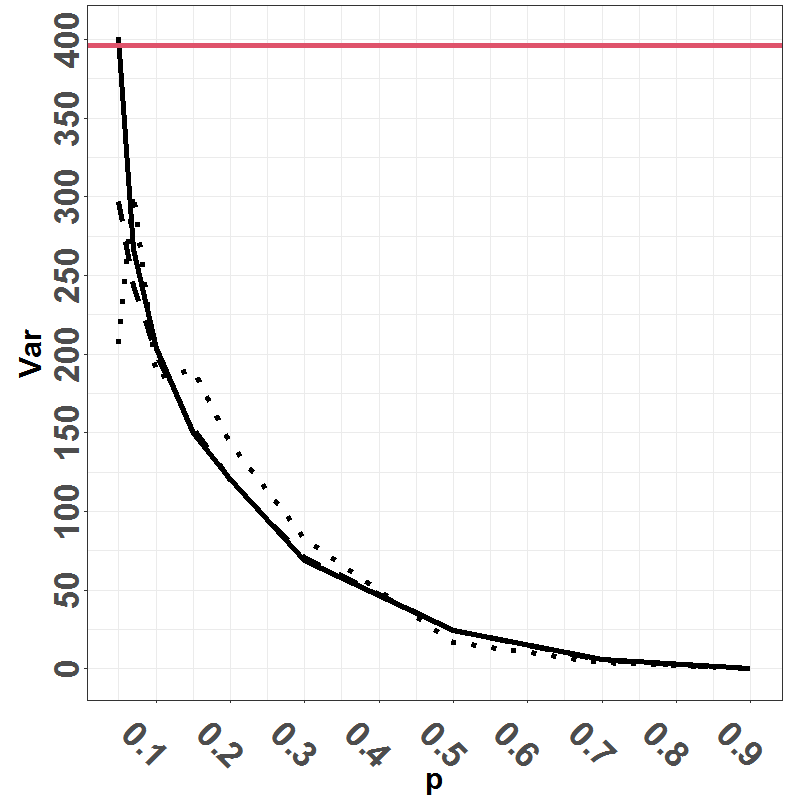}
\includegraphics[scale=0.15]{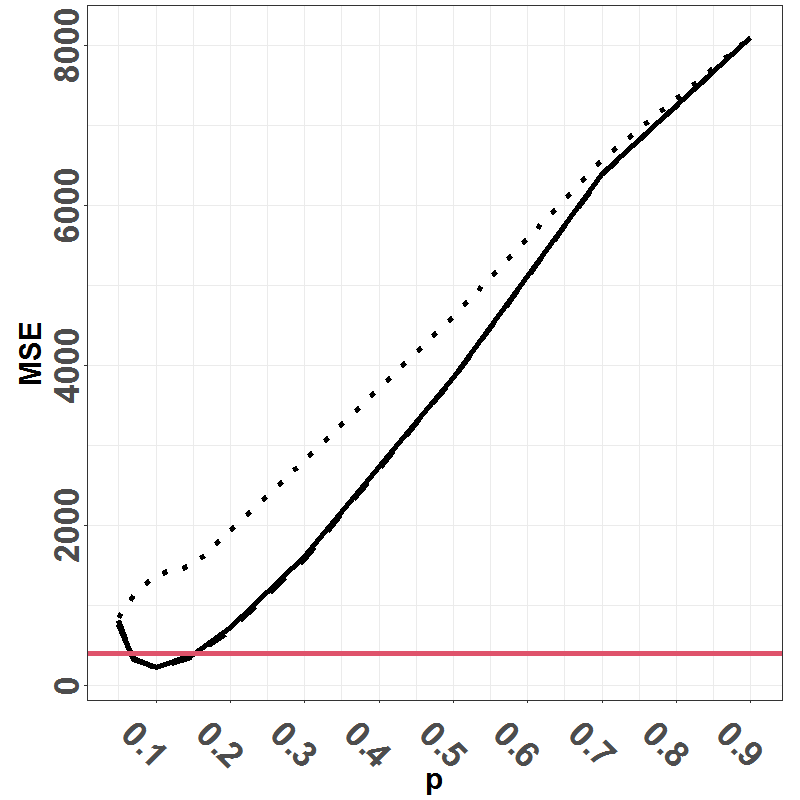}

\includegraphics[scale=0.15]{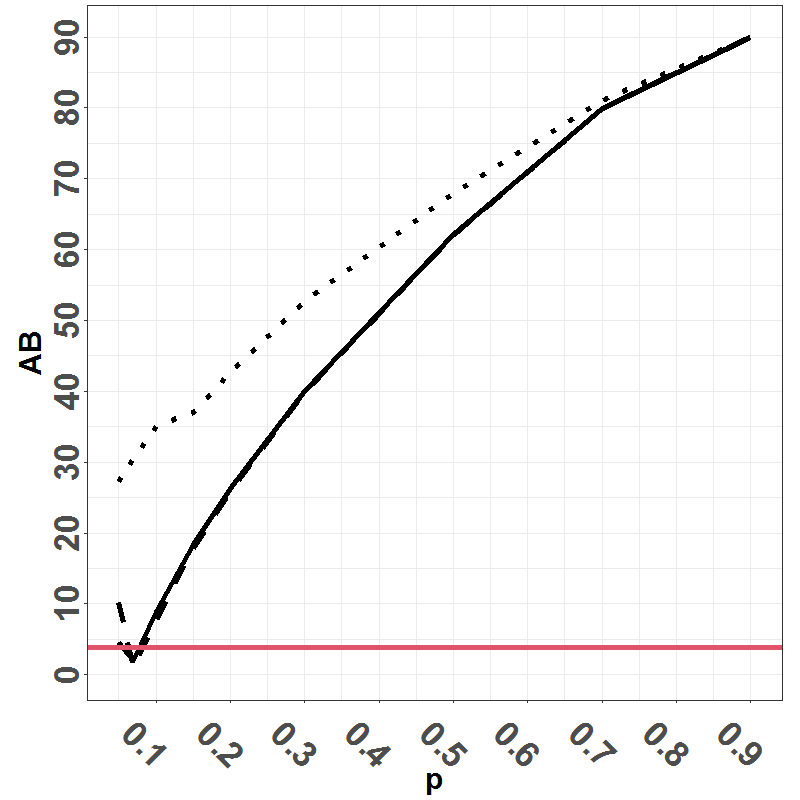}
\includegraphics[scale=0.15]{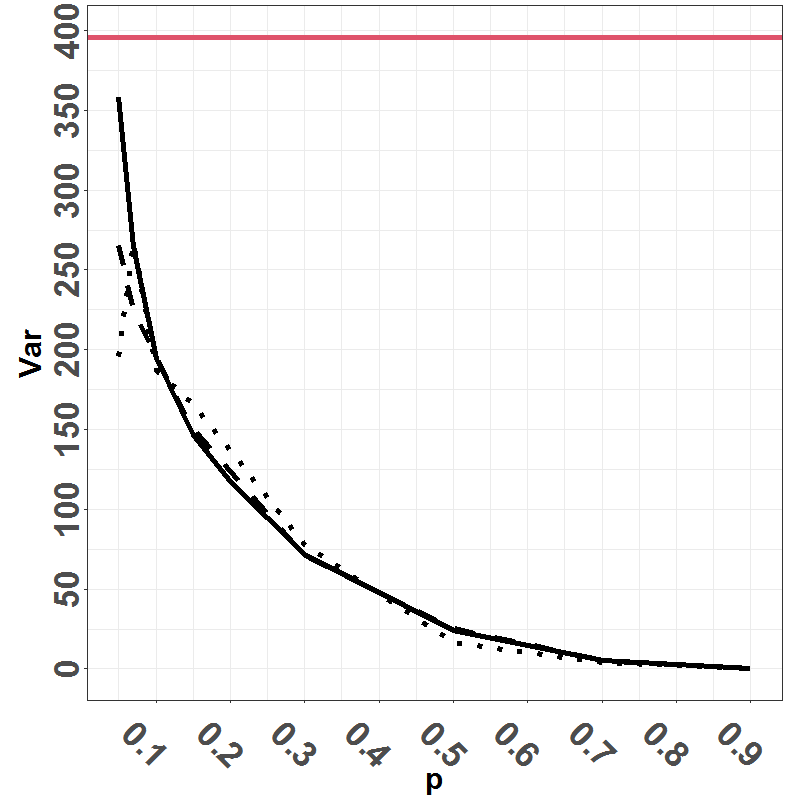}
\includegraphics[scale=0.15]{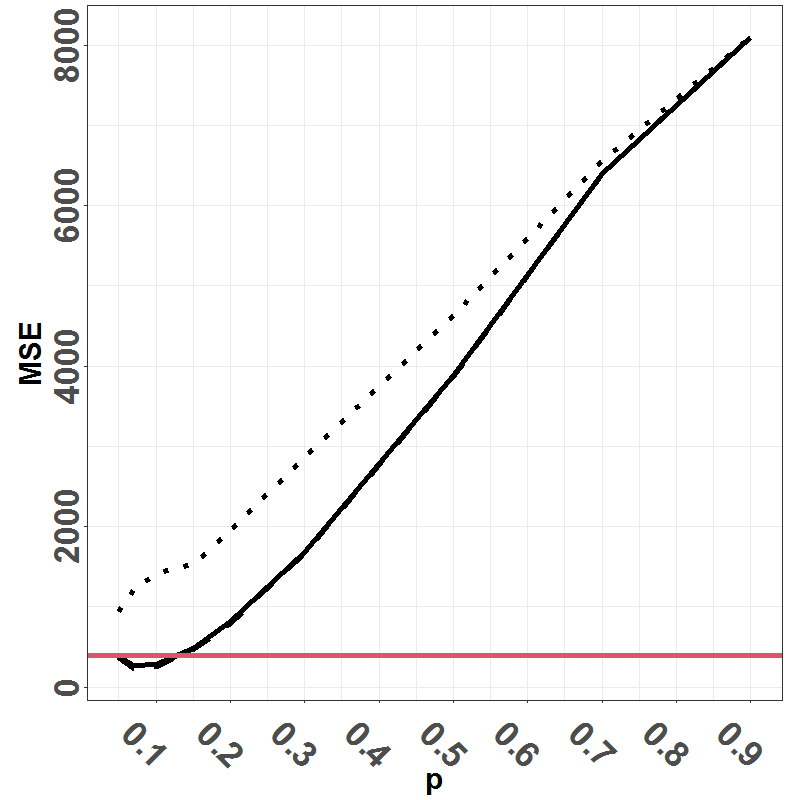}
    \caption{
    Parameter estimation results based on 100 realisations of a hard core model on $W=[0,1]^2$ with parameters $R_0=0.05$ and $\beta_0=100$. Loss functions:  $\mathcal{L}_1$ (solid curves, --), $\mathcal{L}_2$ (dashed curves, $---$) and $\mathcal{L}_3$ (dotted curves, $\cdots$), in combination with MCCV, where $p \in \{ 0.05, 0.07, 0.1, 0.15, 0.2, 0.3, 0.5, 0.7, 0.9 \}$ and $k=400$, and the test functions $f(x)=1/x$ (top row) and $f(x)=1/\sqrt{x} $ (bottom row). 
    The red lines represent pseudolikelihood estimation. 
    From left to right, the plots correspond to absolute bias ($\mathrm{AB}$), variance ($\mathrm{Var}$), and mean squared error ($\mathrm{MSE}$), respectively. 
    }
    \label{f:Strauss}
\end{figure}

\subsection{Non-parametric kernel intensity estimation}
\label{s:Kernel}

We next turn to non-parametric intensity estimation. Recalling Section \ref{s:NonParametrics}, we here consider a non-parametric intensity estimator
$\widehat\rho_{\theta}(u,\y)$, $u\in W$, $\y\in\X$, $\theta\in\Theta$, and our aim is to choose the tuning parameter $\theta$ in some optimal way, using an observed point pattern $\x\subseteq W$. 

From Section \ref{s:NonParametricsCV} we have that the bivariate innovations here are given by
\begin{align}
\label{e:InnovationKernel}
      \mathcal{I}_{\xi_{\theta}^1}^{h_{\theta}}(W;\x_i^V,\x_i^T)
        =&
        \mathcal{I}_{p(1-p)^{-1}\widehat\rho_{\theta}}^{h_{\theta}}(W;\x_i^V,\x_i^T)
        \\
        =&
        \sum_{x\in\x_i^V\cap W}
        h_{\theta}(x;\x_i^T)
      -
        \frac{p}{1-p}
        \int_{W}
        h_{\theta}(u;\x_i^T)
        \widehat\rho_{\theta}(u;\x_i^T)
        \de u
        ,
        \quad i=1,\ldots,k,
        \nonumber
\end{align}
and the tuning parameter selection/estimation may be carried out using
these innovations in e.g.~any of the loss functions
\eqref{e:EstFunGeneralMedian}, \eqref{e:EstFunGeneral} and
\eqref{e:EstFunGeneralMean}.  Aside from choosing some suitable
estimator $\widehat\rho_{\theta}$ and CV parameters $p$ and $k$, the
crucial choice to be made here is clearly the test function. Motivated
by \citet{cronie2018bandwidth}, we here consider the form
$h_{\theta}(u,\x_i^T)=f(p\widehat\rho_{\theta}(u;\x_i^T)/(1-p))$, for
some test function $f\geq0$. 
This results in 
\begin{align}
  \label{e:NonParam}
  \widetilde{\mathcal{I}}_{\xi_{\theta}^1}^{h_{\theta}}(W;\x_i^V,\x_i^T)
  =&
     I_i
     \left(
     \sum_{x\in \x_i^V}f\left(\frac{\widehat\rho_{\theta}(x,\x_i^T)p}{1-p}\right)
     -
     \int_W f\left(\frac{\widehat\rho_{\theta}(u,\x_i^T)p}{1-p}\right)
     \frac{\widehat\rho_{\theta}(u,\x_i^T)p}{1-p}
     \de u
     \right)
     ,
\end{align}
where $I_i=\1\{1\leq\#\x_i^V\leq \#\x-1\}$, $i=1,\ldots,k$.

\subsubsection{Kernel intensity estimation}
In what follows, we will focus on kernel estimation for point processes in $\R^d$; see Section \ref{s:NonParametrics} for details.  More specifically, we will focus on optimal selection of
the bandwidth $\theta\in\Theta=(0,\infty)$, given some point pattern 
$\x\subseteq W\subseteq\R^d$. 
Our main objective here is to study what effect our point process learning approach has on bandwidth
selection and to do so we will use the approach of  \citet{cronie2018bandwidth} as benchmark since, generally speaking, it outperformed its
predecessors/competitors. 

\begin{remark}
It should be emphasised that we may normalise intensity estimators to become density estimators, so the conclusions below apply equally well to density estimation (in the context of generalised random sampling).  
\end{remark}

Recall from Section
\ref{s:NonParametrics} that \citet{cronie2018bandwidth} conjectured
that \eqref{e:CvL} with $f(x)=1/x$ results in \eqref{e:CvL} being a
convex function, when using a Gaussian kernel and no edge
correction, i.e.~$w_{\theta}(u,x)\equiv1$. Assuming that their conjecture is
true, also $\mathcal{L}_2(\theta)$ in combination with \eqref{e:NonParam} would be convex under these conditions
and consequently there would be a global minimiser. 
Hence, we here consider a Gaussian kernel estimator, using no edge
correction when we select the bandwidth, and let the test function be
given by
$$
h_{\theta}(u,\x_i^T)
=
f(p\widehat\rho_{\theta}(u;\x_i^T)/(1-p))
=
1/(p\widehat\rho_{\theta}(u;\x_i^T)/(1-p)). 
$$
This
puts us in a setting which is equivalent to that in
\citet{cronie2018bandwidth} and, consequently, we may explicitly
analyse how our point process learning framework improves the performance of the approach of
\citet{cronie2018bandwidth} by simply comparing the performances of
the two approaches. We further note that the above choice of test function
is a natural and convenient one for our purposes but 
many other choices may perform equally well/better. Taking the observations in Section \ref{s:ConstantIntensity} into account, we will also look closer at the choice
$$
h_{\theta}(u,\x_i^T)
=
f(p\widehat\rho_{\theta}(u;\x_i^T)/(1-p))
=
1/\sqrt{p\widehat\rho_{\theta}(u;\x_i^T)/(1-p)}. 
$$

\begin{remark}
  As an aside, it should be noted that the Poisson process likelihood
  leave-one-out CV approach (see Section \ref{s:NonParametrics} for details) can be straightforwardly altered such that
  any of our CV-based approaches would replace the leave-one-out
  part in it. One would maximise
  \( \theta \mapsto \frac{1}{k}
     \sum_{i=1}^k
     I_i( \sum_{x\in
    \x_i^V}\log(\widetilde\rho_{\theta}(x,\x_i^T)) - \int_W
  \widetilde\rho_{\theta}(u,\x_i^T)\de u ) \) to select the bandwidth.
\end{remark}


To evaluate the proposed setup numerically, we let the study region be
given by
$W=[0,1]^2$ and we consider 100 realisations of each of three different models, which
constitute a subset of the models evaluated in \citet{cronie2018bandwidth}. 
The models considered, which represent aggregation, complete randomness 
and inhibition are the following.
\begin{itemize}
    
  \item Recalling Section \ref{s:Cox}, we here consider a log-Gaussian Cox process (LGCP) with random
    intensity $\Lambda(u)=\exp\{Z(u)\}$,
    $u=(u_1,u_2)\in W=[0,1]^2$, where $Z$ is a Gaussian
    random field with mean function $u=(u_1,u_2)\mapsto 10+80u_1$,
    $u\in W$, and exponential covariance function
    $(u,v)\mapsto \sigma^2\exp\{-r\|u-v\|_2\}$, $u,v\in W$, with $(\sigma^2,r)=(2\log5,50)$. 
    Consequently, the intensity function is given
    by $\rho(u)=(10+80u_1)\e^{\sigma^2/2}$, $u\in W$, and the expected number of points in $W$ is 250. 
    

    \item We consider a linear trend Poisson process (recall Section \ref{s:Poisson}) with intensity function $\rho(u)=\rho(u_1,u_2)=10+a u_1$, $u\in W=[0,1]^2$, where $a=480$, so that the expected number of points is given by 250.

    
  \item We consider a homogeneous determinantal point process (DPP) with kernel
    $(u,v)\mapsto \sigma^2\exp\{-r\|u-v\|_2\}$,
    $u,v\in W=[0,1]^2$, where $(\sigma^2,\beta)=(250,50)$; recall Section \ref{s:DPP}. 
    By
    applying independent thinning to it, using the retention
    probability function $u=(u_1,u_2)\mapsto(10+80u_1)/90$, $u\in W$,
    we obtain an inhomogeneous determinantal point process with
    intensity function $\rho(u)=\rho(u_1,u_2)=\sigma ^2(10+80u_1)/90$, $u\in W$, and expected total point count given by 138.9. 
    
    
\end{itemize}

As noted above, we let the kernel $\kappa$ be a Gaussian
kernel and 
we let $w_{\theta}(u,x)\equiv1$, i.e.~we use no edge correction, when we carry out the bandwidth selection. 
Then, when we
generate the final intensity estimates based on the selected
bandwidths, $\widehat\theta$, we use the local edge correction
$w_{\widehat\theta}(u,x)=\int_W\kappa_{\widehat\theta}(u-x)\de x$. 
Moreover, to measure the performance, for each
model and bandwidth selection approach we report estimates of the integrated absolute bias ($\mathrm{IAB}$), 
the integrated squared bias ($\mathrm{ISB}$),
the integrated variance ($\mathrm{IV}$) 
and the
mean integrated squared error ($\mathrm{MISE}$):
\begin{align*}
\mathrm{IAB}
=& \int_W|\widehat\E[\widehat\rho_{\widehat\theta}(u,X)]-\rho(u)|\de u,
\\
\mathrm{ISB}
=& \int_W(\widehat\E[\widehat\rho_{\widehat\theta}(u,X)]-\rho(u))^2\de
u,
\\
\mathrm{IV}
=& 
\int_W\widehat\Var(\widehat\rho_{\widehat\theta}(u,X))\de u,
\\
\mathrm{MISE}
=& \mathrm{ISB} + \mathrm{IV}.
\end{align*}
For a given model, these estimates are obtained by averaging with
respect to the outcomes for the 100 simulated realisations. 

What we will show here is that, in terms of $\mathrm{MISE}$, our new approach quite substantially outperforms the state
of the art, which here is represented by the approach of
\citet{cronie2018bandwidth}. 
What we specifically do is to considerate the following settings:
\begin{itemize}
    \item For each model, we use each of the loss functions \eqref{e:EstFunGeneralMedian}, \eqref{e:EstFunGeneral} and  \eqref{e:EstFunGeneralMean} with $n=1$, the innovations \eqref{e:InnovationKernel} and the test functions $h_{\theta}(u,\x_i^T) = f(p\widehat\rho_{\theta}(u;\x_i^T)/(1-p))$, where $f(x)=1/x$; recall \eqref{e:NonParam}. 
    In the case of MCCV, we consider the sequence $p=0.1, 0.3, 0.5, 0.7, 0.9$ and $k=400$, whereas in the case of multinomial CV we evaluate the performance for the sequence $k=2,3,\ldots,10$. 
    The numerical output for the MCCV case can be found in Figure \ref{f:Bandwidth} and the output for the multinomial CV case can be found in Figure \ref{f:kFold}.
    
    \item For each model, we also consider the approach of \citet{cronie2018bandwidth}, i.e.~the loss function \eqref{e:CvL} with $f(x)=1/x$, which has been implemented in the function \verb|bw.CvL| in the \textsf{R} package \textsf{spatstat} \citep{BRT15}. The numerical results can be found in Table \ref{tab:CvL}.
    
    \item Finally, to shed some light on the choice of test function, we evaluate the loss function $\mathcal{L}_2(\theta)$ in combination with \eqref{e:NonParam}, i.e.~\eqref{e:EstFunGeneral} with $n=1$, using the test function $f(x)=1/\sqrt{x}$ as well as MCCV with $p=0.1, 0.3, 0.5, 0.7, 0.9$ and $k=400$. The corresponding numerical results can be found in Figure \ref{f:Pearson}. 
\end{itemize}

First of all, regarding the choice $k=400$ in the MCCV case, it should be stated here that we have observed that increasing $k$ beyond 100 essentially has little/no effect on the chosen performance measures so our general suggestion is to fix
$k\geq100$. Choosing $k$ is clearly a question of computational cost and as a computationally cheaper alternative in the MCCV case, one could
instead sequentially increase $k$ and stop once the loss function shows signs
of converges (theoretically, the convergence is guaranteed by the law
of large numbers). 

Most importantly, comparing Figure \ref{f:Bandwidth}, Figure \ref{f:Pearson} and Figure \ref{f:kFold} with Table \ref{tab:CvL}, the first thing we note is that regardless of the choice of $p$, $k$ and model, all of the point process learning approaches outperform the approach of \citet{cronie2018bandwidth} in terms of $\mathrm{MISE}$. 
Although the approach of \citet{cronie2018bandwidth} performs slightly better in terms of bias, it performs comparatively poorly in terms of variance, which is consequently the reason for its higher $\mathrm{MISE}$; it is worth emphasising that it is precisely the lower variance which ensures that the \citet{cronie2018bandwidth} approach outperforms its predecessors \citep{cronie2018bandwidth,Moradi2019}, e.g.~the bandwidth selection approach in \eqref{e:PPL}. We do however hypothesise that if $p\to0$ in the MCCV case (possibly in combination with $k\to\infty$), or $k\to\infty$ in the multinomial CV case (e.g. in combination with $\mathcal{L}_3$), we would reach the same bias level as the \citet{cronie2018bandwidth} approach, but still with a significantly lower $\mathrm{MISE}$. 
In conclusion, 
the approach of
\citet{cronie2018bandwidth} is improved substantially by framing it within our point process learning approach. 

Having compared the performances of the new approaches to the state of the art, we next look closer at the different choices we have made. 
First of all, it seems that $\mathcal{L}_3$ favours a lower bias over a lower variance/$\mathrm{MISE}$, whereas $\mathcal{L}_2$ favours the opposite; $\mathcal{L}_1$ seems to offer some middle-ground between the two. 
We further see in Figure \ref{f:Bandwidth} and Figure \ref{f:Pearson} that, in the case of MCCV, $p\in[0.5,0.7]$ tends to be a safe/good choice, which balances the trade-off between bias
and variance, irrespective of the degree of clustering/inhibition of the underlying model. 
Comparing the two associated test functions, i.e.~$f(x)=1/x$ (Figure \ref{f:Bandwidth}) and $f(x)=1/\sqrt{x}$ (Figure \ref{f:Pearson}), we see that the latter reduces the bias with respect to the former, but at the cost of increasing the variance and, consequently, $\mathrm{MISE}$. 
Moreover, it seems that the performance of multinomial CV in terms of $\mathrm{MISE}$ is the best when $k=2$ (see Figure \ref{f:kFold}), which is equivalent to considering MCCV with $p=0.5$ and $k=1$. Since multinomial CV overall performs only slightly poorer 
than the MCCV approach with $p=0.5$ and $k=400$, irrespective of the choice of $k$, we draw the conclusion that the most relevant aspect of the choice of CV setting in the bandwidth selection context is that we split the data roughly in half. 

Turning to the computational aspects, we always have to compute $k$ innovations, irrespective of whether we use $k$-fold multinomial CV or MCCV. This e.g.~means that 2-fold MCCV is roughly 400 times faster than MCCV with $k=400$ and $p=0.5$. It should be stated, however, that parallelisation may reduce this speed difference by roughly as many times as there are cores to be accessed. Using a standard laptop, we have experienced that for a single pattern of moderate size the speed difference between the two approaches is barely noticeable.  
However, since multinomial CV is faster than MCCV, it would be one's go-to method if computational aspects are the main priority, whereas MCCV is the go-to method if precision is what one is going for.



The fact that in the case of clustered models the performance in terms of MISE is improved when the training sets have few points, i.e.~when $p$ is large, corresponds precisely with observations made in \citet{Moradi2019}. We
hypothesise that the reason for clustered processes requiring more
thinning is that here the points need to be sufficiently sparse for the approach to capture the general trends of the underlying intensity.  

\begin{figure}[!htpb]
\centering
\includegraphics[scale=0.16]{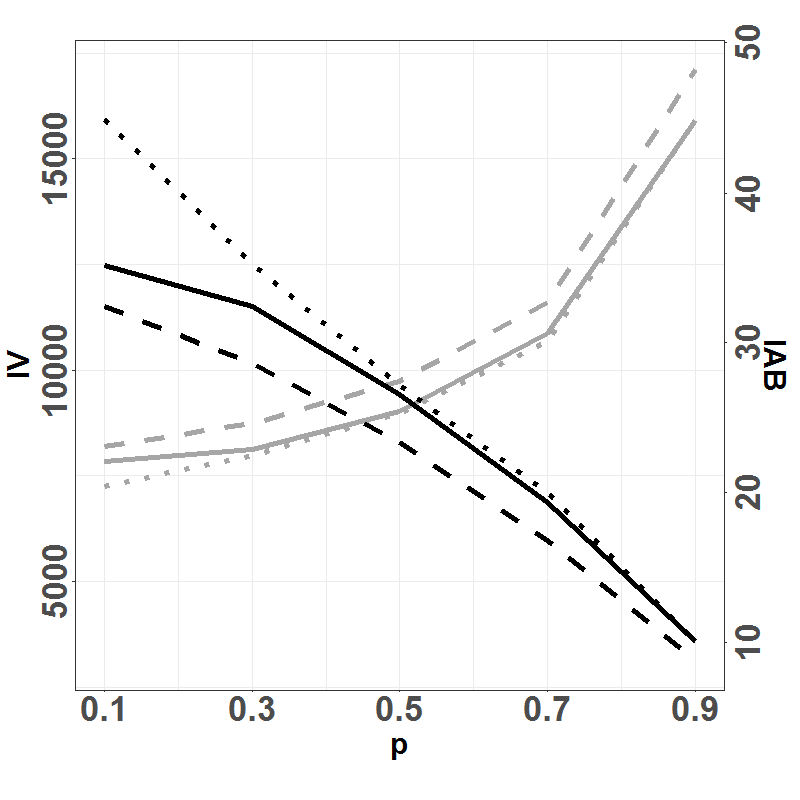}
\includegraphics[scale=0.16]{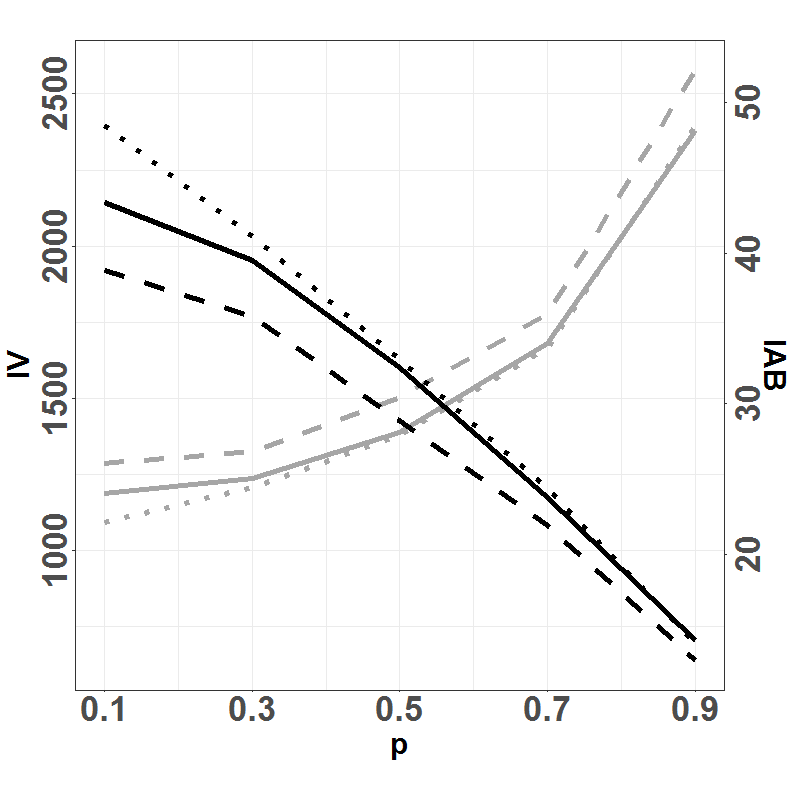}
\includegraphics[scale=0.16]{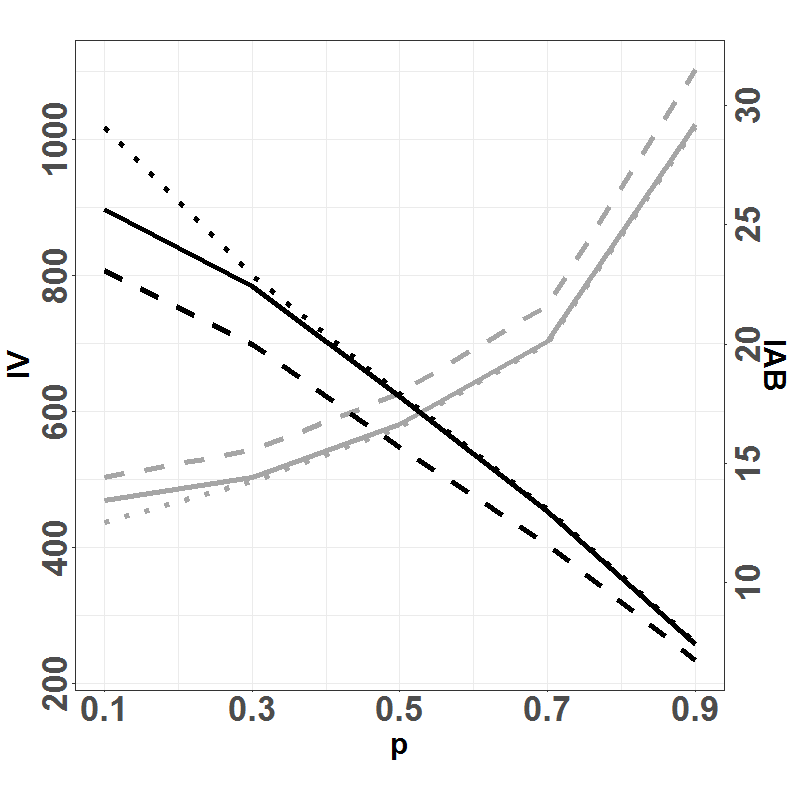}

\includegraphics[scale=0.16]{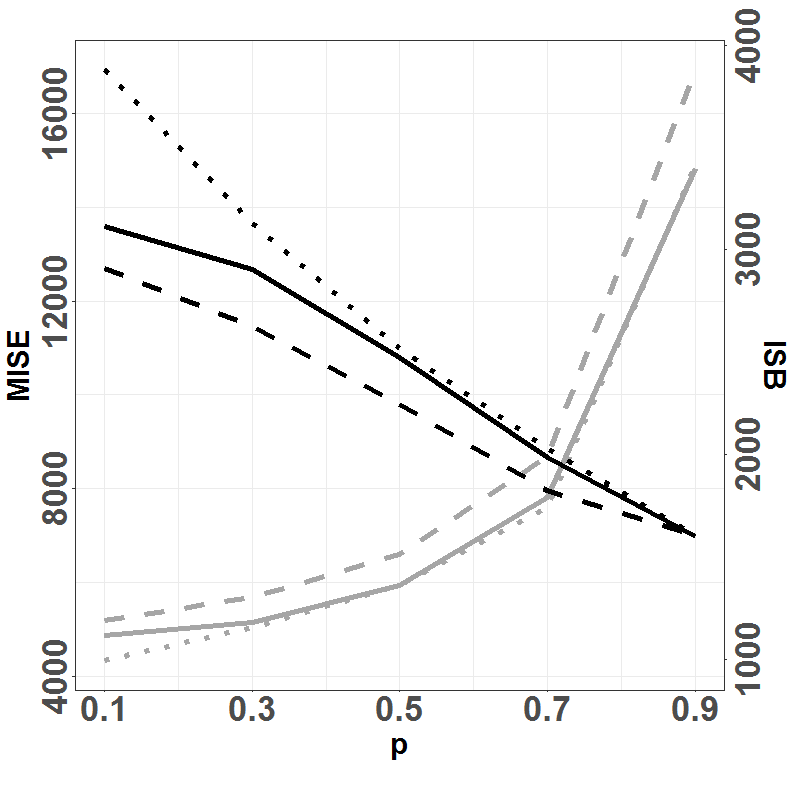}
\includegraphics[scale=0.16]{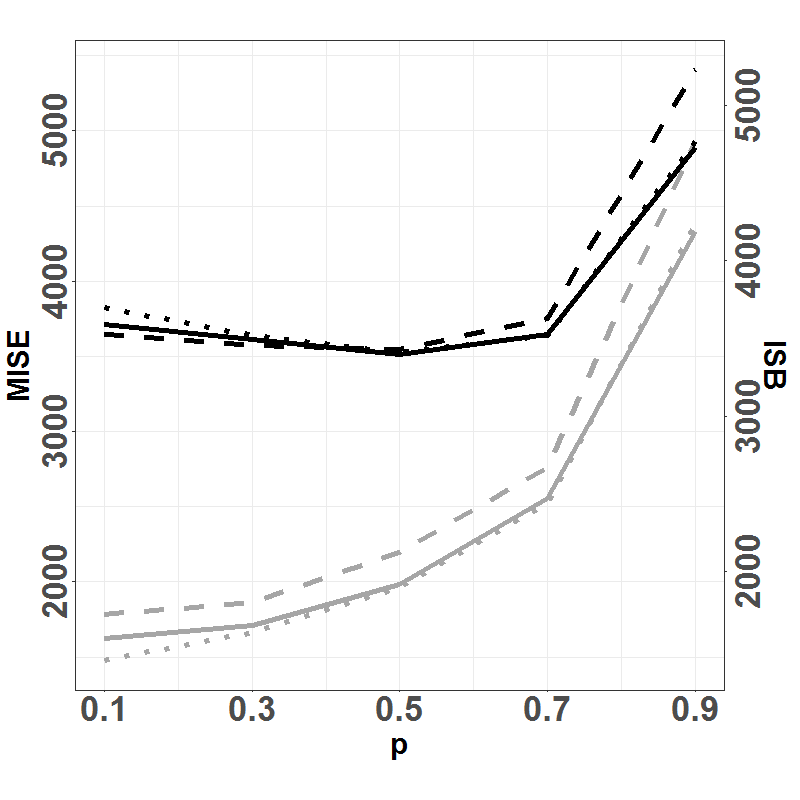}
\includegraphics[scale=0.16]{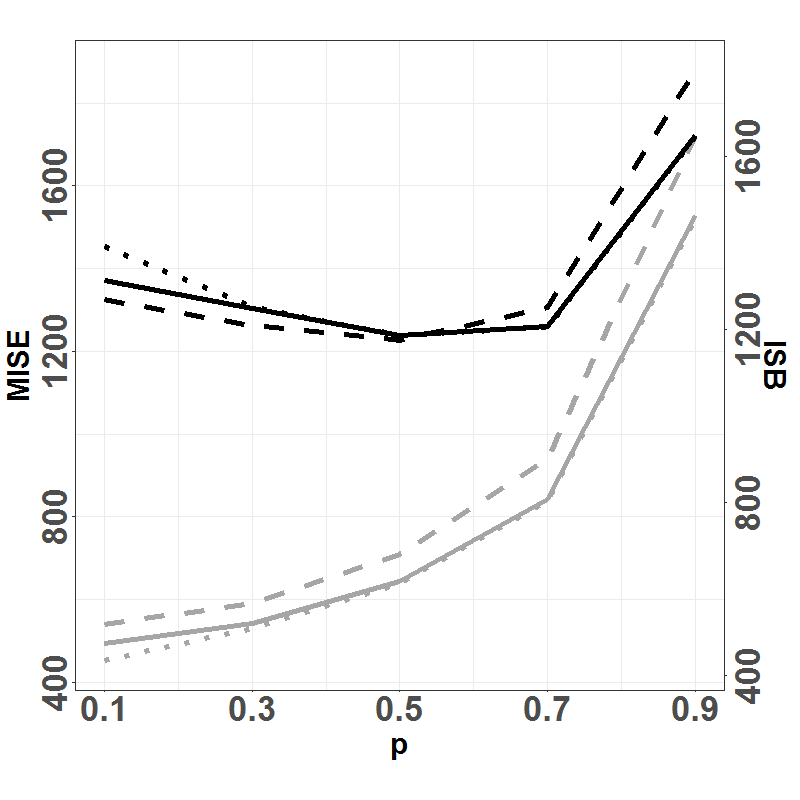}

\caption{
Performance of the loss functions $\mathcal{L}_1$ (solid curves, --), $\mathcal{L}_2$ (dashed curves, $---$) and $\mathcal{L}_3$ (dotted curves, $\cdots$), using MCCV with $p=0.1,0.3,0.4,0.7,0.9$ and $k=400$ together with the test function $f(x)=1/x$. 
Columns: 
LGCP (left), Poisson (middle) and DPP (right). 
Top row: $\mathrm{IAB}$ (grey curve, right axis) and $\mathrm{IV}$ (black curve, left axis). 
Bottom row: $\mathrm{ISB}$ (grey curve, right axis) and $\mathrm{MISE}$ (black curve, left axis). 
}
\label{f:Bandwidth}
\end{figure}

\begin{figure}[!htpb]
 \centering
\includegraphics[scale=0.16]{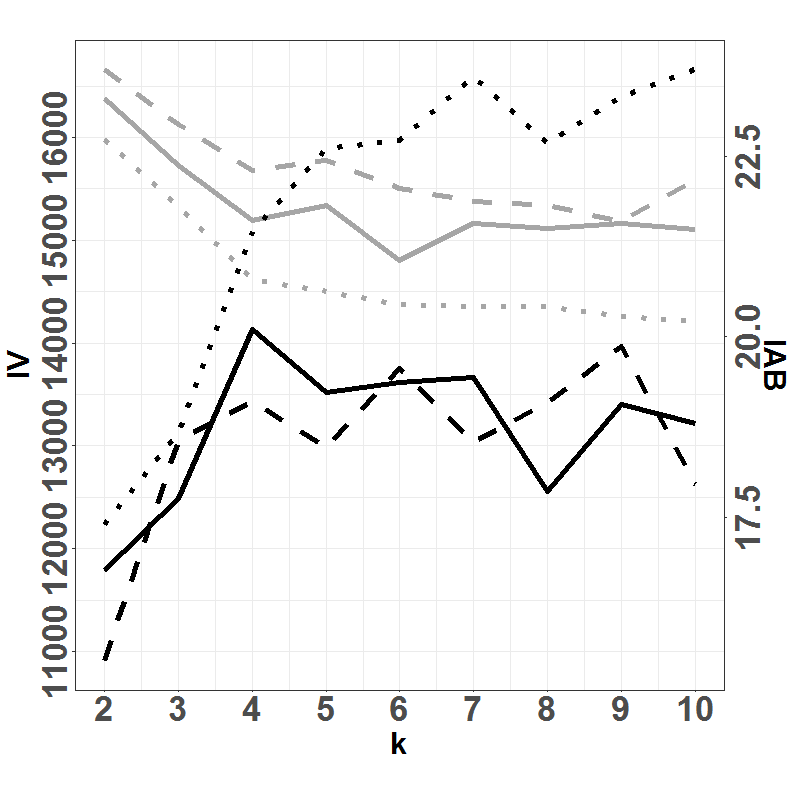}
\includegraphics[scale=0.16]{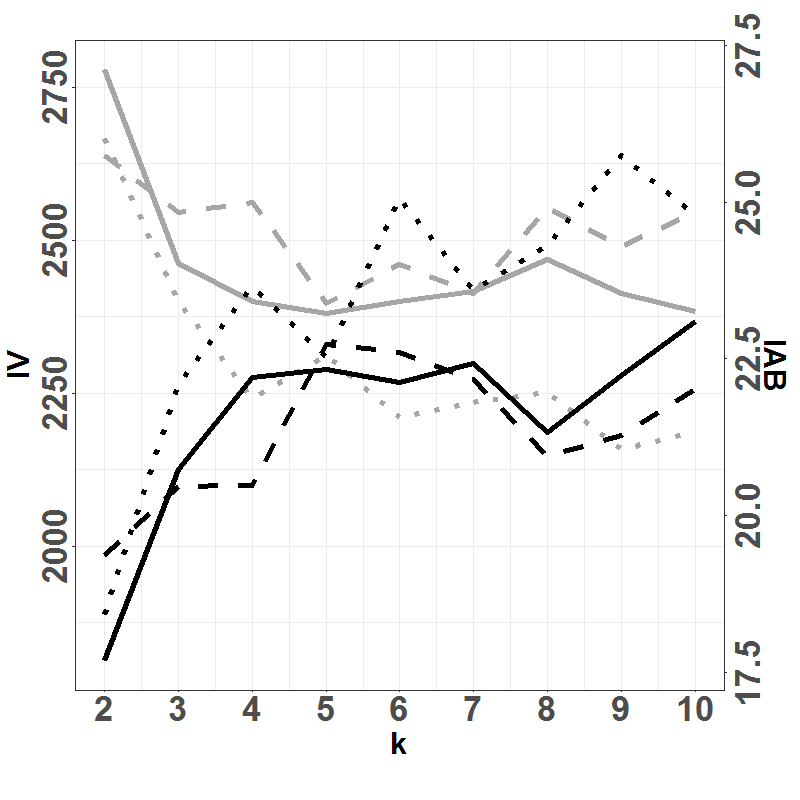}
\includegraphics[scale=0.16]{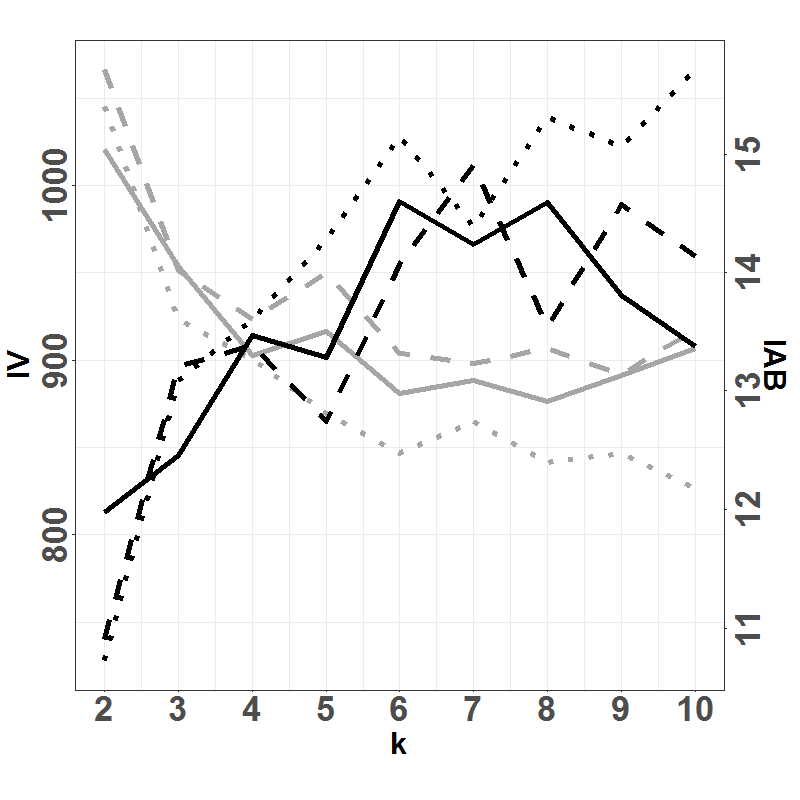}

\includegraphics[scale=0.16]{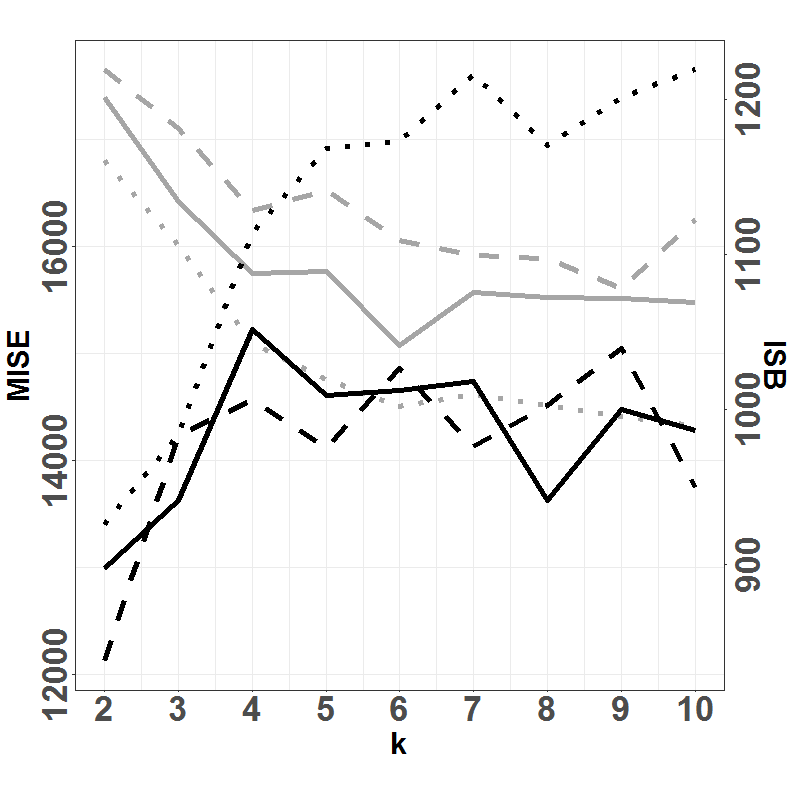}
\includegraphics[scale=0.16]{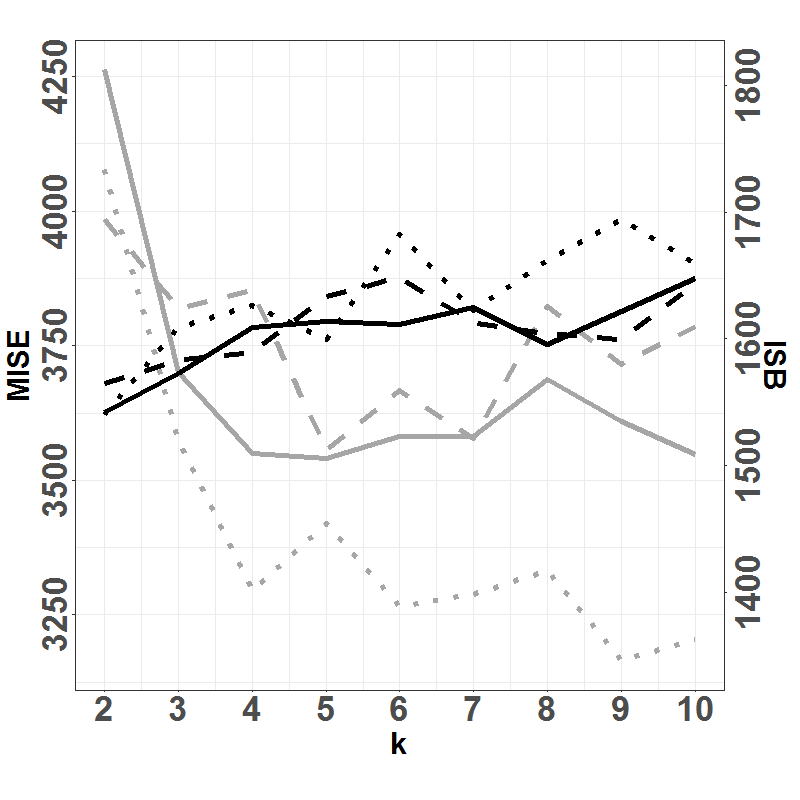}
\includegraphics[scale=0.16]{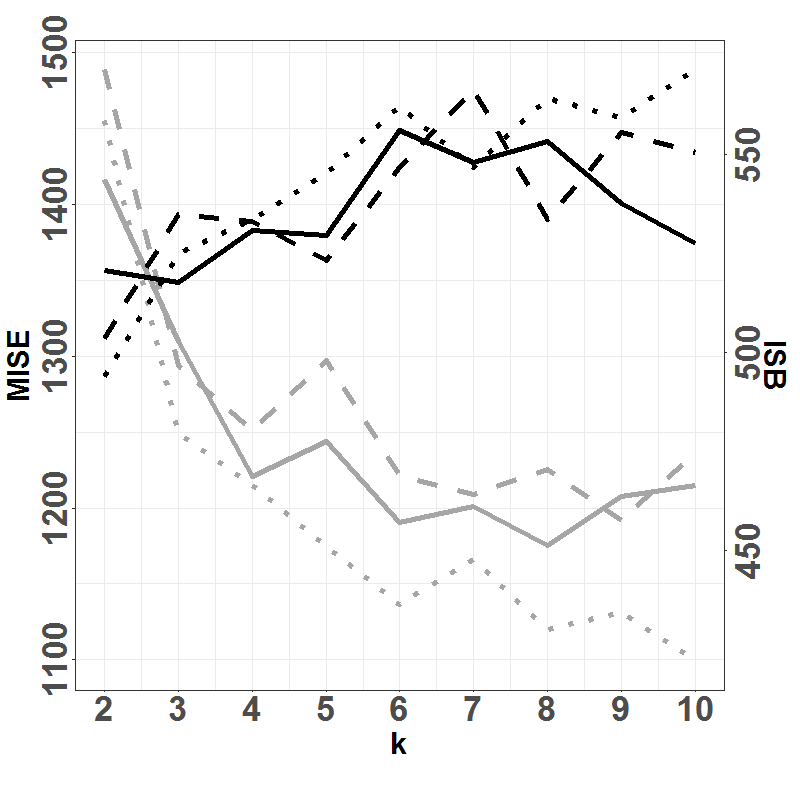}

\caption{
Performance of the loss functions $\mathcal{L}_1$ (solid curves, --), $\mathcal{L}_2$ (dashed curves, $---$) and $\mathcal{L}_3$ (dotted curves, $\cdots$), using multinomial CV with $k=2,3,\ldots,10$ and the test function $f(x)=1/x$. 
Columns: 
LGCP (left), Poisson (middle) and DPP (right). 
Top row: $\mathrm{IAB}$ (grey curve, right axis) and $\mathrm{IV}$ (black curve, left axis). 
Bottom row: $\mathrm{ISB}$ (grey curve, right axis) and $\mathrm{MISE}$ (black curve, left axis). 
}
\label{f:kFold}
\end{figure}

\begin{figure}[!htpb]
 \centering
\includegraphics[scale=0.16]{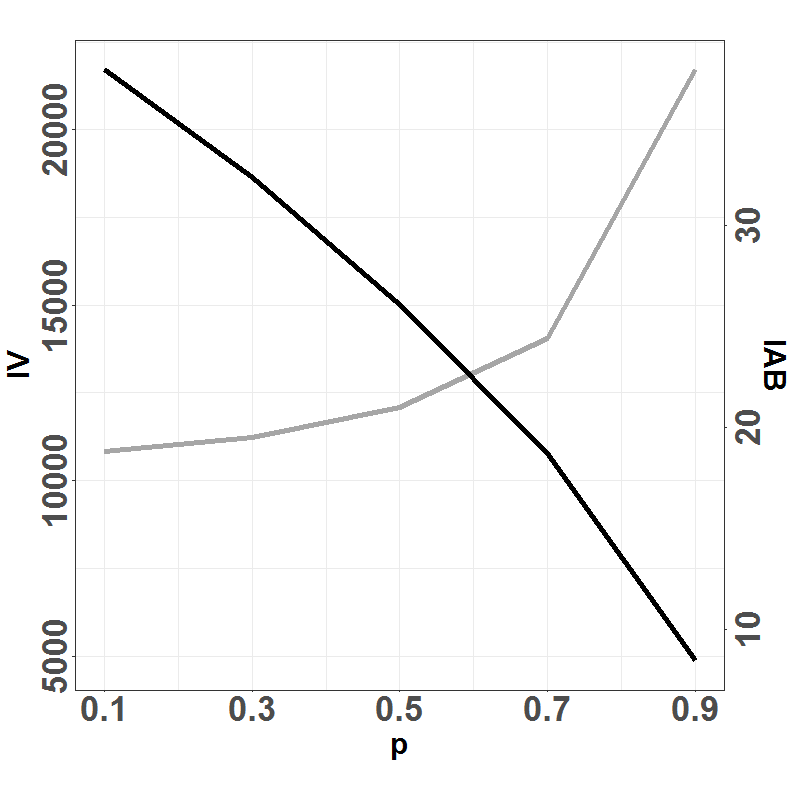}
\includegraphics[scale=0.16]{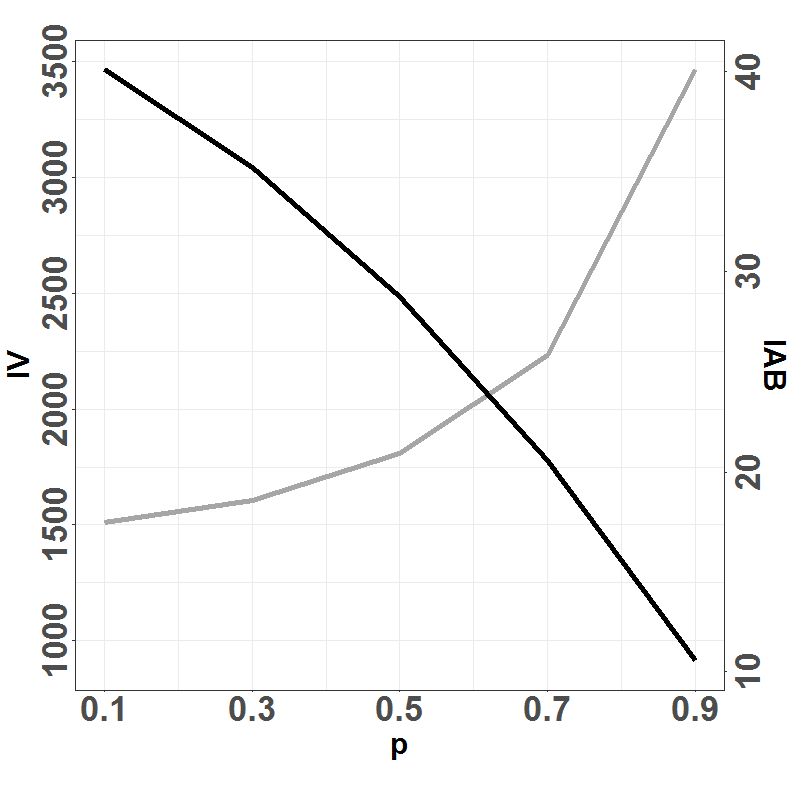}
\includegraphics[scale=0.16]{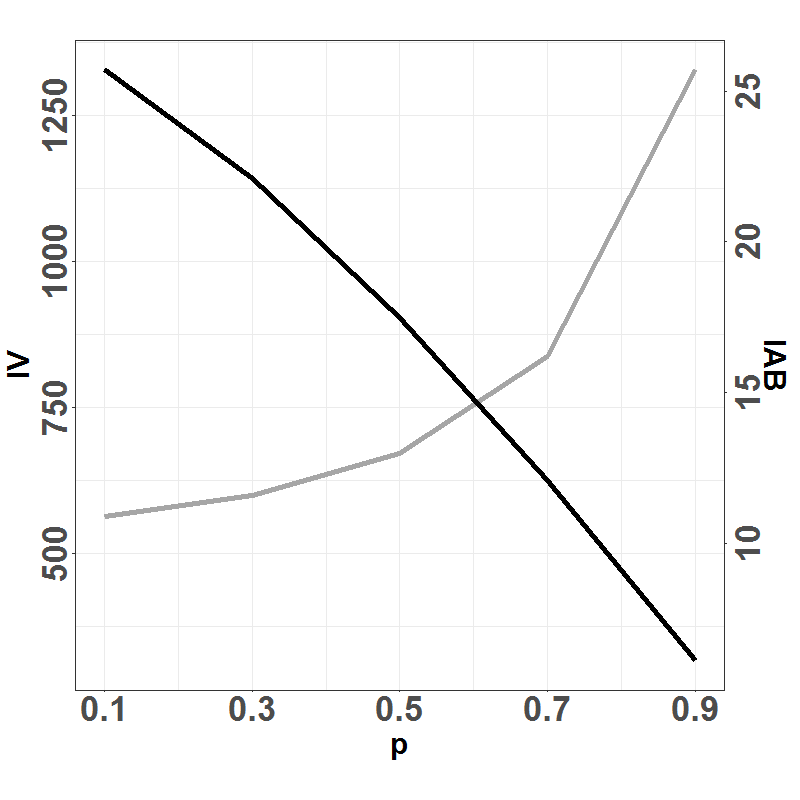}

\includegraphics[scale=0.16]{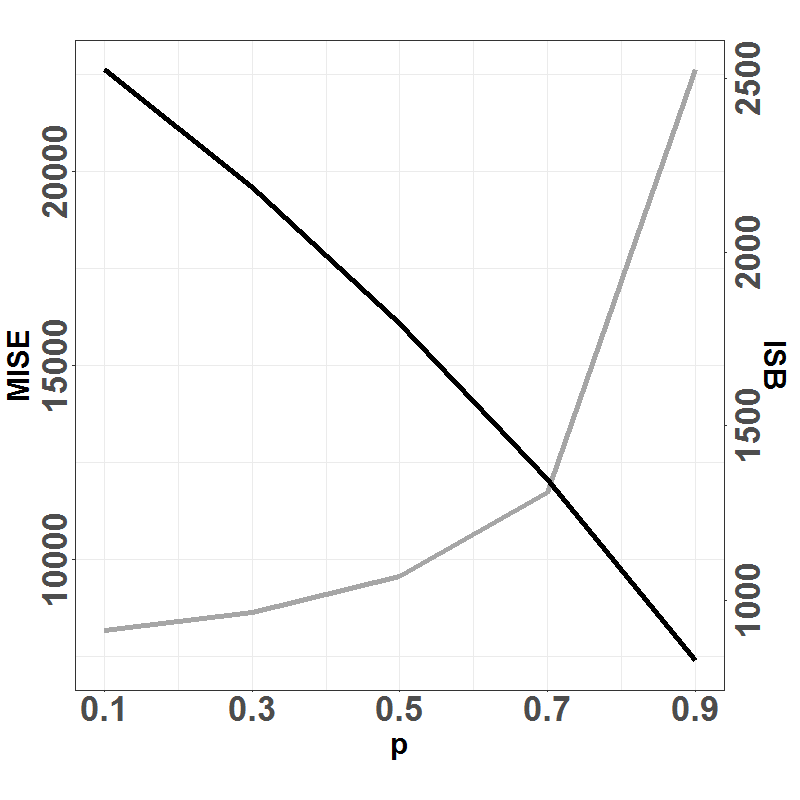}
\includegraphics[scale=0.16]{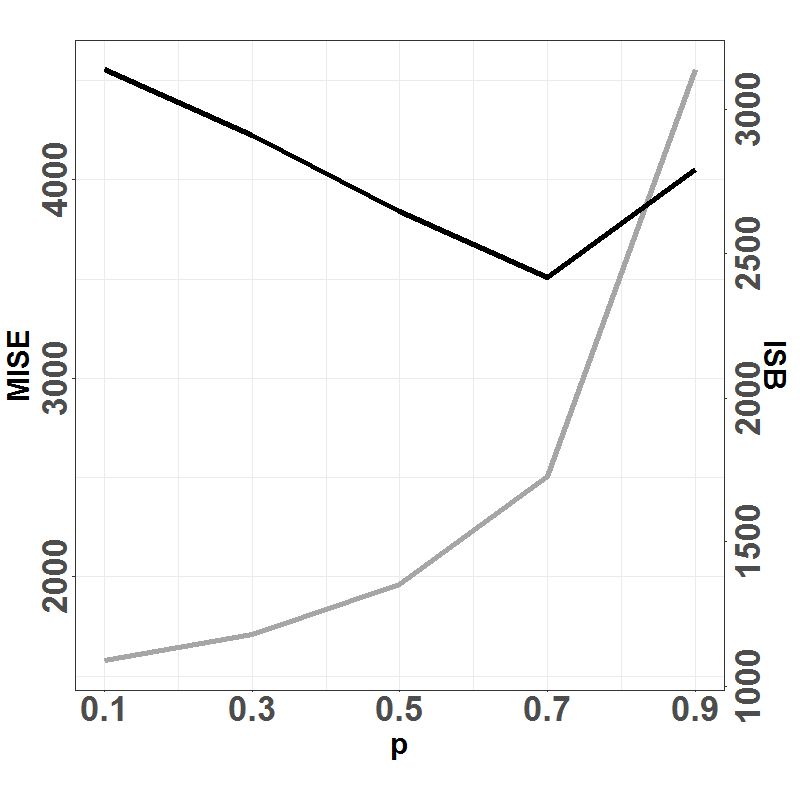}
\includegraphics[scale=0.16]{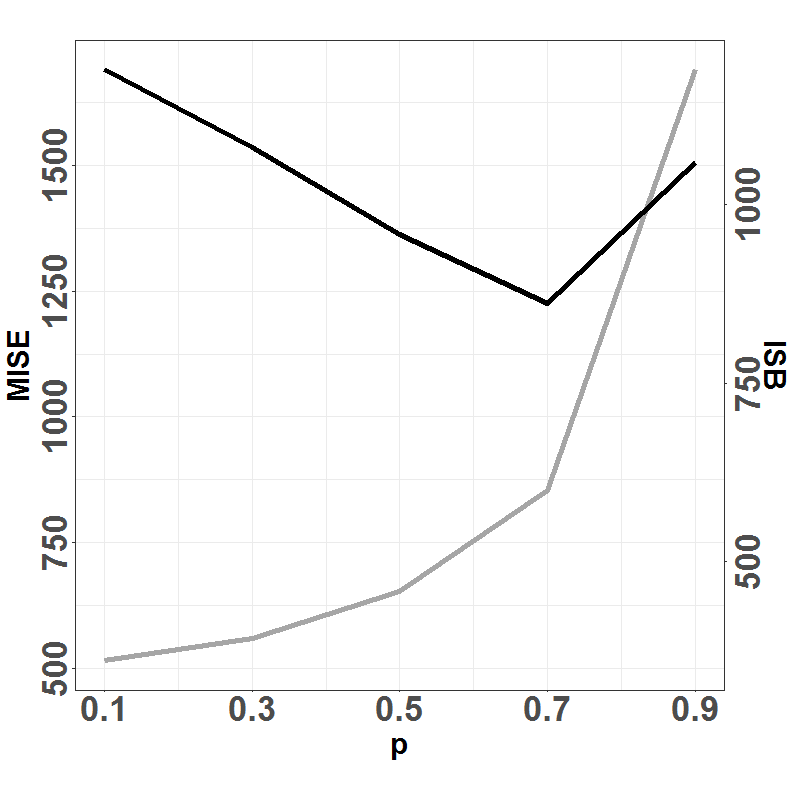}

\caption{
Performance of the loss function $\mathcal{L}_2$, using MCCV with $p=0.1,0.3,0.4,0.7,0.9$ and $k=400$ together with the test function $f(x)=1/\sqrt{x}$. 
Columns: 
LGCP (left), Poisson (middle) and DPP (right). 
Top row: $\mathrm{IAB}$ (grey curve, right axis) and $\mathrm{IV}$ (black curve, left axis). 
Bottom row: $\mathrm{ISB}$ (grey curve, right axis) and $\mathrm{MISE}$ (black curve, left axis). 
}
\label{f:Pearson}
\end{figure}

\begin{table}[!htpb]
    \centering
    \caption{The numerical performance of the method of \citet{cronie2018bandwidth}, i.e.~\eqref{e:CvL} in combination with $f(x)=1/x$, for each of the models above.}
    \begin{tabular}{l|rrrr}
    & IAB & ISB & IV & MISE \\
    \hline
    LGCP    &19.48 & 963.47 & 17597.99 & 18561.47\\
    Poisson &  15.80 & 921.82 & 4408.21 & 5330.04\\
    DPP & 9.14 & 276.75 & 2002.55 & 2279.31\\
    \end{tabular}
    \label{tab:CvL}
\end{table}








\subsubsection{Real data example}
We next consider an application of our point process learning-based kernel intensity estimation approach, in the context of a point pattern $\x\subseteq W\subseteq\R^2$. 
More precisely, we consider MCCV with $(k,p)=(400,0.7)$, the test function $f(x)=1/x$ and the loss function $\mathcal{L}_2$. 
Figure \ref{f:bei} shows the estimated intensity for the dataset, which consists of the locations of 3605 tropical rain forest trees of the species Beilschmiedia pendula (Lauraceae), sampled on Barro Colorado Island, Panama \citep{hubbell1983diversity}. 
This dataset, which is accessible through the \textsf{R} package \textsf{spatstat} \citep{BRT15}, has previously been analysed in various papers \citep[see e.g.][]{moller2007modern}. We obtain a bandwidth of $56.65$ (meters) and from Figure \ref{f:bei} we can see that the selected bandwidth leads to an estimated intensity which adapts well to the inhomogeneity of the spatial locations of the trees.


\begin{figure}[!h]
    \centering
    \includegraphics[scale=0.3]{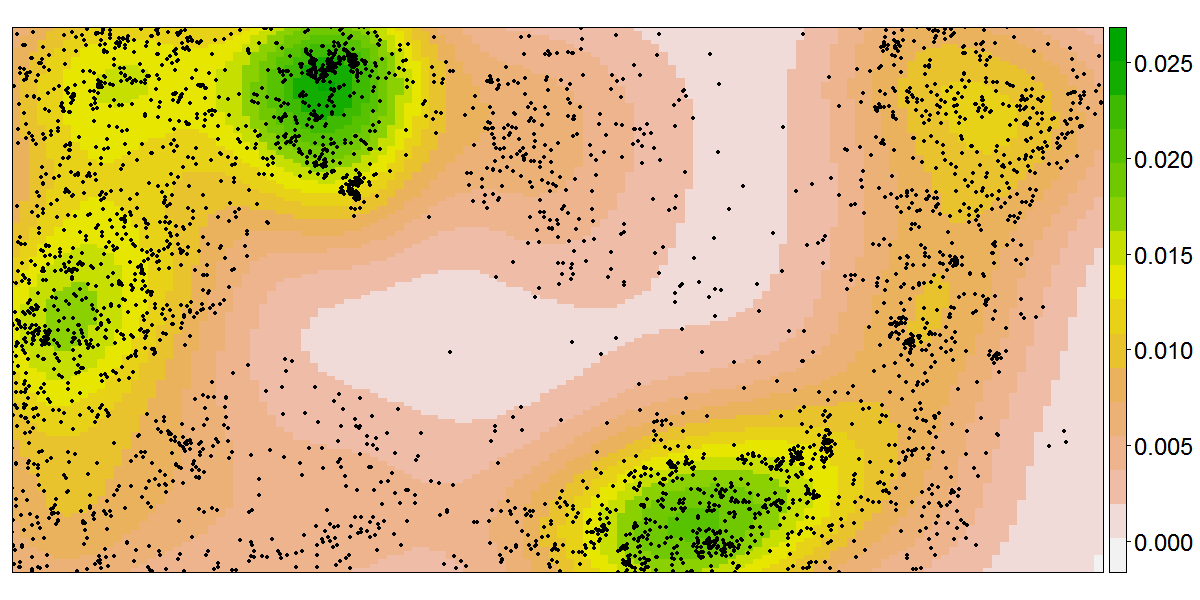}
    \caption{Kernel intensity estimate for a dataset consisting of the locations of tropical rain forest trees on Barro Colorado Island, Panama.}
    \label{f:bei}
\end{figure}



\subsection{Fixed size samples}

For completeness, we here look closer at the setting in Section \ref{s:RandomSamples}, i.e.~the case where we condition on the total point count $X(S)=N\geq1$. 

We here have the parametric intensity function family $\rho_{\theta}(u)=Nf_1(u;\theta)$, $u\in S$, where $f_1(\cdot;\theta)$, $\theta\in\Theta$, is a parametric family of density functions and $f_1(\cdot;\theta_0)$ is the true common marginal density function of $x_1,\ldots,x_N$. From \eqref{e:InnovationParametric} we obtain the innovations 
\begin{align*}
\mathcal{I}_{(1-p)\rho_{\theta}}^{h_{\theta}}(S;\x_i^T)
    =&
  \sum_{x\in\x_i^T}
     h_{\theta}(x)
     -
     (1-p)
     \int_{S}
     h_{\theta}(u)
     \rho_{\theta}(u)
     \de u
     \\
     =&
     \sum_{i=1}^N
     \1\{x_i\in\x_i^T\}
     h_{\theta}(x)
     -
     (1-p)N
     \int_{S}
     h_{\theta}(u)
     f_1(u;\theta)
     \de u; 
\end{align*}
we have to require that the test functions are such that the expectation $\int_{S} h_{\theta}(u) f_1(u;\theta) \de u$ is finite. 
In particular, by setting $h_{\theta}(\cdot)=((1-p) N f_1(\cdot;\theta))^{-1}\nabla_{\theta} f_1(\cdot;\theta)$ in this unbiased estimating equation (provided that the gradient $\nabla_{\theta} f_1(\cdot;\theta)$ exists), since $f_1(\cdot;\theta)$ is a density function, we obtain
\begin{align*}
\mathcal{I}_{(1-p)\rho_{\theta}}^{\nabla_{\theta} f_1(\cdot;\theta)/((1-p)Nf_1(\cdot;\theta))}(S;\x_i^T)
    =&
    \frac{1}{(1-p)N}
     \sum_{i=1}^N
     \1\{x_i\in\x_i^T\}
     \frac{\nabla_{\theta} f_1(x_i;\theta)}{f_1(x_i;\theta)}
     -
     \nabla_{\theta} 
     \int_{S}
     f_1(u;\theta)
     \de u
     \\
     =&
     \frac{1}{N}
     \sum_{i=1}^N
     \frac{\1\{x_i\in\x_i^T\}}{1-p}
     \frac{\nabla_{\theta} f_1(x_i;\theta)}{f_1(x_i;\theta)}
     -
     1,
\end{align*}
assuming that integration and differentiation may be interchanged. This is connected to the likelihood score function based on the sample $\x_i^T$. Hence, we obtain a form of subsampled maximum likelihood estimation approach. Note that we here have not imposed the assumption that $x_1,\ldots,x_N$ are independent, i.e.~that we are considering a Binomial point process.

We next consider the case where we want to fit a multivariate density $f_N(u_1,\ldots,u_N;\theta)$, $u_1,\ldots,u_N\in S$,  $\theta\in\Theta$, to $x_1,\ldots,x_N$, where we assume that the true joint density of $x_1,\ldots,x_N$ is given by  $f_N(u_1,\ldots,u_N;\theta_0)$, $u_1,\ldots,u_N\in S$. 
Recall from Section \ref{s:RandomSamples} the associated Papangelou conditional intensity $\lambda_{\theta}(u,\x) = n f_1(u|\x;\theta)$, $u\in S$, $\x=\{x_1,\ldots,x_{n-1}\}\subseteq S^{n-1}$, $1\leq n\leq N-1$, where $n=1$ yields $f_1(u|\x;\theta)=f_1(u|\emptyset;\theta)=f_1(u;\theta)$. 
Here the innovations in \eqref{e:PapangelouInnovation1} become
\begin{align*}
  \mathcal{I}_{\xi_{\theta}^1}^{h_{\theta}}(S;\x_i^V,\x_i^T)
  =&
     \mathcal{I}_{p(1-p)^{-1}\lambda_{\theta}}^{h_{\theta}}(S;\x_i^V,\x_i^T)
  \\
  =&
     \sum_{x\in\x_i^V}
     h_{\theta}(x;\x_i^T)
     -
     \frac{p}{1-p}
     \int_{S}
     h_{\theta}(u;\x_i^T)
     \lambda_{\theta}(u;\x_i^T)
     \de u
     \\
     =&
     \sum_{j=1}^N
     \1\{x_j\in\x_i^V\}
     h_{\theta}(x_j;\x_i^T)
     -
     \frac{p}{1-p}
     (\#\x_i^T+1) 
     \int_{S}
     h_{\theta}(u;\x_i^T)
     f_1(u|\x_i^T;\theta)
     \de u
     ,
\end{align*}
and we use $\widetilde{\mathcal{I}}_{\xi_{\theta}^1}^{h_{\theta}}(S;\x_i^V,\x_i^T)=I_i\mathcal{I}_{\xi_{\theta}^1}^{h_{\theta}}(S;\x_i^V,\x_i^T)$, where $I_i=\1\{1\leq\#\x_i^T\leq\#\x-1\}=\1\{0<\#\x_i^T< N\}$, to build the loss function. Arguing as in the marginal case, if we set $h_{\theta}(\cdot;\x_i^T)=
\nabla_{\theta} f_1(\cdot|\x_i^T;\theta)/f_1(\cdot|\x_i^T;\theta)p$, we obtain a subsampled conditional likelihood estimation approach;
\begin{align*}
  \mathcal{I}_{\xi_{\theta}^1}^{h_{\theta}}(S;\x_i^V,\x_i^T)
  =&
     \sum_{j=1}^N
     \frac{
     \1\{x_j\in\x_i^V\}}{p}
     \frac{\nabla_{\theta} f_1(x_j|\x_i^T;\theta)}{f_1(x_j|\x_i^T;\theta)}
     -
     \frac{
     (\#\x_i^T+1)
     }{1-p}
     .
\end{align*}

Note that, as we have seen in the previous sections, by changing the chosen test function we also change the emphasis put on either the bias or the variance of the resulting estimators. Hence, we may obtain estimators which behave quite differently than likelihood-type estimators.

\section{Discussion}
\label{s:Discussion}



In this paper we present a general theory for statistical learning for point processes, point process learning, which essentially consists of two components. The first is what we refer to as bivariate innovations, which essentially are measures of discrepancy between two point processes. These are used to ``predict'' one point process from another point process and they are motivated by integral formulas such as Campbell, Campbell-Mecke and Georgii-Nguyen-Zessin (GNZ) formulas, which relate expectations of sums over the points of a point process to integrals with respect to various distributional characteristics of the point process, including~factorial moment measures and Papangelou conditional intensity functions. A particular instance of our bivariate innovations reduce to the ``classical'' innovations of  \citet{baddeley2005residual,baddeley2008properties}, which were introduced to define point process residuals. The second component of our statistical learning approach is thinning-based cross-validation for point processes. Due to the many appealing  distributional properties of independent thinnings, which are obtained by independently retaining the points of a point process according to some given probability function, we study the case of independent thinning-based cross-validation in detail. 
By using the bivariate innovations to predict cross-validation-generated validation sets, using the corresponding training sets, we connect the two components to form our (supervised) statistical learning approach. Through a range of applications, including kernel intensity estimation, parametric intensity estimation and Papangelou conditional intensity fitting, we show that our learning approach generally performs better than the classical statistical state of the art for point processes.

Our hope is that this paper becomes the starting point for a research direction which deals with various aspects of statistical learning for point processes. Aside from the range of future/ongoing work topics indicated here and there in this paper, we are currently also looking into further topics. To begin with, we are studying additional theoretical properties of bivariate innovations, in particular asymptotic properties, where there are clear connections to recent work in stochastic geometry \citep[see e.g.][]{biscio2018note, biscio2019alpha, yukich2019fastdecay, biscio2020TDA}. 
Moreover, a natural extension of our work in this paper is to apply our point process learning framework to marked point processes, since these enable regression-type estimation in the context of dependent samples.  Further ongoing work on extensions is mentioned in Section \ref{s:DataDrivenLoss} below; the sections below discuss various alternatives and extensions for the framework developed in this paper. 






\subsection{Alternative cross-validation approaches}

As a sort of mix between the cross-validation (CV) methods in Definition \ref{def:MCCV}
  and Definition \ref{def:kFoldMCCV}, one could consider letting
  $\x_i^V$ be a $p$-thinning of $\x$ with retention probability
  $p_i=i/k$ and letting $\x_i^T=\x\setminus\x_i^V$,
  $i=1,\ldots,k\geq1$. On the one hand, the possible advantage over
  MCCV is that we only have to choose the parameter $k$. On the other
  hand, it is not likely that it would perform better than an
  "optimally" chosen pair $(k,p)$ for MCCV. 
  A further variant of the CV method in Definition \ref{def:MCCV} is to consider an "empirical Bayes" type CV approach, where the retention probability used 
  would be estimated non-parametrically, e.g.~by
  means of a scaled intensity estimate, using (a part of) the
  data. The issue here is that we do not actually employ independent
  thinning to generate the training and validation sets. 
It is not clear whether there are any actual benefits of doing this, but this may be worth exploring further.

Another natural and intuitive approach is what we call {\em domain partitioning CV}. The idea is to split the study region $W$ into $k\geq1$ pieces $W_i$, $i=1,\ldots,k$, and let the $i$th validation set be $\x_i^V=\x\cap W_i$; 
see the left panel of Figure  \ref{f:Splitting} for an illustration. 
This may be obtained through a sequence of $k$ independent thinnings, where in the $i$th run we use the retention probability $p_i(u)=\1\{u\in W_i\}$, $u\in W$, $i=1,\ldots,k$. 
Note the philosophical difference between domain partitioning CV and independent thinning based CV strategies. In the former we "block", or sample "from left to right/top to bottom", since we first split $W$ and then generate the validation sets. In the latter, however, we sample "from above"; see the right panel of Figure  \ref{f:Splitting}.

\begin{figure}[!htpb]
  \centering \includegraphics[width=0.4\textwidth]{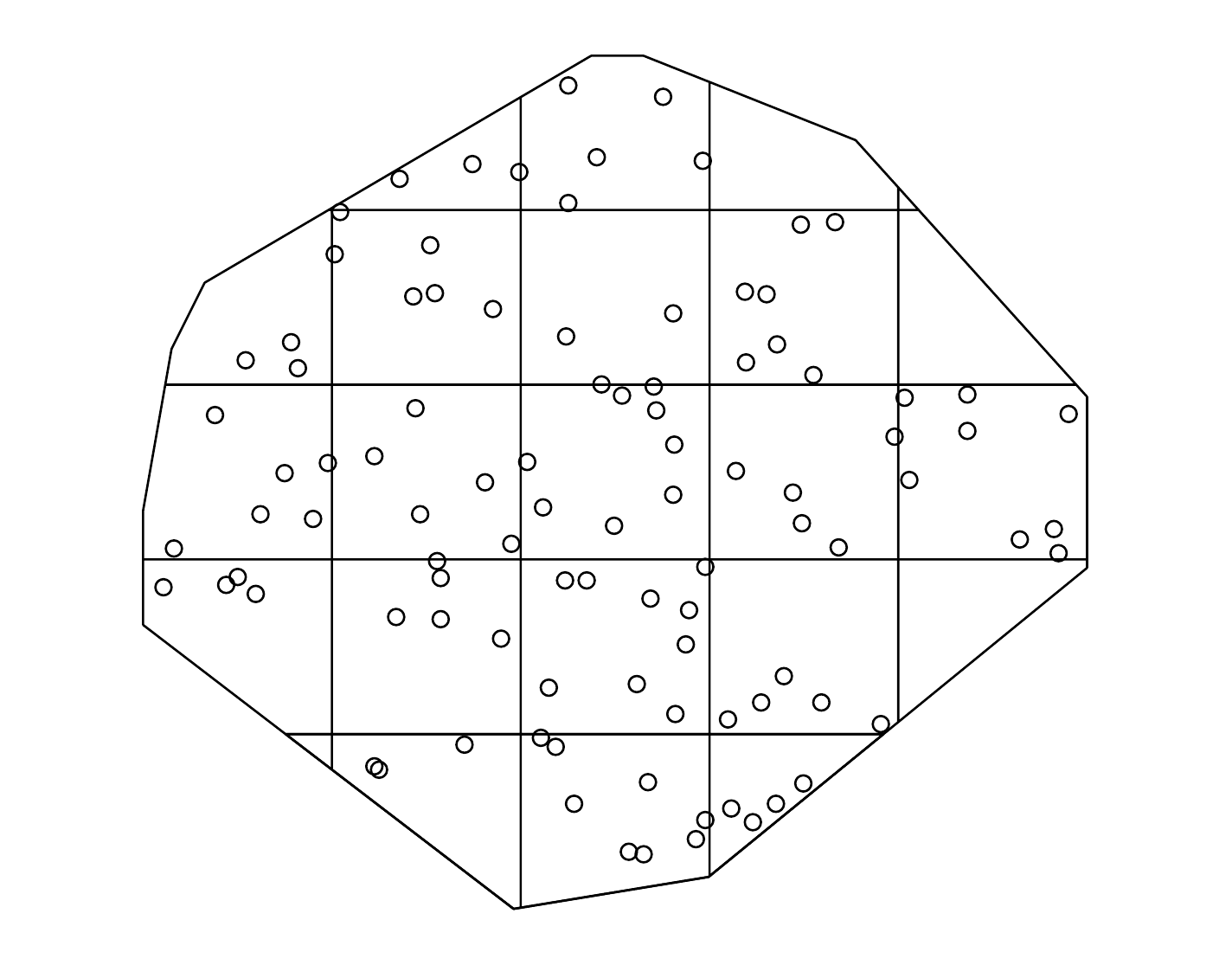}
  \includegraphics[width=0.4\textwidth]{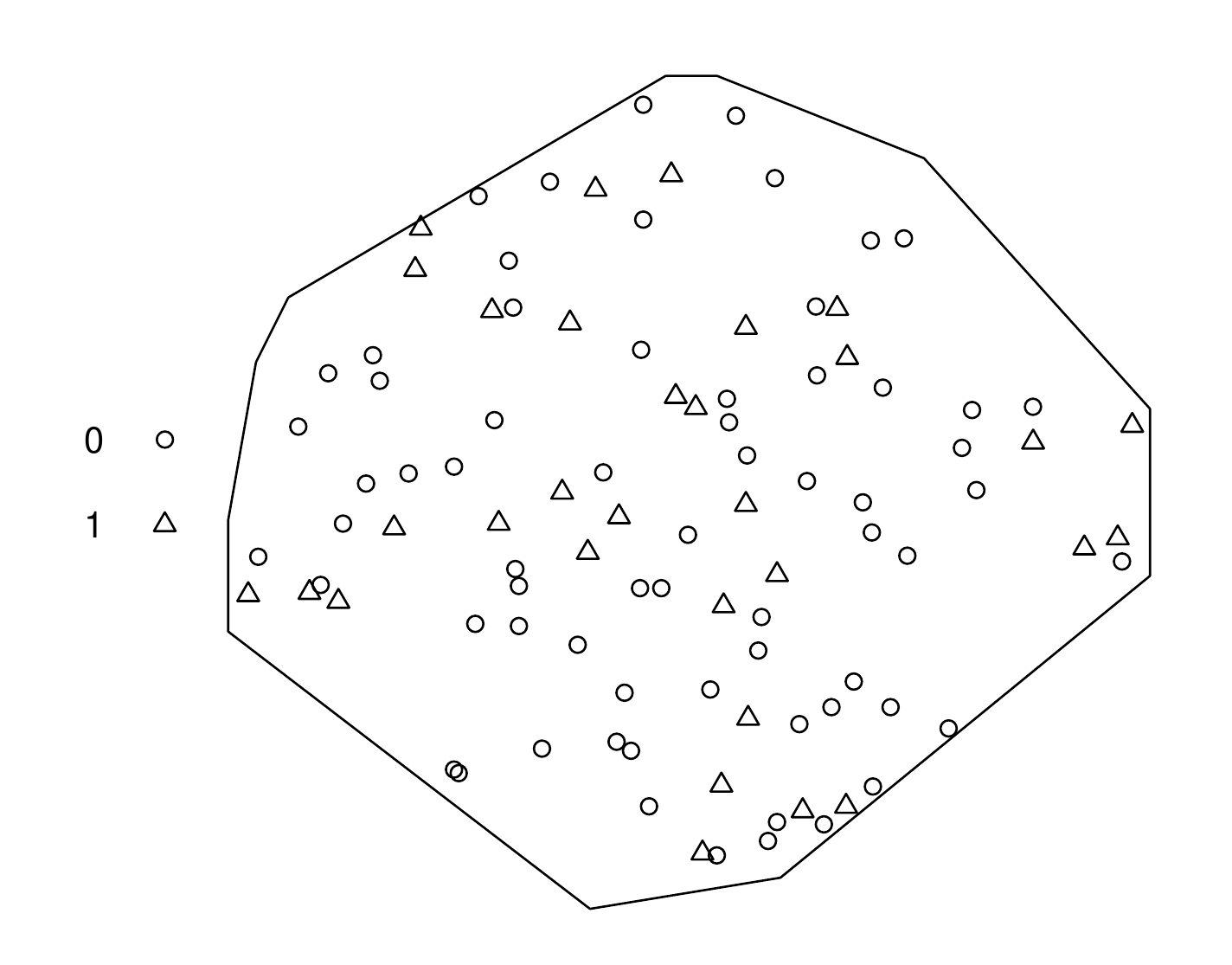}
  \caption{ A point pattern $\x$ in a study region $W\subseteq\R^2$.
    Left: A partition $\{W_i\}_{i=1}^k$ of $W$, which in turn yields
    the validation sets $\x_i^V=\x\cap W_i$, $i=1,\ldots,k\geq2$.
    Right: A validation set $\x_i^V$ (illustrated by $\Delta$) which
    is given by a $p$-thinning of $\x$.  }
  \label{f:Splitting}
\end{figure}

However natural and intuitive domain partitioning CV may seem, there are problems related to our point process learning approach. 
The main problem is related to the retention probability for $\x_i^V$ being 1 and the retention probability for $\x_i^T$ being 0. 
If $\Xi_{\Theta}^n=\{\xi_{\theta}^n: \theta\in\Theta\}$ and $\mathcal{H}_{\Theta}=\{h_{\theta}:\theta\in\Theta\}$ are of the form \eqref{e:GeneralEstimator}, then  \eqref{e:WeightsPapangelou} tells us that we should have $\xi(u_1,\ldots,u_n;X_i^T)
  =
  \prod_{j=1}^n p(u_j)
  (\prod_{j=1}^n (1-p(u_j)))^{-1}
  \lambda_Y^{(n)}(u_1,\ldots,u_n;X_i^T)
  =
  \1\{u_1,\ldots,u_n\in W\}
  \1\{u_1,\ldots,u_n\notin W\}^{-1}
  \lambda_Y^{(n)}(u_1,\ldots,u_n;X_i^T)$ a.e., 
  which is infinite if any $u_j$ lies in $W_i$ and 0 if any $u_j$ lies outside $W_i$; this would mean the our loss functions would be infinite. 
  Moreover, what we here are dealing with is "extrapolation/interpolation", i.e.~prediction outside the observation domain which is generally a hard task. One may also speculate about whether there emerge other problems here; 
  the folds may have highly varying point counts and such an approach may introduce edge/boundary effects which one may have to correct for \citep{cronie2011edge,CSKWM13,Diggle14Book,BRT15}. Finally, it is not clear how exactly the partitioning of $W$ should be chosen; how many folds there should be, whether they should be equally sized etc.






\subsection{Optimal loss functions, test functions and hyperparameters}
\label{s:DataDrivenLoss}

A burning question is how one should choose the test function family and the CV parameter $p$ here. The current paper has taken a proof-of-concept approach, with the aim of showing that there exist test functions with associated values for $p$ which result in (substantially) improved fits in different contexts. As previously indicated, these may be viewed as hyperparameters and, quite naturally, one may ask how a hyperparameter may be chosen "optimally". 

Turning to the literature, in the context of univariate innovation-based parametric intensity estimation (see Section \ref{s:ParametricIntensity}), more precisely composite likelihood estimation, \citet{guan2015quasi} exploited estimating equation theory to find an optimal test function in the form of a Fredholm integral equation. Further,  \citet{coeurjolly2016towards} gave an extensive account on test function selection in the context of using classical innovations for the estimation (see Section \ref{s:PapangelouUnivariateInnovation}), so-called Takacs-Fiksel estimation. In addition, motivated by \citet{guan2015quasi}, they found a similar Fredholm integral equation approach to finding a (semi-)optimal test function. 
For our purposes, however, since our point process learning framework (in general) cannot be expressed through estimating equations, we believe these findings to be of limited use here.

A further idea for finding an optimal test function is to apply calculus of variations to find a minimiser of the variances in Theorem \ref{thm:InnovationsMean}. 
Aside from the possible associated mathematical intractability of such an approach, the optimality may be model specific and the solution may not solve the issue with the remaining hyperparameters. 
Arguing from the point of view of Taylor expansion, another potential idea would be to let the test function $h_{\theta}$ be given by a combination of derivatives of $\xi_{\theta}^n$.




\subsubsection{Data-driven hyperparameter selection}

A further approach, which we are currently exploring in a parallel paper, embraces the statistical learning paradigm to a larger extent than the approaches indicated above. It is (arguably) also more general in the sense that it deals with hyperparameters in general, and thereby test functions and CV parameters in particular. Its algorithm, which is laid out below, is motivated by an algorithm commonly encountered in the statistical learning literature \citep{james2013introduction}. 
Our aim here is to have a 
completely data-driven approach with a minimal amount of hard-coded parameters.
Roughly speaking, it states that out of a class of potential
loss functions, the one which (in some suitable sense) best predicts the validation sets $\x_i^V$, $i=1,\ldots,k$, will be labelled the optimal one; clearly, in some cases this is equivalent to finding an optimal
family of test functions $\mathcal{H}_{\Theta}$.  

Consider a collection of loss functions of either of the forms
    \begin{align*}
    &\{\mathcal{L}_{\gamma}(\theta;\z,\y) : \z,\y\in\X,
    \theta\in\Theta\},
    \\
    &\{\mathcal{L}_{\gamma}(\theta;\y) : \y\in\X,
    \theta\in\Theta\},
    \end{align*}
    where $\gamma\in\Theta_{\gamma}\subseteq\R^{l_{\gamma}}$,  $l_{\gamma}\geq1$, is the associated hyper-parameter. 
    In addition, consider some measure $\mathcal{G}:\Theta\times\X\to[0,\infty)$ of
    goodness of fit, where small means a good fit. 
    Most naturally here,
    $\mathcal{L}_{\gamma}(\theta;\z,\y)=|\widetilde{\mathcal{I}}_{\xi_{\theta}^n}^{h_{\theta}^{\gamma}}(W^n;\z,\y)|$ if $\xi_{\theta}^n$ and $h_{\theta}$ are
  of the form \eqref{e:GeneralEstimator} or $\mathcal{L}_{\gamma}(\theta;\y)=|\widetilde{\mathcal{I}}_{\xi_{\theta}^n}^{h_{\theta}^{\gamma}}(W^n;\y)|$ if $\xi_{\theta}^n$ and $h_{\theta}$ are
  of the form \eqref{e:GeneralEstimatorNoPP}, 
    where $\mathcal{H}_{\Theta}^{\gamma}=\{h_{\theta}^{\gamma}:\theta\in\Theta\}$,
    $\gamma\in\Theta_{\gamma}$, are different
    families of test functions, say
    $h_{\theta}^{\gamma}(\cdot)=\xi_{\theta}^n(\cdot)^{\gamma}$,
    $\gamma\in\Theta_{\gamma}\subseteq\R$. However, $\gamma$ can also e.g.~incorporate the CV parameter $p$, if the algorithm below is run for different CV settings.
    Similarly, as goodness of fit criterion we may e.g.~consider a univariate
    innovation-based one,
    $\mathcal{G}(\theta;\y)=|\mathcal{I}_{\xi_{\theta}^n}^{h_E}(W^n;\y)|^i$,
    $i\in\{1,2\}$, for some fixed test function $h_E(\cdot)$.

  \begin{remark}
    In a parallel paper we exploit the current data-driven approach in the context of 
    penalisation, most notably regularisation \citep{james2013introduction,yue2015variable,Rajala2018,choiruddin2018convex}. Here we let $\mathcal{L}_{\gamma}(\theta;\z,\y) =
    |\widetilde{\mathcal{I}}_{\xi_{\theta}^n}^{h_{\theta}}(W;\z,\y)| +
    \gamma R(\theta)$ or 
    $\mathcal{L}_{\gamma}(\theta;\y) =
    |\widetilde{\mathcal{I}}_{\xi_{\theta}^n}^{h_{\theta}}(W;\y)| +
    \gamma R(\theta)$, where  $\mathcal{H}_{\Theta}=\{h_{\theta}:\theta\in\Theta\}$ is a fixed family of test functions,
    $\gamma\geq0$ and $R(\cdot)$ is a suitable penalty, e.g.~a  regularisation
    penalty such as the elastic-net penalty 
    \citep{zou2005regularization}. 
    Another penalisation setting which is parallely being looked into is smoothness penalisation in the context of non-parametric intensity estimation;  specifically, the case of soap film smoothing \citep{wood2008soap} is currently being studied. 
\end{remark}

The general approach described in the algorithm below illustrates how an ``optimal'' $\mathcal{L}_{\gamma}$ may be found in a data-driven way.

\begin{enumerate}

\item When the loss functions are of the form $\mathcal{L}_{\gamma}(\theta;\z,\y)$, as in Section \ref{s:TestSet}, consider triples
  $(\x_i^T,\x_i^V,\x_i^E)$, $i=1,\ldots,k$, which have been obtained
  through two layers of CV-partitioning of $\x$; here $\x_i^E$,
  $i=1,\ldots,k$, are generated by a retention probability
  $p_E\in(0,1)$ and $(\x_i^T,\x_i^V)$,
  $\x_i^T\cup\x_i^V=\x\setminus\x_i^E$, $i=1,\ldots,k$, are generated
  by a retention probability $p\in(0,1)$, both using some of the
  CV-approaches in Definition \ref{def:MCCV} and Definition 
  \ref{def:kFoldMCCV}. 
  
  When the loss functions are of the form $\mathcal{L}_{\gamma}(\theta;\y)$, we instead consider CV-splittings $(\x_i^T,\x_i^V)$,
  $\x_i^V=\x\setminus\x_i^T$, $i=1,\ldots,k$, using one of the
  CV-approaches in Definition \ref{def:MCCV} and \ref{def:kFoldMCCV} with some retention probability $p\in(0,1)$. 
    
  \item For each $i=1,\ldots,k$:
    \begin{enumerate}
    \item For each $\gamma\in\Theta_{\gamma}$, find a minimiser $\widehat\theta_i(\gamma)\in\Theta$ of
        $\mathcal{L}_{\gamma}(\theta;(\x_i^T,\x_i^V))$ or $\mathcal{L}_{\gamma}(\theta;\x_i^T)$.
        
      \item Consider the $i$th "prediction error path"
        $\mathcal{G}_i(\gamma)=\mathcal{G}(\widehat\theta_i(\gamma);\x_i^E)$,
        $\gamma\in\Theta_{\gamma}$, or $\mathcal{G}_i(\gamma)=\mathcal{G}(\widehat\theta_i(\gamma);\x_i^V)$,
        $\gamma\in\Theta_{\gamma}$ (when $\mathcal{G}$
        is innovation-based, in keeping with
        \citet{baddeley2005residual}, these should be called
        residuals).
      \end{enumerate}
    
    \item Let $\widehat\gamma$ be a minimiser of
      $\gamma\mapsto\frac{1}{k}\sum_{i=1}^k
      \mathcal{G}_i(\gamma)$,
      $\gamma\in\Theta_{\gamma}$.
    
    \item The optimal loss function is given by
      $\mathcal{L}_{\widehat\gamma}(\theta;\z,\y)$,
      $\z,\y\in\X$, $\theta\in\Theta$, or $\mathcal{L}_{\widehat\gamma}(\theta;\y)$, 
      $\y\in\X$, $\theta\in\Theta$.
    
    \item Generate a new collection of training-validation splittings
      $(\x_i^T,\x_i^V)$, $i=1,\ldots,k$, of the full dataset $\x$
      (possibly using a different $k\geq1$ and a different retention probability $p\in(0,1)$)
      and consider some test function generated by a combination of $\mathcal{L}_{\widehat\gamma}(\theta;\x_i^T)$, $i=1,\ldots,k$, or $\mathcal{L}_{\widehat\gamma}(\theta;\x_i^V,\x_i^T)$, $i=1,\ldots,k$, e.g.~as was done in Section \ref{s:FixedTestFunctionLoss}, or, in the former case, the ($p\to0$) limiting case,  $\mathcal{L}_{\widehat\gamma}(\theta;\x)$. 
      We obtain the final estimate $\widehat\theta$ by minimising
      this loss function.
      

  \end{enumerate}

A further option is to (additionally) let the choice of $\xi_{\theta}^n$ vary with $\gamma$ as well, e.g.~if we want to compare different models which are represented by different families $\Xi_{\Theta}^n$ -- in some sense we here would jointly carry out fitting, model selection and model validation. In the context of large and complex datasets, one could e.g.~think of hybrid Gibbs processes \citep{baddeley2013hybrids} or pairwise interaction processes \citep{VanLieshoutBook} with an increasing number of components.

\section{Acknowledgements}
We are grateful to Achmad Choiruddin (Institut Teknologi Sepuluh Nopember, Indonesia) for feedback on an early draft and to Mathew Penrose (University of Bath, UK) for supplying us with details on the proof of \citet[Exercise 5.9]{last2017lectures}. 

\bibliographystyle{dcu}
\bibliography{main.bib}

\clearpage
\appendix
\appendixpage

\section{Univariate innovations}
\label{s:UnivariateInnovations}

In the subsequent subsections we will consider different settings
where the statistical inference fits into the univariate innovations framework. We
stress that there may be existing approaches which do not fit into
this framework and, also, that there may be approaches which do fit in but
which (accidentally) have not been included; we have focused on
approaches which deal with inhomogeneous point processes.

\subsection{Non-parametric intensity estimation}
\label{s:NonParametrics}
As the intensity function $\rho(u)$, $u\in W$, governs the univariate
distributions of the points of $X\cap W$, it is not hard to imagine
that its estimation forms the very foundation of subsequent
statistical analyses. When we do not have access to any covariates
(on $W$), we have to resort to non-parametric estimation.  Although a
considerable amount of progress has been made regarding non-parametric estimation of
intensity functions
\citep{vanIntensity,BRT15,cronie2018bandwidth,moradi2018spatial,Moradi2019,van2020infill},
there is still significant room for improvement.

We may represent non-parametric intensity estimation using the general
estimator structure in \eqref{e:GeneralEstimator}:
$$
\widehat\rho_{\theta}(u,\y)=\xi_{\theta}^1(u;\y), \quad u\in W, \quad
\y\in\X,
$$
where $\theta\in\Theta\subseteq\R^l$, $l\geq1$, is the
tuning/smoothing parameter associated to the estimator.
Given an observed point pattern $\x$, to obtain an estimate/selection
$\widehat\theta_W(\x)$
we consider $\widehat\rho_{\theta}(u,\x)=\xi_{\theta}^1(u;\x)$ and
minimise
$\mathcal{L}(\theta;\x) =
\mathcal{I}_{\widehat\rho_{\theta}}^{h_{\theta}}(W;\x)^2$,
$\theta\in\Theta$,
where
\[
\mathcal{I}_{\widehat\rho_{\theta}}^{h_{\theta}}(W;\x)
=
\sum_{x\in\x\cap W}
h_{\theta}(x;\x\setminus\{x\})
-
\int_{A}
h_{\theta}(u;\x)
\widehat\rho_{\theta}(u;\x)
\de u,
\]
using some family of test functions $h_{\theta}$, $\theta\in\Theta$. We
then plug $\widehat\theta_W(\x)$ into the intensity estimator,
i.e.~$\widehat\rho_{\widehat\theta_W(\x)}(u,\x)$, $u\in W$, to obtain
the final estimate of the intensity function.

\begin{remark}
  If the estimator has no tuning parameter $\theta$, which e.g.~is the
  case for the Voronoi intensity estimator
  \citep{barr2010voronoi,Moradi2019}
\begin{align*}
  \widehat\rho_{\theta}(u,\y)
  &=
    \sum_{x\in\y}\frac{\1\{u\in\mathcal{V}_x\}}{|\mathcal{V}_x|}, 
    \quad u\in W,
  \\
  \mathcal{V}_x&=\{v\in W:d(v,x)\leq d(v,y)\text{ for all }y\in\y\setminus\{x\}\}, 
                 \quad x\in\y,
\end{align*}
we can simply set $\widehat\rho_{\theta_1}=\widehat\rho_{\theta_2}$
for any $\theta_1,\theta_2\in\Theta$ or $\Theta=\emptyset$. In the
case of the so-called Resample-smoothed Voronoi intensity estimator of
\citet{Moradi2019} on the other hand, there are two tuning parameters
to be selected.
\end{remark}


\subsubsection{Kernel estimation}
The arguably most prominent intensity estimation technique is kernel
estimation \citep{IPSS08,vanIntensity,Diggle14Book,BRT15}. E.g., when $W\subseteq S=\R^d$, $d\geq1$, in its
most typical form a kernel estimate is given by
\begin{align}
    \label{e:Kernel}
  \widehat\rho_{\theta}(u,\x)
  =\sum_{x\in\x}
  \frac{\kappa_{\theta}(u-x)}{w_{\theta}(u,x)}
  =\sum_{x\in\x}
  \frac{\theta^{-d}\kappa((u-x)/\theta)}{w_{\theta}(u,x)}, 
  \quad u\in W, 
\end{align}
where the kernel $\kappa$ is a symmetric density function on $\R^d$,
$w_{\theta}(u,x)$ is an edge correction term which compensates for
possible interaction between points inside and outside $W$, and the
smoothing parameter $\theta\in\Theta=(0,\infty)$ is the so-called
bandwidth.  A common edge correction factor is the local factor
$w_{\theta}(u,x)=\int_W\kappa_{\theta}(v-x)\de v$, which also ensures
that $\int_W\widehat\rho_{\theta}(u,\x)\de u = \#\x$.  It should be
emphasised that for point processes in other (non-Euclidean) domains
$S$, kernel functions $\kappa$, and thereby kernel estimators, may
generally take other forms and can be quite abstract entities \citep{pelletier2005kernel,mcswiggan2017,moradi2018spatial,mateu2019spatio}. In certain cases, however, there are straightforward extensions of \eqref{e:Kernel}; for a point pattern on a linear network $S=L$, we may e.g.~consider the kernel estimate obtained by leaving the kernel unchanged in \eqref{e:Kernel} and simply replacing the edge correction $w_{\theta}(u,x)$ by $w_{\theta}(u,x)=\int_L\kappa_{\theta}(v-x)\de v$, where the integration is with respect to arc length, i.e.~1-dimensional Hausdorff measure in $\R^2$  \citep{rakshit2019fast}.

In kernel estimation, the main challenge is optimal bandwidth selection; often the
choice of kernel is of much less importance than the choice of bandwidth \citep{S86}. To the best of our
knowledge, the two best-performing approaches currently available for point processes in $\R^d$ are
the Poisson process likelihood leave-one-out cross-validation approach
\citep{Load99,BRT15} and the Campbell formula-based approach of
\citet{cronie2016bandwidth,cronie2018bandwidth}. Given a point pattern $\x\subseteq W$,
the Poisson process likelihood leave-one-out cross-validation approach
maximises
\begin{align}
\label{e:PPL}
  \theta\mapsto
  \sum_{x\in \x\cap W}\log\widehat\rho_{\theta}(x,\x\setminus\{x\})
  -
  \int_W \widehat\rho_{\theta}(u,\x)\de u
\end{align}
in order to obtain an optimal bandwidth; the Poisson process
log-likelihood function is a direct consequence of
\eqref{e:PoissonLogLikelihood}.  To express this as a loss function,
we may either multiply \eqref{e:PPL} by $-1$ or consider its
derivative with respect to $\theta$, assuming sufficient
differentiability; the latter results in a univariate innovation of
the kind given in Definition \eqref{def:Innovations}. This approach is
particularly suited when the data come from a Poisson process
\citep{cronie2018bandwidth}.  Moreover, \citet{cronie2018bandwidth}
proposed to select the bandwidth by minimising the squared innovation
\[
\mathcal{L}(\theta)=
\mathcal{I}_{\widehat\rho_{\theta}}^{f(\widehat\rho_{\theta}(\cdot))}(W;\x)^2
=
\left(
\sum_{x\in \x\cap W}f(\widehat\rho_{\theta}(x,\x\setminus\{x\}))
-
\int_W f(\widehat\rho_{\theta}(u,\x))\widehat\rho_{\theta}(u,\x)\de u
\right)^2,
\]
where 
$f:\R\to\R$. They studied in detail the choice $f(x)=1/x$, $x\in\R$, which
results in the so-called inverse/Stoyan-Grabarnik
statistic/diagnostic, and found that 
\begin{align}
\label{e:CvL}
  \mathcal{L}(\theta)=
  \left(
  \sum_{x\in \x\cap W}f(\widehat\rho_{\theta}(x,\x))
  -
  \int_W f(\widehat\rho_{\theta}(u,\x))\widehat\rho_{\theta}(u,\x)\de u
  \right)^2
\end{align}
with the choice $f(x)=1/x$, $x\in\R$, gave rise to \eqref{e:CvL} being
the square of a (conjectured) monotonic function of $\theta\geq0$ when
using $w_{\theta}(u,x)\equiv1$ (no edge correction) and a Gaussian
kernel; the unmodified version is not monotonic. Note that $f(x)=1/x$,
$x\in\R$, sets the integral in \eqref{e:CvL} to $|W|$ when
$\widehat\rho_{\theta}(u,\x)>0$, $u\in W$, so we choose the bandwidth
by estimating the (known) size of the study region by the sum of
reciprocal intensity estimates.
\citet{Moradi2019} indicated that this particular choice of test
function favours a low variance with respect to the variance-bias
trade off, which makes it particularly suited for aggregated point
processes.


\subsection{Parametric product density and Palm-likelihood estimation}
\label{s:ParametricIntensity}

The general idea in parametric product density estimation is to fit
some parametric family of functions
$\rho_{\theta}^{(n)}(u_1,\ldots,u_n)=\xi_{\theta}^1(u_1,\ldots,u_n)\geq0$,
$u_1,\ldots,u_n\in W$, $\theta\in\Theta$, to the observed point
pattern $\x\subseteq W$. The corresponding innovations
$\mathcal{I}_{\rho_{\theta}^{(n)}}^{h_{\theta}}(\cdot;\y)$,
$\theta\in\Theta$, $\y\in\X$, and loss function are given by
\begin{align}\label{e:QuasiLikProdDens}
  \mathcal{I}_{\rho_{\theta}^{(n)}}^{h_{\theta}}(A;\y)
  =&
     \sum_{(x_1,\ldots,x_n)\in\y_{\neq}^n\cap A}
     h_{\theta}(x_1,\ldots,x_n)
  \\
   &-
     \int_{A}
     h_{\theta}(u_1,\ldots,u_n)
     \rho_{\theta}^{(n)}(u_1,\ldots,u_n)
     \de u_1\cdots\de u_n,
     \quad 
     A\subseteq S^n,
     \nonumber
  \\
  \mathcal{L}(\theta)
  =&
     \mathcal{I}_{\rho_{\theta}^{(n)}}^{h_{\theta}}(W^n;\x)^2,
     \nonumber
\end{align}
respectively, and by minimising $\mathcal{L}(\theta)$,
$\theta\in\Theta$, we consequently obtain an estimate
$\widehat\theta=\widehat\theta_W(\x)\in\Theta$ of $\theta_0$, which
yields the final intensity estimate
$\rho_{\widehat\theta_W(\x)}^{(n)}(u_1,\ldots,u_n)$,
$u_1,\ldots,u_n\in W$.

\subsubsection{Intensity estimation}
In particular, when $n=1$, i.e.~when we are carrying out parametric
intensity estimation by means of some parametric family of intensity
functions $\rho_{\theta}(u)=\xi_{\theta}^1(u)\geq0$, $u\in W$,
$\theta\in\Theta$, we obtain the innovations
\begin{align}\label{e:QuasiLik}
  \mathcal{I}_{\rho_{\theta}}^{h_{\theta}}(A;\y)
  =&\sum_{x\in\y\cap A}
     h_{\theta}(x)
     -
     \int_{A}
     h_{\theta}(u)
     \rho_{\theta}(u)
     \de u,
     \quad 
     A\subseteq S,
     \quad 
     \theta\in\Theta,
     \quad \y\in\X.
\end{align}
Viewed as an unbiased estimating equation approach (recall Lemma
\ref{lemma:InnovationUnbiased}), the corresponding loss function here
essentially corresponds to the quasi-likelihood approach of
\citet{guan2015quasi}. 
Moreover, when one has the access to a vector of
$l\geq1$ covariates $z(u)=( z_1(u),\ldots,z_l(u))^\top$ measured at
location $u\in W$, it is often natural to model the intensity function
parametrically through these covariates.
The arguably most common and popular model choice for the intensity
function is the log-linear form
\begin{align}
\label{eq:int}
\rho_{\theta}(u)=\exp\{\theta^\top z(u)\}, \quad u \in W, 
\end{align}
where $\theta=(\theta_1,\ldots,\theta_l)^\top\in\Theta\subseteq\R^l$
is a real-valued $l$-dimensional parameter \citep{coeurjolly2019understanding}. Note that an
intercept may be included by setting $z_1(\cdot)\equiv 1$ and, in
addition, (projections of) the location $u\in W$ itself may also be
included in $z(u)$ if explicit spatial dependence is warranted.  When
we assume that $X$ is a Poisson process, the standard procedure is to
maximise the log-likelihood function (see
\eqref{e:PoissonLogLikelihood}) with respect to $\theta\in\Theta$,
which for the intensity model~\eqref{eq:int} is given by
\begin{align}\label{eq:likepois}
  \ell(\theta;\x) 
= \sum_{x\in\x \cap W} \log\rho_{\theta}(u) - \int_{W} \rho_{\theta}(u)\de u
= \sum_{x\in\x \cap W} \theta^\top z(u) - \int_{W} \exp\{\theta^\top z(u)\}\de u.
\end{align}
This has gradient
\begin{equation}\label{eq:Un}
  \nabla\ell(\theta;\x) 
  =
  \sum_{u \in \x\cap W}z(u) - \int_{W} z(u) \rho_{\theta}(u)\de u\in\R^l,
\end{equation}
and solving $\nabla\ell(\theta;\x) =0 \in \R^l$ is equivalent to
minimising an $l$-dimensional vector of squared innovations of the
form \eqref{e:QuasiLik}, where entry $i=1\ldots,l$ is based on the
test function $h_i(u)=z_i(u)$, $u\in W$. Even when $X$ is not a
Poisson process, we obtain that \eqref{eq:Un} remains an unbiased
estimating equation (recall Lemma
\ref{lemma:InnovationUnbiased}). Hence, the maximum of \eqref{eq:likepois} still makes sense for non-Poissonian models, an the obtained estimator may be interpreted as a composite likelihood
estimator \citep{waagepetersen2007estimating,guan2010weighted,guan2015quasi}.

\subsubsection{Palm-likelihood estimation}

In certain cases, the go-to approach is Palm-likelihood estimation
\citep{BRT15,coeurjolly2019understanding}. Following \citet{prokevsova2017two}, for a SOIRS point
process $X$, one instance of log-Palm-likelihood estimation is
obtained by maximising
\begin{align*}
  \ell_P(\theta;\x)=&
                      \mathop{\sum\nolimits\sp{\ne}}_{x_1,x_2\in\x\cap W}
                      \1\{x_1\in W_{\ominus r}\}
                      \1\{d(x_1,x_2)\leq r\}
                      \log\rho_{\theta}(x_1)\rho_{\theta}(x_2)
                      g_{\theta}^{(2)}(x_1-x_2)
  \\
                    &-
                      \int_W\int_W
                      \1\{u\in W_{\ominus r}\}
                      \1\{d(u_1,u_2)\leq r\}
                      \rho_{\theta}(u_1)\rho_{\theta}(u_2)g_{\theta}^{(2)}(u_1-u_2)
                      \de u
\end{align*}
for some $r>0$, where $W_{\ominus r}$ is an $r$-trimming of $W$, and
$\rho_{\theta}(\cdot)$ and $g_{\theta}^{(2)}(\cdot)$ are the
intensity function and the pair correlation function of the model parametrised by $\theta\in\Theta$. 
\citet{prokevsova2017two} considered a two-step procedure where the
intensity was estimated externally and then plugged into the above
log-Palm-likelihood function, thus replacing $\rho_{\theta}(\cdot)$.
The name Palm-likelihood stems from the original approach \citep{ogata1991maximum} for
stationary point processes in $\R^d$, where the idea is to apply
Poisson process likelihood heuristics to the intensity functions of
the first order reduced Palm processes $X_{u}^!$; recall expression
\eqref{e:ReducedPalmIntensity}.  Assuming sufficient
differentiability, if we let
\begin{align*}
  \xi_{\theta}^2(u_1,u_2)
  =&
     \rho_{\theta}(u_1)\rho_{\theta}(u_2)g_{\theta}^{(2)}(u_1-u_2),
  \\
  h_{\theta}(u_1,u_2)
  =&
     \1\{u\in W_{\ominus r}\}
     \1\{d(u_1,u_2)\leq r\}
     \frac{\nabla(\rho_{\theta}(u_1)\rho_{\theta}(u_2)g_{\theta}^{(2)}(u_1-u_2))}
     {\rho_{\theta}(u_1)\rho_{\theta}(u_2)g_{\theta}^{(2)}(u_1-u_2)},
\end{align*}
then the corresponding score function is given by the $l$-dimensional
vector of innovations
\begin{align*}
  &\nabla\ell_P(\theta;\x)
    =
    \mathcal{I}_{\xi_{\theta}^2}^{h_{\theta}}(W^2;\x)=
  \\
  =&
     \mathop{\sum\nolimits\sp{\ne}}_{x_1,x_2\in\x\cap W}
     \1\{x_1\in W_{\ominus r}\}
     \1\{d(x_1,x_2)\leq r\}
     \frac{\nabla(\rho_{\theta}(x_1)\rho_{\theta}(x_2)g_{\theta}^{(2)}(x_1-x_2))}
     {\rho_{\theta}(x_1)\rho_{\theta}(x_2)g_{\theta}^{(2)}(x_1-x_2)}
  \\
  &-
    \int_W\int_W
    \1\{u\in W_{\ominus r}\}
    \1\{d(u_1,u_2)\leq r\}
    \nabla(\rho_{\theta}(u_1)\rho_{\theta}(u_2)g_{\theta}^{(2)}(u_1-u_2))
    \de u
    \in\R^l,
\end{align*}
provided that we can interchange differentiation and
integration. Hence, carrying out Palm-likelihood is equivalent to
minimising a vector of squared innovations of the form
\eqref{e:QuasiLikProdDens}.

\begin{remark}
  One could, potentially, consider a similar approach for other
  (higher-order) reduced Palm product densities/correlation functions,
  where other (higher-order) innovations would be exploited.
\end{remark}

\subsection{Papangelou conditional intensity estimation}
\label{s:PapangelouUnivariateInnovation}
The idea here is to assume that the observed point pattern
$\x\subseteq W$ is a realisation of a point process $X$ with
Papangelou conditional intensity $\lambda_{\theta_0}$,
$\theta_0\in\Theta$, which is a member of some parametric family of
Papangelou conditional intensity functions $\Xi_{\Theta}^1=
\{\xi_{\theta}^1:\theta\in\Theta\}=\{\lambda_{\theta}:\theta\in\Theta\}$,
$\Theta\subseteq\R^l$, $l\geq1$.  
Hence, we obtain classical innovations and the
innovation-based loss function
\begin{align*}
  \mathcal{L}(\theta)
  =
  \mathcal{I}_{\lambda_{\theta}}^{h_{\theta}}(W;\x)^2
  =&
     \left(
     \sum_{x\in\x\cap W}
     h_{\theta}(x;\x\setminus\{x\})
     -
     \int_{W}
     h_{\theta}(u;\x)
     \lambda_{\theta}(u;\x)
     \de u
     \right)^2,
     \quad
     \theta\in\Theta,
\end{align*}
for some suitable family of test functions $\mathcal{H}_{\Theta}$. In the
literature, this is sometimes referred to as Takacs–Fiksel estimation
\citep[see e.g.][and the references
therein]{moller2017some,coeurjolly2019understanding}.  The most
prominent special case hereof, namely pseudo-likelihood estimation
\citep{besag1974spatial,ripley1991statistical,VanLieshoutBook,MW04,BRT15}, is obtained by setting
$h_{\theta}(\cdot)
=\nabla\lambda_{\theta}(\cdot)/\lambda_{\theta}(\cdot)\in\R^{l}$,
where $\nabla\lambda_{\theta}(\cdot)$ is the (well-defined) gradient
of $\lambda_{\theta}$. Note that this yields the score function of the
(log-)pseudo-likelihood function, which in turn may be
expressed as the loss function
\begin{align*}
  \mathcal{L}(\theta)
  =&
     \int_{W}
     \lambda_{\theta}(u;\x)
     \de u
     -
     \sum_{x\in\x\cap W}
     \log\lambda_{\theta}(x;\x\setminus\{x\})
     ,
     \quad
     \theta\in\Theta.
\end{align*}

\subsection{Minimum contrast estimation of second-order summary statistics}
\label{s:MinimumContrast}

We next turn to the case where the general parametrised estimator
family is of the form \eqref{e:GeneralEstimatorNoPP} and is given by
$\xi_{\theta}^2=\rho_{\theta}^{(2)}$, $\theta\in\Theta$, for some
family of second-order product densities. Assuming that the point
process $X$ in $S=\R^d$
has the second-order product density $\rho_{\theta_0}^{(2)}$ for some
$\theta_0\in\Theta$, we consider the test function
$$
h(u_1,u_2;r)= \1\{d(u_1,u_2)\leq r\}/(|W|\rho(u_1)\rho(u_2)), \quad
u_1,u_2\in W,
$$
where $r\geq0$, $\rho(u)$, $u\in W$, is the true (unknown) intensity
of $X$ and $d(u_1,u_2)=\|u_1-u_2\|_2$ is the Euclidean metric on
$\R^d$. This yields the innovations
\begin{align}
\label{e:Kinnovation}
  \mathcal{I}_{\rho_{\theta}^{(2)}}^{h(\cdot;r)}(W\times\R^d;X)
  =&
     \frac{1}{|W|}
     \sum_{(x_1,x_2)\in X_{\neq}^2\cap(W\times\R^d)}
     \frac{
     \1\{d(x_1,x_2)\leq r\}
     }{\rho(x_1)\rho(x_2)}
  \\
   &-
     \frac{1}{|W|}
     \int_W\int_{b(u_1,r)}
     g_{\theta}^{(2)}(u_1,u_2)
     \de u_1\de u_2
\quad\theta\in\Theta,
\nonumber
\end{align}
where $b(u,r)$ denotes the closed Euclidean $r$-ball around $u\in\R^d$
and
$g_{\theta}^{(2)}(u_1,u_2)=\rho_{\theta}^{(2)}(u_1,u_2)/(\rho(u_1)\rho(u_2))$
is the pair correlation function corresponding to $\theta$.  If the
true second-order product density of $X$ is given by
$\rho_{\theta_0}^{(2)}$ for some $\theta_0\in\Theta$, by Lemma
\ref{lemma:InnovationUnbiased} we obtain that
$\E[\mathcal{I}_{\rho_{\theta_0}^{(2)}}^{h}(W;X)]=0$.  Under the
assumption of second-order intensity reweighted stationarity (recall Section \ref{s:ProductDensities}), i.e.~when the intensity is positive and the pair
correlation functions satisfy
$g_{\theta}^{(2)}(u_1,u_2)=g_{\theta}^{(2)}(u_1-u_2)$,
$u_1,u_2\in\R^d$, $\theta\in\Theta$,
we have that the integral term in \eqref{e:Kinnovation} becomes
$$
K_{\rm inhom}(r;\theta)
=\int_{b(v,r)}g_{\theta}^{(2)}(u)\de u
$$
for any $W\subseteq\R^d$, $|W|>0$, and any $\theta\in\Theta$; these
functions are constant as functions of $v\in\R^d$.  Note that the
function $K_{\rm inhom}(r)=K_{\rm inhom}(r;\theta_0)$, $r\geq0$, is
the inhomogeneous $K$-function of $X$ \citep{InhomK2000}, which a
measure of interaction. For a Poisson process $X$ we have that
$K_{\rm inhom}(r)=|b(v,r)|$, $v\in\R^d$, so
$K_{\rm inhom}(r)>|b(v,r)|$ indicates clustering/aggregation between
points with inter-point distance at most $r$ whereas
$K_{\rm inhom}(r)<|b(v,r)|$ indicates inhibition/regularity.

In practice one considers
$\mathcal{I}_{\rho_{\theta}^{(2)}}^{h(\cdot;r)}(W\times W;X)$ because
we only observe events within $W$.  Consequently, one needs to include
an edge-correction/normalisation factor $e(\cdot,\cdot)$ in the test
function. This is done by simply changing the test function to
\begin{align*}
  \widetilde h_e(u_1,u_2;r)=e(u_1,u_2)h(u_1,u_2;r).
\end{align*}
An edge-correction factor adjusts for (unobserved) interactions
between events in $X\cap W$ and those in $X\cap \R^d\setminus
W$. 
Ideally, including $e(\cdot,\cdot)$ will ensure that
$\mathcal{I}_{\rho_{\theta}^{(2)}}^{\widetilde h(\cdot;r)}(W^2;X)$ has
expectation 0. For further details, see e.g.~\citep{InhomK2000,MW04,BRT15}. When employing
$\mathcal{I}_{\rho_{\theta}^{(2)}}^{\widetilde h_e(\cdot;r)}(W^2;X)$,
i.e.~using $\widetilde h_e$ as test function, we refer to the sum-term
in \eqref{e:Kinnovation} as a non-parametric estimator of
$K_{\rm inhom}(r)$.

Assuming second-order intensity reweighted stationarity, by minimising the loss function
$$
\mathcal{L}(\theta)=\int_0^{r_{max}}\mathcal{I}_{\rho_{\theta}^{(2)}}^{\widetilde h_e(\cdot;r)}(W^2;X)^2\de r, \quad 
\theta\in\Theta, 
$$
for some $r_{max}>0$, we obtain a minimum contrast estimator
$\widehat\theta_W(X)$ of $\theta_0$. There are also further
developments of this concept, where metrics other than the
$L_2$-distance are used to measure discrepancies between
$K_{\rm inhom}(r;\theta)$, $r\geq0$, and its estimator
\citep{BRT15,Diggle14Book}.

\section{Proofs}
\label{s:Proofs}

\subsection{Proof of Theorem \ref{lemma:Thinning}}

\begin{proof}[Proof of Theorem \ref{lemma:Thinning}]
Starting with expression \eqref{e:ThinningProdDens}, the form of the Papangelou conditional intensity is a direct consequence of \citet[Theorem 4.7]{decreusefond2018stein} and \eqref{e:nthPapangelou}. The result on the product densities follows from e.g.~\citet[Proposition 2.3.24]{baccelli2020}.

The structure of the proof of the prediction formula in \eqref{eq:thinning Z vs Y} follows the lines of the proof of  \citet[Exercise 5.9]{last2017lectures}. 
Consider the random measure representation of $X$, where there are random variables $N =X(S)\in \{0,\ldots,\infty \}$ and $X_1, \ldots X_N\in S$ such that $X(A) = \sum_{i=1}^{N} \delta_{X_i}(A)=\sum_{i=1}^{N} \1\{X_i\in A\}$, $A\subseteq S$.  
An independent thinning of $X$ has the same distribution as 
$Z(\cdot)=\sum_{i=1}^{N} B_i \delta_{X_i}(\cdot)$, where, conditionally on $N$ and $X_1,\ldots,X_N$, the random variables $B_1,\ldots, B_N$ are
mutually independent and, for any $i=1,\ldots,N$, conditionally on $X_i$, the random variable $B_i$ is Bernoulli distributed with parameter $p(X_i)$. Similarly, $Y=X\setminus Z$ has the random measure representation $Y(\cdot)=X(\cdot)-Z(\cdot)= \sum_{i=1}^{N} (1-B_i) \delta_{X_i}(\cdot)$. 
For any $m\geq n$, 
let $\mathcal{A}_m$ be the set of all $n$-tuples of distinct 
integers $i_1,\ldots,i_n\in \{1,\ldots, m\}$;
if $m$ is infinite, 
we let $i_1,\ldots,i_n$ be finite. 
It now follows that 
\begin{align*}
    &\E\left[\mathop{\sum\nolimits\sp{\ne}}_{x_1,\ldots,x_n\in Z}
    h(x_1,\ldots,x_n, Y)
    \prod_{i=1}^n(1-p(x_i))\right]
    =\\
    =&
    \E \left[\sum_{i_1,\ldots,i_n \in \mathcal{A}_N}
    h(X_{i_1},\ldots,X_{i_n}, Y)
    \prod_{j=1}^{n} B_{i_j} (1-p(X_{i_j})) 
    \right] \\
    =& 
    \E \left[\sum_{i_1,\ldots,i_n \in \mathcal{A}_N}
    h(X_{i_1},\ldots,X_{i_n}, Y)
    \prod_{j=1}^{n}\E[1-B_{i_j}| X] 
    \prod_{j=1}^{n} B_{i_j} 
    \right]
    \\
    =& 
    \E\left[\sum_{i_1,\ldots,i_n \in \mathcal{A}_N}
    h(X_{i_1},\ldots,X_{i_n}, Y\setminus \{X_{i_1},\ldots, X_{i_n}\})
    \prod_{j=1}^{n}\E[1-B_{i_j}| X] 
    \prod_{j=1}^{n} B_{i_j} 
    \right],
\end{align*}
where we have used that $p(X_{i_j})=\E[B_{i_j} | X]$ and that $Y\cap\{ X_{i_1},\ldots, X_{i_n} \}=\emptyset$, i.e.~$Y= Y \setminus \{ X_{i_1},\ldots, X_{i_n} \}$, when $B_{i_j}=1$ for all $j=1,\ldots,n$. 
By the conditional independence of the $B_{i_j}$'s, we have that 
$\prod_{j=1}^{n}\E[1-B_{i_j}| X]
=
\E[\prod_{j=1}^{n}(1-B_{i_j})| X]$, and writing 
$\tilde{h}_{i_1,\ldots,i_n}(X,Y) =  h(X_{i_1},\ldots,X_{i_n}, Y \setminus \{ X_{i_1},\ldots, X_{i_n} \})$, we obtain
 \begin{align*}
    &\E\left[\sum_{i_1,\ldots,i_n \in \mathcal{A}_N}
    h(X_{i_1},\ldots,X_{i_n}, Y\setminus \{X_{i_1},\ldots, X_{i_n}\})
    \prod_{j=1}^{n}\E[1-B_{i_j}| X] 
    \prod_{j=1}^{n} B_{i_j} 
    \right]=
    \\
    =&
    \E\left[
    \sum_{i_1,\ldots,i_n \in \mathcal{A}_N}
    \tilde{h}_{i_1,\ldots,i_n}(X,Y) \left.\E\left[\prod_{j=1}^{n}(1-B_{i_j})\right| X\right] 
    \prod_{j=1}^{n} B_{i_j} 
    \right] \\
    =&
    \sum_{i_1,\ldots,i_n \in \mathcal{A}_\infty}
    \E \left[
    \1\{N \geq \max\{i_1,\ldots,i_n\}\}
    \tilde{h}_{i_1,\ldots,i_n}(X,Y) \left.\E\left[\prod_{j=1}^{n}(1-B_{i_j})\right| X\right]
    \prod_{j=1}^{n} B_{i_j} 
    \right] \\
    =&
    \sum_{i_1,\ldots,i_n \in \mathcal{A}_\infty}
    \E\left[
    \left.
    \E\left[
    \1\{N \geq \max\{i_1,\ldots,i_n\}\}
    \tilde{h}_{i_1,\ldots,i_n}(X,Y)
    \prod_{j=1}^{n}(1-B_{i_j})\right| X\right]
    \prod_{j=1}^{n} B_{i_j} 
    \right],
\end{align*}
where the last equality follows from the "pulling out known factors" property of conditional expectations; $N$ and $\tilde{h}_{i_1,\ldots,i_n}(X,Y)$ are measurable with respect to the $\sigma$-algebra generated by $X$. 
By the law of total expectation it follows that 
\begin{align*}
    &\sum_{i_1,\ldots,i_n \in \mathcal{A}_\infty}
    \E\left[
    \left.
    \E\left[
    \1\{N \geq \max\{i_1,\ldots,i_n\}\}
    \tilde{h}_{i_1,\ldots,i_n}(X,Y)
    \prod_{j=1}^{n}(1-B_{i_j})\right| X\right]
    \prod_{j=1}^{n} B_{i_j} 
    \right]
    =\\
    =&
    \sum_{i_1,\ldots,i_n \in \mathcal{A}_\infty}
    \E\left[
    \left.
    \E\left[
    \left.
    \E\left[
    \1\{N \geq \max\{i_1,\ldots,i_n\}\}
    \tilde{h}_{i_1,\ldots,i_n}(X,Y)
    \prod_{j=1}^{n}(1-B_{i_j})\right| X\right]
    \prod_{j=1}^{n} B_{i_j} 
    \right| X\right]
    \right]\\
    =&
    \sum_{i_1,\ldots,i_n \in \mathcal{A}_\infty}
    \E\left[
    \left.
    \E\left[
    \1\{N \geq \max\{i_1,\ldots,i_n\}\}
    \tilde{h}_{i_1,\ldots,i_n}(X,Y)
    \prod_{j=1}^{n}(1-B_{i_j})\right| X\right]
    \left.
    \E\left[
    \prod_{j=1}^{n} B_{i_j} 
    \right| X\right]
    \right]
    \\
    =&
    \sum_{i_1,\ldots,i_n \in \mathcal{A}_\infty}
    \E\left[
    \left.
    \E\left[
    \1\{N \geq \max\{i_1,\ldots,i_n\}\}
    \tilde{h}_{i_1,\ldots,i_n}(X,Y)
    \prod_{j=1}^{n}(1-B_{i_j})\right| X\right]
    \prod_{j=1}^{n} p(X_{i_j})
    \right]
    \\
    =&
    \sum_{i_1,\ldots,i_n \in \mathcal{A}_\infty}
    \E\left[
    \left.
    \E\left[
    \1\{N \geq \max\{i_1,\ldots,i_n\}\}
    \tilde{h}_{i_1,\ldots,i_n}(X,Y)
    \prod_{j=1}^{n}(1-B_{i_j})
    \prod_{j=1}^{n} p(X_{i_j})
    \right| X\right]
    \right]
    \\
    =&
    \sum_{i_1,\ldots,i_n \in \mathcal{A}_\infty}
    \E\left[
    \1\{N \geq \max\{i_1,\ldots,i_n\}\}
    \tilde{h}_{i_1,\ldots,i_n}(X,Y)
    \prod_{j=1}^{n}(1-B_{i_j})
    \prod_{j=1}^{n} p(X_{i_j})
    \right]
    \\
    =&
    \E\left[
    \sum_{i_1,\ldots,i_n \in \mathcal{A}_{N}}
    h(X_{i_1},\ldots,X_{i_n}, Y \setminus \{ X_{i_1},\ldots, X_{i_n} \})
    \prod_{j=1}^{n}(1-B_{i_j})
    \prod_{j=1}^{n} p(X_{i_j})
    \right]
    ,
\end{align*}
where we have used the fact that 
$\E[\prod_{j=1}^{n} B_{i_j} | X ]
=   
\prod_{j=1}^{n} \E[B_{i_j}| X ] 
= 
\prod_{j=1}^{n} p(X_{i_j})$ by the conditional independence of the $B_{i_j}$'s, as well as the above-mentioned property of conditional expectations for $\prod_{j=1}^{n} p(X_{i_j})$ and the $\sigma$-algebra generated by $X$. Exploiting the representation $Y(\cdot)=X(\cdot)-Z(\cdot)= \sum_{i=1}^{N} (1-B_i) \delta_{X_i}(\cdot)$, we finally obtain that 
\begin{align*}
&
    \E\left[
    \sum_{i_1,\ldots,i_n \in \mathcal{A}_{N}}
    h(X_{i_1},\ldots,X_{i_n}, Y \setminus \{ X_{i_1},\ldots, X_{i_n} \})
    \prod_{j=1}^{n}(1-B_{i_j})
    \prod_{j=1}^{n} p(X_{i_j})
    \right]=
    \\
    =&
    \E \left[
    \mathop{\sum\nolimits\sp{\ne}}_{x_1,\ldots,x_n\in Y}
    h(x_1,\ldots,x_n, Y \setminus \{ x_1,\ldots, x_n \}) 
    \prod_{i=1}^{n} p(x_i) 
    \right]
    ,
\end{align*}
which proves \eqref{eq:thinning Z vs Y}.

Next, let 
$\psi_0(\breve{\x})=\{x:(x,m)\in
\breve{\x}\cap S\times \{0\} \}$, 
$\breve{\x}\in\breve{\X}$, and consider any non-negative or integrable
$h:S^n\times\X\to\R$. 
Applying the GNZ formula to the left-hand side of \eqref{eq:thinning Z vs Y} yields 
\begin{align*}
&\E\left[\mathop{\sum\nolimits\sp{\ne}}_{x_1,\ldots,x_n\in Z}
    h(x_1,\ldots,x_n, Y)
    \prod_{i=1}^n(1-p(x_i))
    \right] 
=
\\
=&
\E\left[\mathop{\sum\nolimits\sp{\ne}}_{(x_1,m_1),\ldots,(x_n,m_n)\in \breve{X}}
    h(x_1,\ldots,x_n,\psi_0(\breve{X}))
\prod_{i=1}^{n}m_i(1-p(x_i))
    \right] 
\\
=&
\int_{S^n}
    \E\left[
    \prod_{i=1}^n(1-p(u_i))
    h(u_1,\ldots,u_n; \psi_0(\breve{X}))
    \breve\lambda^{(n)}((u_1,1),\ldots,(u_n,1);\breve{X})
    \right]
    \de u_1\cdots\de u_n
\\
=&
\int_{S^n}
    \E\left[
    \prod_{i=1}^n(1-p(u_i))
    h(u_1,\ldots,u_n; Y)
    \breve\lambda^{(n)}((u_1,1),\ldots,(u_n,1);\breve{X})
    \right]
    \de u_1\cdots\de u_n,
\end{align*}
since the reference measure on the mark space is the counting measure on $\M=\{0,1\}$. 
On the other hand, applying the GNZ formula to the right hand side of \eqref{eq:thinning Z vs Y} yields 
\begin{align*}
&\E\left[\mathop{\sum\nolimits\sp{\ne}}_{x_1,\ldots,x_n\in Y}
    h(x_1,\ldots,x_n, Y\setminus \{x_1,\ldots,x_n\})
    \prod_{i=1}^n p(x_i)
    \right]
=
\\
=&
\int_{S^n}
\E\left[
    \prod_{i=1}^n p(u_i)
    h(u_1,\ldots,u_n; Y)
    \lambda_Y^{(n)}(u_1,\ldots,u_n;Y)
    \right]
    \de u_1\cdots\de u_n
    .
\end{align*}
The equality of these two expressions for arbitrary $h:S^n\times\X\to\R$ yields that for almost every $u_1,\ldots,u_n\in S^n$ and every non-negative or integrable $h^*:\X\to\R$ we have that
\[
\E\left[
    h^*(Y)\left(
    \breve\lambda^{(n)}((u_1,1),\ldots,(u_n,1);\breve{X})
    -
    \frac{\prod_{i=1}^n p(u_i)}{\prod_{i=1}^n (1-p(u_i))}\lambda_Y^{(n)}(u_1,\ldots,u_n;Y)
    \right)
    \right]
=0,
\]
which concludes the proof.

\end{proof}

\subsection{Proof of Lemma \ref{lemma:InnovationUnbiased}}
\begin{proof}[Proof of Lemma \ref{lemma:InnovationUnbiased}]
  This is an immediate consequence of the GNZ
  formula \eqref{eq:GNZ} and the Campbell formula \eqref{eq:Campbell}.
\end{proof}

\subsection{Proof of Theorem \ref{thm:InnovationsMean}}
\begin{proof}[Proof of Theorem \ref{thm:InnovationsMean}]
  For ease of notation, we sometimes write $\de u$ for $\de u_1 \cdots \de u_n$.
  
  \subsubsection*{When $h$ and $\xi$ are of the form
\eqref{e:GeneralEstimatorNoPP}}

When $h$ and $\xi$ are of the form
\eqref{e:GeneralEstimatorNoPP}, by the Campbell formula we have that
\begin{align*}
\E[\mathcal{I}_{\xi}^{h}(A;Z,Y)]
=&
\E\left[
\sum_{(x_1,\ldots,x_n)\in Z_{\neq}^n\cap A}
       h(x_1,\ldots,x_n)
\right]
-
\int_{A}
     h(u_1,\ldots,u_n)
     \xi(u_1,\ldots,u_n)
     \de u
\\
=&
\int_A
     h(u_1,\ldots,u_n)
     \left(
     \rho_Z^{(n)}(u_1,\ldots,u_n)
     \de u
     -
     \xi(u_1,\ldots,u_n)
     \right)
     \de u, \qquad A \subseteq S^n.
\end{align*}
Hence,
$\E[\mathcal{I}_{\xi}^{h}(A;Z,Y)]=0$ for any
(bounded) $A\subseteq S^n$ and function $h$ if and only if
\begin{equation*}
\xi(u_1,\ldots,u_n)
=
\rho_Z^{(n)}(u_1,\ldots,u_n),   
\end{equation*}
for $|\cdot|^n$-almost every $(u_1,\ldots,u_n)\in S^n$; see e.g.~\citet[Section 2.3.3]{MW04}. 

We further have that 
\begin{align*}
\Var(\mathcal{I}_{\xi}^{h}(A;Z,Y))
=& 
\Var\left(
\sum_{(x_1,\ldots,x_n)\in Z_{\neq}^n\cap A}
       h(x_1,\ldots,x_n)
\right)
\\
=&
\E\left[\left(
\sum_{(x_1,\ldots,x_n)\in Z_{\neq}^n\cap A}
       h(x_1,\ldots,x_n)
\right)^2\right]
\\
&
 - \left(
\int_{A}
     h(u_1,\ldots,u_n)
     \rho_Z^{(n)}(u_1,\ldots,u_n)
     \de u\right)^2,
\end{align*}
where, by \citet[Equation (B.3)]{poinas2019mixing},
\begin{align*}
&\E\left[\left(
\sum_{(x_1,\ldots,x_n)\in Z_{\neq}^n\cap A}
       h(x_1,\ldots,x_n)
\right)^2\right]
=
\\
=&
\E\left[\left(
\mathop{\sum\nolimits\sp{\ne}}_{x_1,\ldots,x_n\in Z}
\1\{(x_1,\ldots,x_n)\in A\}
h(x_1,\ldots,x_n)
\right)^2\right]
\\
=&
\E\left[\left(
\sum_{\y\subseteq Z}
n!
\1\{\#\y=n\}\1\{\y\in A\}h(\y)
\right)^2\right]
\\
=&
\sum_{j=0}^n
\frac{(n!)^2}{(2n-j)!}
\binom{n}{j}
\binom{2n-j}{n}
\int_{S^{2n-j}}
h(u_1,\ldots,u_n)
h(u_1,\ldots,u_j,u_{n+1},\ldots,u_{2n-j})
\times
\\
&\times
\1\{(u_1,\ldots,u_n)\in A\}
\1\{(u_1,\ldots,u_j,u_{n+1},\ldots,u_{2n-j})\in A\}
\times
\\
&\times
\rho_Z^{(2n-j)}(u_1,\ldots,u_{2n-j})
\de u_1\cdots\de u_{2n-j}.
\end{align*}


\subsubsection*{When $h$ and $\xi$ are of the form
\eqref{e:GeneralEstimator}}

We here let $h$ and $\xi$ be of the form
\eqref{e:GeneralEstimator} and start by defining
  \begin{align*}
    H_1(A)
    =&
       \sum_{(x_1,\ldots,x_n)\in Z_{\neq}^n\cap A}
       h(x_1,\ldots,x_n;Y\setminus\{x_1,\ldots,x_n\})
       ,
    \\
    H_2(A)=&
     \int_{A}
     h(u_1,\ldots,u_n;Y)
     \xi(u_1,\ldots,u_n;Y)
     \de u,
    \\
    \mu_1(A)
    =&\E\left[
       H_1(A)\right],
    \\
    \mu_2(A)=&
    \E\left[
    H_2(A)
    \right],
    \quad A\subseteq S^n,
\end{align*}
where we note that
\begin{align}
  \E[\mathcal{I}_{\xi}^{h}(A;Z,Y)]
  =&
     \mu_1(A)-\mu_2(A),
     \notag
  \\
  \E[\mathcal{I}_{\xi}^{h}(A;Z,Y)^2]
  =&
     \E[H_1(A)^2] + \E[H_2(A)^2]
     - 2\E[H_1(A) H_2(A)],
     \notag
  \\
  \Var(\mathcal{I}_{\xi}^{h}(A;Z,Y)) 
  =& \E[H_1(A)^2] + \E[H_2(A)^2]
     - 2\E[H_1(A) H_2(A)]
    -
     (\mu_1(A)-\mu_2(A))^2.
     \label{e:variance expectation innov I H1 H2}
\end{align}




Next, recall the associated marked point process $\breve{X}$ in \eqref{e:ZuYmpp}, with Papangelou conditional
intensity $\breve\lambda^{(n)}(\cdot)$. 
Given $\psi_0(\breve{\x})=\{x:(x,m)\in
\breve{\x}\cap S\times \{0\} \}$, $\breve{\x}\in\breve{\X}$, by the GNZ formula, 
\begin{align*}
  &\mu_1(A)=\\
    =&\E \left[
    \sum_{(x_1,\ldots,x_n)\in Z_{\neq}^n\cap A}
    h(x_1,\ldots,x_n;Y\setminus\{x_1,\ldots,x_n\})
    \right]
  \\
  =&
     \E\left[
     \sum_{((x_1,m_1),\ldots,(x_n,m_n))\in \breve{X}_{\neq}^n\cap (A\times\M^n)}
     \prod_{i=1}^n m_i
    h(x_1,\ldots,x_n;
    \psi_0(\breve{X} \setminus \{(x_1,m_1),\ldots,(x_n,m_n)\}))
    \right]\\
    =&
    \int_A\sum_{m_1,\ldots,m_n\in\{0,1\}}
    \E\left[
    \prod_{i=1}^n m_i
    h(u_1,\ldots,u_n;
    \psi_0(\breve{X}))
    \breve\lambda^{(n)}((u_1,m_1),\ldots,(u_n,m_n);\breve{X})
    \right]
    \de u
    \\
    =&
\int_A
    \E \left[
    h(u_1,\ldots,u_n; \psi_0(\breve{X}))
    \breve\lambda^{(n)}((u_1,1),\ldots,(u_n,1);\breve{X})
    \right]
\de u
\\
  =&
\int_A
     \E \left[
     h(u_1,\ldots,u_n;Y)
     \breve\lambda^{(n)}((u_1,1),\ldots,(u_n,1);\breve{X})
     \right]
     \de u
    ,
\end{align*}
since the reference measure on the mark space is given by the counting measure on the mark space $\M=\{0,1\}$. On the other hand, by the Fubini-Tonelli theorem, 
\begin{align*}
\mu_2(A)  =&
  \E \left[
  \int_{A}
  h(u_1,\ldots,u_n;Y)
  \xi(u_1,\ldots,u_n;Y)
\de u 
    \right]
\\
=&
\int_{A}
    \E \Big[
    h(u_1,\ldots,u_n;Y)
    \xi(u_1,\ldots,u_n;Y)
    \Big]
     \de u
     .
\end{align*}
%
Hence, $\E[\mathcal{I}_{\xi}^{h}(A;Z,Y)]=0$ for
any (bounded) $A\subseteq S^n$ if and only if
\begin{align*}
\E \left[ h(u_1,\ldots,u_n;Y)
    \left(
    \breve\lambda^{(n)}((u_1,1),\ldots,(u_n,1);\breve{X})
    -
    \xi(u_1,\ldots,u_n;Y)
    \right)
    \right]
    =
    0,
\end{align*}
for $|\cdot|^n$-almost every $(u_1,\ldots,u_n)\in S^n$; see e.g.~\citet[Section 2.3.3]{MW04}. 
Moreover, under the assumption that 
$\E[\breve\lambda^{(n)}((u_1,1),\ldots,(u_n,1);\breve{X})^2]<\infty$ and $\E[h(u_1,\ldots,u_n;Y)^2]<\infty$, 
$L_2$-projection
yields that 
\begin{equation*}
\xi(u_1,\ldots,u_n;Y) 
=
\E[\breve\lambda^{(n)}((u_1,1),\ldots,(u_n,1);\breve{X})|Y]
.
\end{equation*}
 



We next turn to the variance. Similarly to \citet[Equation (B.3)]{poinas2019mixing}, we find that 
\begin{align*}
  &\E[H_1(A)^2]
    =\\
    =&
    \E\Bigg[
    \sum_{(x_1,\ldots,x_n)\in Z_{\neq}^n}
    \sum_{(y_1,\ldots,y_n)\in Z_{\neq}^n}
    \1\{(x_1,\ldots,x_n),(y_1,\ldots,y_n)\in A\}
    \times
  \\
  &\times
    h(x_1,\ldots,x_n;Y\setminus\{x_1,\ldots,x_n\})
    h(y_1,\ldots,y_n;Y\setminus\{y_1,\ldots,y_n\})
    \Bigg]
  \\
  =&
     n!^2
     \E\left[
     \sum_{\x=\{x_1,\ldots,x_n\}\subseteq Z}
     \sum_{\y=\{y_1,\ldots,y_n\}\subseteq Z}
     \1\{\x,\y\in A\}
     h(\x;Y\setminus\x)
     h(\y;Y\setminus\y)
     \right]
  \\
  =&
     n!^2
     \sum_{j=0}^n
     \E\Bigg[
     \sum_{\x=\{x_1,\ldots,x_n\}\subseteq Z}
     \sum_{\y=\{y_1,\ldots,y_n\}\subseteq Z}
     \1\{\#(\x\cap\y)=j\}
     \1\{\x,\y\in A\}
     h(\x;Y\setminus\x)
     h(\y;Y\setminus\y)
     \Bigg]
     ,
\end{align*}
where the factor $n!^2$ comes from the fact that when we go from
$n$-subsets to $n$-tuples we count the same thing $n!$ times; we can
rearrange $(x_1,\ldots,x_n)$ in $n!$ different ways.  In the last sum,
assuming that $\#(\x\cap\y)=j$, i.e.~that $\x$ and $\y$ have $j$
elements $x_i=y_{i'}\in Z$ in common, there are $\binom{n}{j}$ ways in
which the elements in $\x\cap\y$ can be chosen from $\x$ and $\y$.
The remaining $2n-j$ elements now need to be assigned to
$\x\setminus(\x\cap\y)$ and $\y\setminus(\x\cap\y)$.  There are
$\binom{2n-j}{n-j}=\binom{2n-j}{n}$ ways to assign elements of
$(\x\cap\y)^c$ to $\x\setminus(\x\cap\y)$ so that $\#\x=n$; the
remaining elements will automatically be assigned to
$\y\setminus(\x\cap\y)$. 
In the following, we let $z_1,\ldots,z_j$ denote the elements in $\x\cap \y$, $z_{j+1},\ldots,z_n$ the elements only in $\x$, and 
$z_{n+1},\ldots,z_{2n-j}$  the ones only in $\y$.
Consequently,
\begin{align*}
  &\E[H_1(A)^2]
    =
  \\
  =&
     n!^2
     \sum_{j=0}^n
     \binom{n}{j}
     \binom{2n-j}{n}
     \E\Bigg[
     \sum_{\{z_1,\ldots,z_{2n-j}\}\subseteq Z}
     \1\{\{z_1,\ldots,z_n\}\in A\}
     \times
  \\
  &\times
    \1\{\{z_1,\ldots,z_j,z_{n+1},\ldots, z_{2n-j}\}\in A\}
    h(z_1,\ldots,z_n;Y\setminus\{z_1,\ldots,z_n\}) 
    \times
  \\
  &\times
    h(z_1,\ldots,z_j,z_{n+1},\ldots, z_{2n-j};Y\setminus\{z_1,\ldots,z_j,z_{n+1},\ldots, z_{2n-j}\})
    \Bigg]
  \\
  =&
     \sum_{j=0}^n
     \binom{n}{j}
     \binom{2n-j}{n}
     \frac{n!^2}{(2n-j)!}
     \E\Bigg[
     \sum_{(z_1,\ldots,z_{2n-j})\subseteq Z_{\neq}^{2n-j}}
     \1\{(z_1,\ldots,z_n)\in A\}
     \times
  \\
  &\times
    \1\{(z_1,\ldots,z_j,z_{n+1},\ldots, z_{2n-j})\in A\}
    h(z_1,\ldots,z_n;Y\setminus\{z_1,\ldots,z_n\}) 
    \times
  \\
  &\times
    h(z_1,\ldots,z_j,z_{n+1},\ldots, z_{2n-j};Y\setminus\{z_1,\ldots,z_j,z_{n+1},\ldots, z_{2n-j}\})
    \Bigg],
\end{align*}
where we note that
\[
\frac{(n!)^2}{(2n-j)!}
\binom{n}{j}
\binom{2n-j}{n}
=
\binom{n}{j}
\frac{(n!)^2}{(2n-j)!}
\frac{(2n-j)!}{n!(n-j)!}
=
j!
\binom{n}{j}^2
.
\]
By applying the GNZ formula
to each term in the sum in the last equation, it follows that
\begin{align}
\label{e:variance innov H1 square}
 &\E[H_1(A)^2]=
 \\
  =&
     \sum_{j=0}^n
     j! \binom{n}{j}^2
     \int_{S^{2n-j}}
     \1\{(u_1,\ldots,u_n),(u_1,\ldots,u_j,u_{n+1},\ldots, u_{2n-j})\in A\}
     \notag
  \\
  &\times
  \E\Big[
     h(u_1,\ldots,u_n;Y)
     h(u_1,\ldots,u_j,u_{n+1},\ldots, u_{2n-j};Y)
    \breve\lambda^{(2n-j)}((u_1,1),\ldots, (u_{2n-j},1);\breve{X})
    \Big]
      \notag
    \\
    &\times \de u_1\cdots \de u_{2n-j}
    .
    \notag
\end{align}
We further have that 
\begin{align}
  &\E[H_2(A)^2] = 
  \nonumber
  \\
  =&
     \int_{A}\int_{A}
     \E[
     h(u_1,\ldots,u_n;Y)
     h(v_1,\ldots,v_n;Y)
     \xi(u_1,\ldots,u_n;Y)
     \xi(v_1,\ldots,v_n;Y)
     ]
     \de u
     \de v
     \label{e:variance innov H2 square}
\end{align}
and 
\begin{align*}
  &\E[H_1(A)H_2(A)] =
  \\
    =&
     \E\Bigg[
     \sum_{(x_1,\ldots,x_n)\in Z_{\neq}^n\cap A}
     h(x_1,\ldots,x_n;Y\setminus\{x_1,\ldots,x_n\})
     \times
  \\
  &\times
    \int_{A}
    h(v_1,\ldots,v_n;Y)
    \xi(v_1,\ldots,v_n;Y)
    \de v
    \Bigg]
  \\
  =&
     \E\Bigg[
     \sum_{(x_1,\ldots,x_n)\in Z_{\neq}^n\cap A}
     h(x_1,\ldots,x_n;Y\setminus\{x_1,\ldots,x_n\})
     \times
  \\
  &\times
    \Bigg(
    \int_{A}
    h(v_1,\ldots,v_n;(Y\setminus\{x_1,\ldots,x_n\})\cup\{x_1,\ldots,x_n\})
    \times
  \\
  &\times
    \xi(v_1,\ldots,v_n;(Y\setminus\{x_1,\ldots,x_n\})\cup\{x_1,\ldots,x_n\})
    \de v
    \Bigg)
    \Bigg]
  \\
  =&
     \E\left[
     \sum_{(x_1,\ldots,x_n)\in Z_{\neq}^n\cap A}
     \widetilde h(x_1,\ldots,x_n;Y\setminus\{x_1,\ldots,x_n\})
     \right].
\end{align*}
where 
\begin{align*}
    \widetilde h(x_1,\ldots,x_n;Y\setminus\{x_1,\ldots,x_n\}) 
    =&
    h(x_1,\ldots,x_n;Y\setminus\{x_1,\ldots,x_n\}) 
    \\
    &\times
    \int_{A}
    h(v;
    (Y\setminus\{x_1,\ldots,x_n\})\cup\{x_1,\ldots,x_n\})
    \\
    &\times
    \xi(v;(Y\setminus\{x_1,\ldots,x_n\})\cup\{x_1,\ldots,x_n\})
    \de v.
\end{align*}
Hence, 
\begin{align}
  &\E[H_1(A)
    H_2(A)]
    = 
    \notag
    \\
  =&
    \int_A
    \E\Bigg[
    \widetilde h(u_1,\ldots,u_n;Y)
    \breve\lambda^{(n)}((u_1,1),\ldots,(u_n,1);\breve{X})
    \Bigg]
    \de u
    \notag
  \\
  =&
     \int_A\int_A
     \E\Bigg[
     h(u_1,\ldots,u_n;Y)
     h(v_1,\ldots,v_n;Y\cup\{u_1,\ldots,u_n\})
     \times
     \notag
  \\
  &\times
    \xi(v_1,\ldots,v_n;Y\cup\{u_1,\ldots,u_n\})
    \breve\lambda^{(n)}((u_1,1),\ldots,(u_n,1);\breve{X})
    \Bigg]
    \de u
    \de v
    \label{e:variance innov H1H2}
\end{align}
and, consequently, by combining~\eqref{e:variance expectation innov I H1 H2}
with~\eqref{e:variance innov H1 square},~\eqref{e:variance innov H2 square} and
\eqref{e:variance innov H1H2},
the variance is given by
\begin{align*}
  &\Var(\mathcal{I}_{\xi}^{h}(A;Z,Y)) 
    =
    \E[\mathcal{I}_{\xi}^{h}(A;Z,Y)^2]
    -
    \E[\mathcal{I}_{\xi}^{h}(A;Z,Y)]^2
  \\
  =& 
  \sum_{j=0}^n
     j! \binom{n}{j}^2
     \int_{S^{2n-j}}
     \1\{(u_1,\ldots,u_n),(u_1,\ldots,u_j,u_{n+1},\ldots, u_{2n-j})\in A\}
     \E\Big[
     h(u_1,\ldots,u_n;Y)
  \\
  &\times
     h(u_1,\ldots,u_j,u_{n+1},\ldots, u_{2n-j};Y)
    \breve\lambda^{(2n-j)}((u_1,1),\ldots, (u_{2n-j},1);\breve{X})
    \Big]
    \de u_1\cdots \de u_{2n-j}
  \\
  &+ 
     \int_{A}\int_{A}
     \E[
     h(u_1,\ldots,u_n;Y)
     h(v_1,\ldots,v_n;Y)
     \times
     \notag
  \\
   &\times
     \xi^n(u_1,\ldots,u_n;Y)
     \xi^n(v_1,\ldots,v_n;Y)
     ]
     \de u
     \de v
    \\
    &- 2\Bigg( 
    \int_A\int_A
     \E\Big[
     h(u_1,\ldots,u_n;Y)
     h(v_1,\ldots,v_n;Y\cup\{u_1,\ldots,u_n\})
     \times
     \notag
  \\
  &\times
    \xi(v_1,\ldots,v_n;Y\cup\{u_1,\ldots,u_n\})
    \breve\lambda^{(n)}((u_1,1),\ldots,(u_n,1);\breve{X})
    \Big]
    \de u
    \de v
    \Bigg)  
    \\
&- 
\Bigg( 
\E\left[h(u_1,\ldots,u_n;Y)
    \left(
    \breve\lambda^{(n)}((u_1,1),\ldots,(u_n,1);\breve{X})
    -
    \xi(u_1,\ldots,u_n;Y)
    \right)
    \right]
\Bigg)^2
\end{align*}
when $h$ and $\xi$ are of the form
\eqref{e:GeneralEstimator}. 

\end{proof}

\subsection{Proof of Theorem \ref{thm:ConstantIntensity}}
\begin{proof}[Proof of Theorem \ref{thm:ConstantIntensity}]
By
\eqref{e:EstFunGeneralMedian} 
and~\eqref{e:InnovationParametricIntensity}, 
we have that 
\begin{equation*}
    \mathcal{L}_1(\theta) 
    =
    \frac{1}{k}
    \sum_{i=1}^k
    |\mathcal{I}_{\xi_{\theta}^1}^{h}(W;\x_i^T)|
    =
    \frac{1}{k}
    \sum_{i=1}^k
    \left|\sum_{x\in\x_i^T\cap W}
    h(x)
    -
    \theta (1-p)
    \int_{W}
    h(u)
    \de u\right| \1\{\#\x_i^T\geq1\},
\end{equation*}
which has the same minimum with respect to $\theta$ as 
\begin{equation*}
   \frac{k \mathcal{L}_1(\theta) }{(1-p)|\int_{W}h(u)\de u|}
    =
    \sum_{i=1}^k
    \left|
    \frac{\sum_{x\in\x_i^T\cap W}
     h(x)}{(1-p)
     \int_{W}
     h(u)
     \de u}
     -
     \theta
     \right|
     \1\{\#\x_i^T\geq1\}.
\end{equation*}
The derivative of the last function is defined for almost all $\theta$ by
\begin{equation*}
    \sum_{i=1}^k 
    \mathrm{sgn}\left(
    \frac{\sum_{x\in\x_i^T\cap W}
     h(x)}{(1-p)
     \int_{W}
     h(u)
     \de u} -\theta
     \right)
    \1\{\#\x_i^T\geq1\},
\end{equation*}
where for all $x\in \R$, $\mathrm{sgn}(x)=\1\{x\geq 0\}
+ \1\{x< 0\}$.
This sum is  null only if on the set of all 
$i\in \mathcal{T}_k$, half of the 
$(\sum_{x\in\x_i^T\cap W} h(x)) / ((1-p)
\int_{W}
h(u)
\de u)$
are greater than $\theta$ and half are less than $\theta$. 
Therefore we have that a minimiser is given by $\widehat{\theta}_1(\{(\x_i^T,\x_i^V)\}_{i=1}^k,p,W,h) 
=
\med
\{\widehat\theta_i^h : i \in \mathcal{T}_k\}$.

Elementary calculus shows that the derivatives of $\mathcal{L}_2$ and 
$\mathcal{L}_3$ with respect to $\theta$ are null if $\theta$ satisfies
\begin{equation*}
    \sum_{i=1}^k
    \left(\sum_{x\in\x_i^T\cap W}
    h(x)
    -
    \theta (1-p)
    \int_{W}
    h(u)
    \de u\right) \1\{\#\x_i^T\geq1\}
    =
    0
\end{equation*}
which is equivalent to finding $\theta$ such that
\begin{equation*}
    \sum_{i=1}^k
    \left(
    \frac{\sum_{x\in\x_i^T\cap W}
    h(x)}{(1-p)
    \int_{W}
    h(u)
    \de u}
    -
    \theta \right) 
    \1\{\#\x_i^T\geq1\}
    =
    0.
\end{equation*}
Therefore, for $j=2,3$ we have that 
\begin{align*}
    \widehat{\theta}_j(\{(\x_i^T,\x_i^V)\}_{i=1}^k,p,W,h) 
    =& 
    \frac{1}{\#\mathcal{T}_k}
    \sum_{i\in \mathcal{T}_k} 
    \frac{\sum_{x\in\x_i^T\cap W}
    h(x)}{(1-p)
    \int_{W}
    h(u)
    \de u}
    \\
    =&
    \frac{\sum_{x\in\x\cap W}
    h(x)\frac{1}{\#\mathcal{T}_k}
    \sum_{i\in \mathcal{T}_k}\1\{x\in\x_i^T\}}{(1-p)
    \int_{W}
    h(u)
    \de u}.
\end{align*}

Further, when $j=2,3$, for any point configuration $\x$ we have that
\begin{align*}
    \E[\widehat{\theta}_j(\{(\x_i^T,\x_i^V)\}_{i=1}^k,p,W,h)] 
    =&
    \E\left[ 
    \frac{1}{\#\mathcal{T}_k}
    \sum_{i\in \mathcal{T}_k} 
    \left.
    \E\left[  
    \frac{\sum_{x\in\x\cap W}
    h(x)\1\{x\in\x_i^T\}}{(1-p)
    \int_{W}
    h(u)
    \de u}  \right| \mathcal{T}_k
    \right]
    \right]
    \\
    =&
    \frac{\sum_{x\in\x\cap W}
    h(x)}{(1-p)
    \int_{W}
    h(u)
    \de u}
    \E\left[ 
    \frac{1}{\#\mathcal{T}_k}
    \sum_{i\in \mathcal{T}_k} 
    \E[ 
    \1\{x\in\x_i^T\}
    | \mathcal{T}_k
    ]
    \right]
    \\
    =&
    \frac{\sum_{x\in\x\cap W}
    h(x)}{(1-p)
    \int_{W}
    h(u)
    \de u}
    \E[\E[\1\{x\in\x_1^T\}|\mathcal{T}_k]]
    =
    \frac{\sum_{x\in\x\cap W}
    h(x)}{\int_{W}
    h(u)
    \de u}
\end{align*}
by the law of total expectation, 
since the validation/training set assignments are independent and identically distributed with $\P(x\in\x_i^T)=1-p$, $i=1,\ldots,k$, $x\in\x$. It is straightforward to see that the same holds true if $\mathcal{T}_k=\{1,\ldots,k\}$.
\end{proof}

\subsection{Proof of Lemma~\ref{lemma:constant intensity majoration variance}}

We here want to show that 
\begin{equation*}
  \Var\left(\frac{1}{\#\mathcal{T}_k}
\sum_{i \in \mathcal{T}_k} 
\widehat\theta_i^h\right)
\geq
  \Var\left( 
      \frac{1}{k}
      \sum_{i=1}^k 
      \widehat\theta_i^h
      \right)
     =\frac{p}{k(1-p)} \frac{\sum_{x\in \x\cap W}h(x)^2}
  {(\int_{W} h(u) \de u)^2}.
  \end{equation*}
Staring with the right hand side, 
we have that
\begin{align*}
 \Var\left(\frac{1}{k}
      \sum_{i=1}^k 
      \widehat\theta_i^h\right)
      =&
      \Var\left(
      \frac{1}{k}
      \sum_{i=1}^k 
      \frac{
  \sum_{x\in\x\cap W} h(x)  \1\{x\in\x_i^T\} }{ (1-p) \int_{W} h(u) \de u}
      \right)
    =
      \frac{\Var\left(
  \sum_{x\in\x\cap W} b_x h(x) 
      \right)}{k^2(1-p)^2 (\int_{W} h(u) \de u)^2},
\end{align*}
where, for any $x\in X$, 
$b_x = \sum_{i=1}^k
\1\{x \in \x_i^T \cap W\}$ follows a binomial distribution with parameters $k$ and $(1-p)$, with
$\E[b_x]=k(1-p)$ and $\E[b_x^2]=k p(1-p) + (k(1-p))^2=k(1-p)(p+k(1-p))$. Since the events $x\in \x_i^T$ and $y\in \x_{i'}^T$ are independent for all $i,i'=1,\ldots,k$ and  $x\neq y$, $b_x$ and $b_y$ are independent when $x\neq y$, whereby
\begin{align*}
 \Var\left(\frac{1}{k}
      \sum_{i=1}^k 
      \widehat\theta_i^h\right)
      =&
      \frac{\sum_{x\in\x\cap W} \Var(b_x)h(x)^2}{k^2(1-p)^2 (\int_{W} h(u) \de u)^2}
      =  \frac{p}{k(1-p)} \frac{\sum_{x\in \x\cap W}h(x)^2}
  {(\int_{W} h(u) \de u)^2}.
\end{align*}



We next turn to the inequality. Since $\sum_{x\in\emptyset}
h(x)=0$, we a.s.~have that 
\begin{align*}
    \widehat{\theta}_j(\{(\x_i^T,\x_i^V)\}_{i=1}^k,p,W,h) 
    =&
    \frac{1}{\#\mathcal{T}_k}
\sum_{i \in \mathcal{T}_k} 
\widehat\theta_i^h
=
    \frac{1}{\#\mathcal{T}_k}
    \sum_{i=1}^{k} 
    \frac{\sum_{x\in\x_i^T\cap W}
    h(x)}{(1-p)
    \int_{W}
    h(u)
    \de u}
    \\
    =&
    \frac{k}{\#\mathcal{T}_k}
    \frac{1}{k}
    \sum_{i=1}^{k} 
    \frac{\sum_{x\in\x_i^T\cap W}
    h(x)}{(1-p)
    \int_{W}
    h(u)
    \de u}
    \\
    =&
    \frac{k}{\#\mathcal{T}_k}
    \frac{\sum_{x\in\x\cap W}
    h(x)\frac{1}{k}
    \sum_{i=1}^k\1\{x\in\x_i^T\}}{(1-p)
    \int_{W}
    h(u)
    \de u},
\end{align*}
whereby 
\begin{align*}
\Var(\widehat{\theta}_j(\{(\x_i^T,\x_i^V)\}_{i=1}^k,p,W,h))
    &=
    \frac{\Var\left( 
    \frac{k}{\#\mathcal{T}_k}
    \sum_{x\in\x\cap W}
    \frac{h(x)}{k}
    \sum_{i=1}^k\1\{x\in\x_i^T\}
    \right)}{(1-p)^2
    (\int_{W}
    h(u)
    \de u)^2}.
\end{align*}
The numerator is given by
\begin{align}
    \Var\left( 
    \frac{k}{\#\mathcal{T}_k}
    \sum_{x\in\x\cap W}
    \frac{h(x)}{k}
    \sum_{i=1}^k\1\{x\in\x_i^T\}
    \right)
    =&
    \E\left[\left(\frac{k}{\#\mathcal{T}_k}
    \sum_{x\in\x\cap W}
    h(x)\frac{1}{k}
    \sum_{i=1}^k\1\{x\in\x_i^T\}\right)^2\right]
    \notag \\
    &- 
    \E\left[
    \frac{k}{\#\mathcal{T}_k}
    \sum_{x\in\x\cap W}
    \frac{h(x)}{k}
    \sum_{i=1}^k\1\{x\in\x_i^T\}\right]^2,
    \label{eq:proof constant intensity inequality var 1}
\end{align}
where, by Theorem~\ref{thm:ConstantIntensity},
\begin{equation*}
    \E\left[
    \frac{k}{\#\mathcal{T}_k}
    \sum_{x\in\x\cap W}
    \frac{h(x)}{k}
    \sum_{i=1}^k\1\{x\in\x_i^T\}\right]
    = (1-p)\sum_{x\in\x\cap W}  h(x).
\end{equation*}
We further have that 
\begin{align}
    \Var\left( 
    \sum_{x\in\x\cap W}
    \frac{h(x)}{k}
    \sum_{i=1}^k\1\{x\in\x_i^T\}
    \right)
    =& 
    \E \left[\left(
    \sum_{x\in\x\cap W}
    \frac{h(x)}{k}
    \sum_{i=1}^k\1\{x\in\x_i^T\}\right)^2 \right]
    \notag \\
    &- 
    \E\left[
    \sum_{x\in\x\cap W}
    \frac{h(x)}{k}
    \sum_{i=1}^k\1\{x\in\x_i^T\}\right]^2
    \label{eq:proof constant intensity inequality var 2}
\end{align}
and
\begin{equation*}
   \E\left[
    \sum_{x\in\x\cap W}
    \frac{h(x)}{k}
    \sum_{i=1}^k\1\{x\in\x_i^T\}\right]
    =
    \sum_{x\in\x\cap W}
    \frac{h(x)}{k}
    \sum_{i=1}^k\E[\1\{x\in\x_i^T\}]
    = (1-p)\sum_{x\in\x\cap W}  h(x).
\end{equation*}
Since $k\geq \#\mathcal{T}_k$, from the two expectation equalities and~\eqref{eq:proof constant intensity inequality var 1}-\eqref{eq:proof constant intensity inequality var 2} it follows that 
\begin{equation*}
     \Var\left( 
    \frac{k}{\#\mathcal{T}_k}
    \sum_{x\in\x\cap W}
    \frac{h(x)}{k}
    \sum_{i=1}^k\1\{x\in\x_i^T\}
    \right)
    \geq 
    \Var\left( 
    \sum_{x\in\x\cap W}
    \frac{h(x)}{k}
    \sum_{i=1}^k\1\{x\in\x_i^T\}
    \right),
\end{equation*}
which concludes the proof.

\subsection{Proof of Lemma~\ref{lemma:ConstantIntensityX}}

\begin{proof}[Proof of Lemma~\ref{lemma:ConstantIntensityX}]

By Campbell's theorem, we have
\begin{align*}
    \E[\widehat{\theta}_j(\{(X_i^T,X_i^V)\}_{i=1}^k,p,W,h)] 
    =&
    \E\left[ 
    \frac{1}{\#\mathcal{T}_k}
    \sum_{i\in \mathcal{T}_k} 
    \left.
    \E\left[  
    \frac{\sum_{x\in X\cap W}
    h(x)\1\{x\in X_i^T\}}{(1-p)
    \int_{W}
    h(u)
    \de u}  \right| \mathcal{T}_k
    \right]
    \right]
    \\
    =&
   \E\left[ 
    \frac{1}{\#\mathcal{T}_k}
    \sum_{i\in \mathcal{T}_k} 
    \frac{    (1-p)\theta_0 
    \int_{W}
    h(u)
    \de u}{(1-p)
    \int_{W}
    h(u)
    \de u} 
    \right]
    =
    \theta_0
\end{align*}
and
\begin{align*}
     \E\left[ 
    \frac{1}{k}
    \sum_{i=1}^k 
    \frac{\sum_{x\in X\cap W}
    h(x)\1\{x\in X_i^T\}}{(1-p)
    \int_{W}
    h(u)
    \de u}  
    \right]
    &=
     \E\left[ 
    \frac{1}{k}
    \sum_{i=1}^k 
    \left.
    \E
    \left[  
    \frac{\sum_{x\in X\cap W}
    h(x)\1\{x\in X_i^T\}}{(1-p)
    \int_{W}
    h(u)
    \de u}  \right| \mathcal{T}_k
    \right]
    \right]
    = \theta_0.
\end{align*}
Hence, by the last two equations, we can follow the same calculus as in the proof of Lemma~\ref{lemma:constant intensity majoration variance} to prove that 
\begin{equation*}
    \Var\widehat\theta_j((X_i^T,X_i^V),p,W,h)
    \geq
    \Var\left( 
    \frac{1}{k}
    \sum_{i=1}^k 
    \widehat\theta((X_i^T,X_i^V),p,W,h)
    \right).
\end{equation*}

For all $x\in X$, 
let $b_x = \sum_{i\in \mathcal{T}_k}
\1\{x \in X_i^T \cap W\}$ which, conditionally on $\mathcal{T}_k=k$,
follows a binomial distribution with parameters $k$ and $(1-p)$ so that
$\E[b_x]=k(1-p)$ and $\E[b_x^2]=k p(1-p) + (k(1-p))^2=k(1-p)(p+k(1-p))$. Thus,
\begin{align*}
  & \Var\left( 
    \frac{1}{k}
    \sum_{i=1}^k 
    \widehat\theta((X_i^T,X_i^V),p,W,h)
    \right)=
  \\
  =&
 \frac{\Var
  \left(
  \sum_{x\in X\cap W} b_x h(x) \right)}
  {(k (1-p)\int_{W} h(u) \de u)^2}
  \\
  =&
  \frac{
     \sum_{x\in X\cap W}  h(x)^2\E[
     b_x^2]
     +
     \sum_{x,y\in X\cap W}^{\neq} h(x)h(y)
     \E[b_x b_y ]
     -
     \E[
     \sum_{x\in X\cap W} b_x h(x)]^2
     }
     {(k (1-p) \int_{W} h(u) \de u)^2}\\
     =&
  \frac{
    \frac{p+k(1-p)}{k(1-p)}
    \theta_0\int_W h(u)^2\de u
     +
     \int_W\int_W h(u_1)h(u_2)\rho^{(2)}(u_1,u_2)\de u_1\de u_2
     -
     (\theta_0\int_W h(u)\de u)^2
     }
     {(\int_{W} h(u) \de u)^2}
     \\
     =&
     \left(\frac{p}{(1-p)k}
     +
     1 \right) \theta_0
     \frac{ \int_W h(u)^2\de u }
     {(\int_{W} h(u) \de u)^2}
     +
     \theta_0^2
     \left( \frac{
     \int_W\int_W h(u_1)h(u_2)g_X^{(2)}(u_1,u_2)\de u_1\de u_2
     }
     {( \int_{W} h(u) \de u)^2} - 1\right).
\end{align*}
\end{proof}

\end{document}